\documentclass[showpacs,amsmath,amssymb,twocolumn,pra,superscriptaddress]{revtex4-1}
\usepackage[dvips]{graphicx}
\usepackage{amsmath,amssymb,amsthm,mathrsfs,amsfonts,dsfont}
\usepackage{subfigure, epsfig}
\usepackage{braket}
\usepackage{bm}
\usepackage{bbm}
\usepackage{enumerate}
\usepackage{color}
\usepackage{comment}
\usepackage{diagbox}
\usepackage{youngtab}
\usepackage[colorlinks = true]{hyperref}

\graphicspath{{./figure/}}

\newtheorem{prop}{Proposition}

\newtheorem{definition}{Definition}

\newcommand{\renyi}{R$\mathrm{\acute{e}}$nyi }
\newcommand{\mc}{\mathcal}
\newcommand{\mb}{\mathbf}
\newcommand{\mbb}{\mathbb}
\newcommand{\tr}{\mathrm{Tr}}
\newcommand{\ketbra}[2]{\vert #1 \rangle \langle #2 \vert}
\newcommand{\Hr}{\mathrm{Haar}}
\newcommand{\red}[1]{{\color{red} #1}}
\newcommand{\blue}[1]{{\color{blue} #1}}
\newcommand{\sub}[2]{ \mathop{#1}\limits_{#2}  }
\newcommand{\comments}[1]{}
\begin{document}
\title{Single-copies estimation of entanglement negativity}

\begin{abstract}
Entanglement plays a central role in quantum information processing, indicating the non-local correlation of quantum matters. However, few effective ways are known to detect the amount of entanglement of an unknown quantum state. In this work, we propose a scheme to estimate the entanglement negativity of any bi-partition of a composite system. The proposed scheme is based on the random unitary evolution and local measurements on the single-copy quantum states, which is more practical compared with former methods based on collective measurements on many copies of the identical state.
Meanwhile, we generalize the scheme to quantify the total multi-partite correlation. We demonstrate the efficiency of the scheme with theoretical statistical analysis and numerical simulations.
The proposed scheme is quite suitable for state-of-the-art quantum platforms, which can serve as not only a useful benchmarking tool to advance the quantum technology, but also a probe to study fundamental quantum physics, such as the entanglement dynamics.
\end{abstract}
 \date{\today}
\author{You Zhou}
\email{you\_zhou@g.harvard.edu}
\affiliation{Department of Physics, Harvard University, Cambridge, Massachusetts 02138, USA}
\author{Pei Zeng}
\email{qubitpei@gmail.com}
\affiliation{Center for Quantum Information, Institute for Interdisciplinary Information Sciences, Tsinghua University, Beijing 100084, China}
\author{Zhenhuan Liu}
\email{qubithuan@gmail.com}
\affiliation{School of Physics, Peking University, Beijing 100871, China}
\affiliation{Center for Quantum Information, Institute for Interdisciplinary Information Sciences, Tsinghua University, Beijing 100084, China}

\maketitle

\section{Introduction} \label{Sec:Intro}
Entanglement, the very nature of the quantum correlation \cite{Horodecki2009entanglement}, is instrumental in the study of fundamental quantum mechanics and the quantum information processing tasks \cite{Nielsen2011Quantum}, such as quantum communication \cite{Bennett84cryptography,Ekert1991cryptography,Bennett1993Teleporting}, quantum metrology \cite{Wineland1992squeezing,Giovannetti2006Metrology}, and quantum computing and simulation \cite{Nielsen2011Quantum,Lloyd1996Simulators}. Recently, there are also marriages between the concept of entanglement and other disciplines, such as condensed matter and high energy physics, where entanglement is regarded as the signature of quantum orders and quantum phase transition \cite{Amico2008Entanglement,Bei2019Meets}, as well as the clue of quantum space and time \cite{Qi2018gravity}.

For a small-scale quantum system, quantum tomography \cite{Vogel1989Determination,Paris2004esimation} is a common tool to extract the complete information of the state and hence quantify the entanglement. As the system size increases, traditional tomography becomes impractical and alternative tomographic methods arise, which utilize the prior-knowledge of prepared states, such as the low-rank property \cite{Gross2010Tomography,Flammia2012compressed}, area-law entanglement entropy \cite{Cramer2010Efficient,Baumgratz2013Scalable,Lanyon2017Efficient}, or permutation symmetry  \cite{Toth2010Permutation,Tobias2012Permutation,You2019Sym}.
Nevertheless, these ansatzs may not include the state of interest. For example, the states with volume-law entanglement entropy which arise in the quench dynamics or the eigenstates of chaotic Hamiltonians \cite{Luca2016chaos}.
Moreover, even one can extract the full information, the computing of the related entanglement measure is also a daunting task. 
Entanglement witnesses \cite{GUHNE2009detection,Friis2019Reviews} and the associated quantification protocols \cite{Audenaert2006verification,eisert2007quantitative,Guhne2007Estimating,Marcus2016Quantifying} are also widely used methods. However, they also depend heavily on the prior-knowledge thus could lead to an unsuccessful detection \cite{Zhu2010Minimal,Dai2014Witness,You2020coherent}. Such difficulties motivate us to construct a direct entanglement estimation scheme without any prior-knowledge of quantum states, which also needs not the tomographic efforts.

Entanglement entropy between subsystems is a standard measure for pure states \cite{Horodecki2009entanglement}. There are a few of theoretical proposals \cite{Daley2012Measuring,Abanin2012Measuring,van2012Measuring,Elben2018Random} and experiment realizations to measure the \renyi-2 entropy in bosonic and spin systems \cite{Islam2015Measuring,Kaufmanen2016tanglement,Brydges2019Probing}, which enable us to observe many-body physics through the lens of entanglement. Nevertheless, quantum states are in general mixed, especially in the noisy or open quantum systems. Note that the subsystem may own large entropy even with only classical correlations. The quantification of entanglement for mixed states is a more challenging task \cite{Horodecki2009entanglement}.  
Among the various entanglement measures \cite{Plenio2007Introduction}, the (logarithmic) negativity \cite{Vidal2002measure,Plenio2005Negativity} is a reliable one due to its clear operational meanings in quantum information processing, such as a upper bound of entanglement distillation \cite{Vidal2002measure}, and wide applications in many-body physics \cite{Calabrese2012Negativity,Lanyon2017Efficient}.

Recently, Gray~et.~al.~show that negativity can be faithfully extracted from the first few moments of the partially-transposed density matrix $\rho_{AB}^{T_B}$ \cite{Bose2018Machine}. Despite its accurate prediction shown in the numerical simulations, the scheme there requires a parallel preparation of at least three identical copies together with a joint quantum measurement. This is quite challenging for current quantum devices, especially for the systems with high spatial dimension. 

In this work, we propose a scheme to extract the 3-order negativity-moment, $\tr\left[ (\rho_{AB}^{T_B})^3\right]$ with a single-copy state. Our scheme is based on randomized measurements, which lies in the random unitary evolution followed by standard measurements \cite{van2012Measuring,Elben2018Random,Brydges2019Probing,Vermersch2019Probing,Elben2020topological}.
Even though partial transpose is itself not a physical operation, we can realize this by utilizing permutation operations effectively generated by random unitaries. By creating virtual copies with delicate data post-processing, our scheme can be conducted with quantum operations on single-copies of quantum states, thus dramatically ease the experiment setups compared with the previous proposal \cite{Bose2018Machine}. As a byproduct, the scheme can also be used to measure the total correlation between any two subsystems, which together with negativity can quantify both classical and quantum correlations in composite systems. 

\section{Logarithmic Negativity}
Logarithmic negativity is an entanglement measure defined as,
\begin{equation}\label{Eq:LogNeg}
    \begin{aligned}
        E_N(\rho_{AB})=\log |\rho_{AB}^{T_B}|= \log \sum_k|\lambda_k|,
    \end{aligned}
\end{equation}
where $\lambda_k$ is the eigenvalues of the partially-transposed matrix $\rho_{AB}^{T_B}$, and it is clear that $\rho_{AB}^{T_A}$ share the same eigenvalues with $\rho_{AB}^{T_B}$. Partial transpose is not completely positive, thus not a physical operation. For some entangled states,
$\rho_{AB}^{T_B}$ can own negative eigenvalues leading to $E_N(\rho_{AB})>0$. Log-negativity owns various operational interpretations, such as an upper bound to entanglement distillation rate, a bound on teleportation capacity \cite{Vidal2002measure}, and the entanglement cost under a larger operation set \cite{Eisert2003Cost}.

The value of the negativity $E_N(\rho_{AB})$ depends on the spectrum of $\rho_{AB}^{T_B}$, and thus is a complicated non-linear function of the state. Generally one can utilize tomography to reconstruct and then calculate the measure, which is a daunting task even for medium-scale systems. Recently, Ref.~\cite{Bose2018Machine} shows that with the assistance of machine learning, one can extract the negativity just from the 3-order moment $\tr[(\rho_{AB}^{T_B})^3]$. Note that the first two moments $\tr[(\rho_{AB}^{T_B})]=\tr(\rho_{AB})=1$ (normalization), and $\tr[(\rho_{AB}^{T_B})^2]=\tr(\rho_{AB}^2)$ (purity) do not carry any information about the negative part of the spectrum. The ($3$-order) negativity-moment can be expressed as 
\begin{equation}\label{Eq:PerNg0}
    \begin{aligned}
        \tr\left[ (\rho_{AB}^{T_B})^3 \right]
        &= \tr[ W_{(1,2,3)}^{AB} (\rho_{AB}^{T_B})^{\otimes 3} ]\\
        &= \tr\left\{\left[ {W_{(1,2,3)}^{AB}}\right]^{T_B} \rho_{AB}^{\otimes 3} \right\}\\
        &= \tr\left[ (W_{(1,2,3)}^A\otimes W_{(1,3,2)}^B) \rho_{AB}^{\otimes 3}\right].\\
    \end{aligned}
\end{equation}
Here, in the first line, the cyclic permutation operator is adopted to equivalently express the $3$-power of an operator, and $W_{(1,2,3)}\ket{a_1,a_2,a_3}=\ket{a_2,a_3,a_1}$.
In the second line, the transpose is equivalently put on the permutation operator and $W_{(1,2,3)}^T=W_{(1,2,3)}^{-1}=W_{(1,3,2)}$. Hereafter we use the cycle structures to denote the elements in $S_3$. In Fig.~\ref{Fig:3copyAB} we visualize the $3$-order purity and negativity using diagram representations.
\begin{figure}[htbp]
\centering
\includegraphics[width=0.45\textwidth]{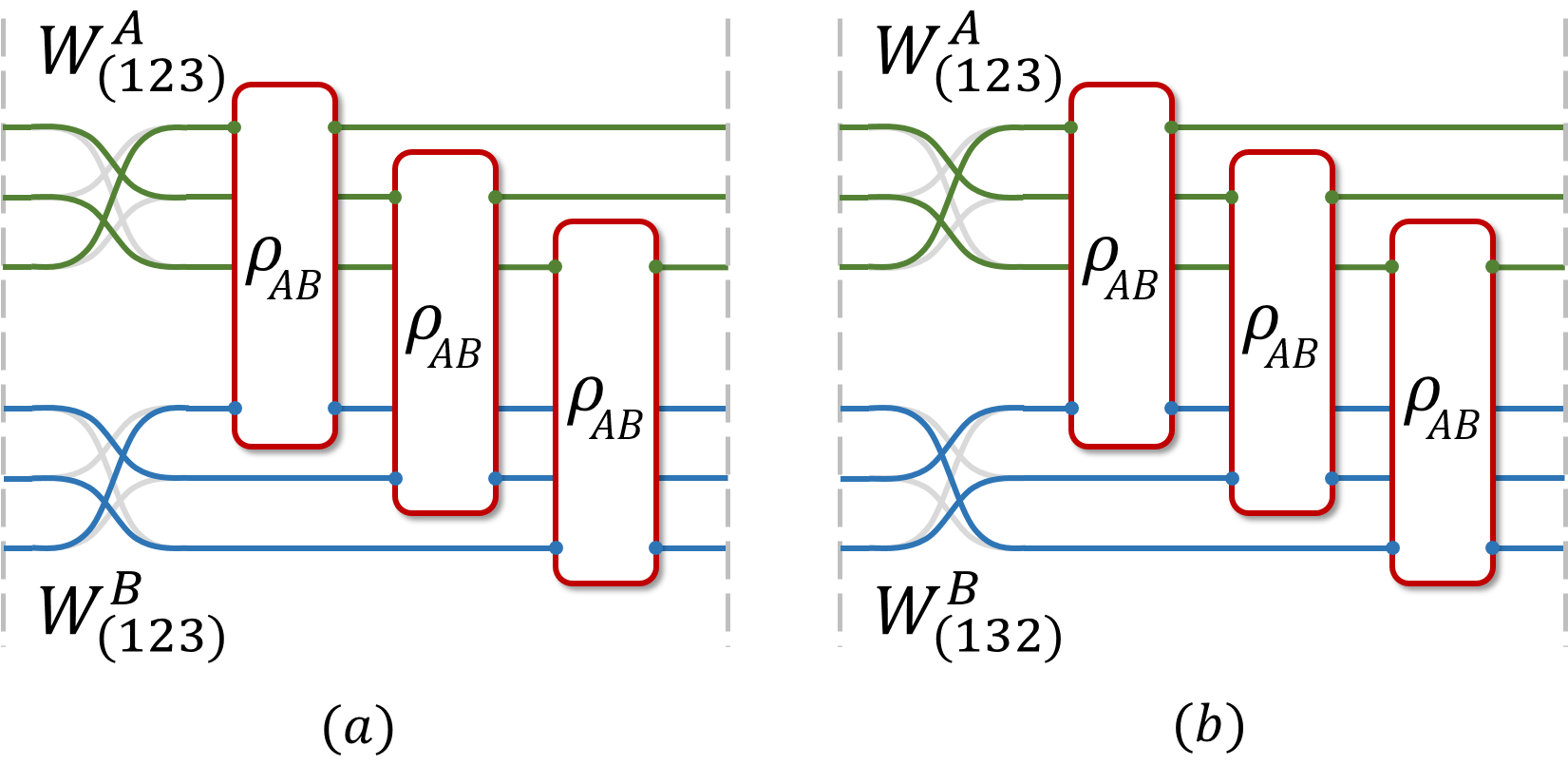}
\caption{Diagram representations of: (a) $3$-order purity $\tr[\rho_{AB}^3]$ and (b) ($3$-order) negativity-moment $\tr[(\rho_{AB}^{T_B})^3]$, as the cyclic operation on 3 copies of $\rho_{AB}$. The gray dashed lines denote a periodic boundary condition, i.e., the trace operation. In the purity case, the two cyclic permutation on $A$ and $B$ are the same, but for the negativity case they are opposite. The shaded lines denote the other possible realizations by symmetry.}
\label{Fig:3copyAB}
\end{figure}


As shown in Eq.~\eqref{Eq:PerNg0}, the direct measurement of the negatvity-moment needs three copies of $\rho_{AB}$ \cite{Bose2018Machine}. In this letter, we utilize the random unitary to effectively make the virtual copies \cite{van2012Measuring,Elben2018Random}, and thus only needs just single-copy of $\rho_{AB}$ to realize the same measurement. Due to the symmetry between two parties, we denote
\begin{equation}\label{Eq:PerNg}
    \begin{aligned}
    M_{neg}=\frac{1}{2} \left( W_{(1,2,3)}^A\otimes W_{(1,3,2)}^B+W_{(1,3,2)}^A\otimes W_{(1,2,3)}^B \right)
       \end{aligned}
\end{equation}
which is a Hermitian operator in the current form.

\section{Weingarten integral and its virtual realization}
The core idea of the single-copy estimation is to effectively create cyclic permutation operations, such as $W_{(1,2,3)}$, using random unitary. 
To this end, We first briefly recast the basics about the integral of Haar random unitary, i.e., Weingarten integral \cite{Collins2006Integration,Collins2016Random}. Given any linear operator on the $k$-copy of $d$-dimension Hilbert space $X\in \mathcal{H}_d^{\otimes{k}}$, the result of the $k$-fold unitary twirling channel shows
\begin{equation}\label{Eq:wein}
    \begin{aligned}
    \Phi^k(X) :&= \int_{\mathrm{Haar}}dU U^{\otimes{k}} X U^{\dag\otimes{k}} \\
    &=\sum_{\pi,\sigma \in S_k} C_{\pi,\sigma}\tr(W_{\pi}X)W_{\sigma},
    \end{aligned}
\end{equation}
where the integral of $U\in \mathcal{H}_d$ is from the Haar measure, and the real coefficients $C_{\pi,\sigma}$ 
constitute the symmetric Weingarten matrix \cite{Collins2006Integration,Roberts2017Chaos}.
Note that the result of unitary twirling $\Phi^k(X)$ is the linear combination of the permutation operators $W_{\sigma}$.
In our case, when $k=3$ one can directly see that the permutation operator such as $W_{(1,2,3)}$ emerges under such integral. We also remark that the integral on the unitary $k$-design ensemble (such as the Clifford gates \cite{Webb2015design,Zhu2017designs}) is enough to reproduce the same twirling result in Eq.~\eqref{Eq:wein}. Thus hereafter we denote the average on unitary ensemble by $\sub{\mbb{E}}{U\in\mc{E}}$, where $\mc{E}$ can be Haar measure or other unitary $3$-design ensembles.

Our scheme utilizes the multiplication of measurement probabilities to \emph{virtually} realize the $3$-copy integral. Here, we start from introducing the single-party scheme shown as follows.
\begin{enumerate}
  \item Prepare the state $\rho\in \mc{H}_d$.
  \item Randomly choose unitary $U\in \mc{H}_d$ from the ensemble $\mc{E}$, and operate it on $\rho$ to get $\mc{U}(\rho)= U\rho U^{\dag}$.
  \item Measure the state $\mc{U}(\rho)$ in the computational basis $\{\ket{s}\}$ of $\mc{H}_d$.
\end{enumerate}
For a given $U$, by repeating the measurements, one can obtain an estimation of the probability $P(s|U)=\tr\left[\ketbra{s}{s}\mc{U}(\rho)\right]$.

By multiplying the probability $P(s|U)$ three times under the same $U$ and average the different realizations from the unitary ensemble, one has
\begin{equation}\label{Eq:MainintSingle}
    \begin{aligned}
    \Omega(\vec{s}, \rho):=&\sub{\mbb{E}}{U\in\mc{E}}P(s|U)P(s'|U)P(s''|U)\\
    &=\sub{\mbb{E}}{U\in\mc{E}}\tr\left[\ketbra{\vec{s}}{\vec{s}} \mc{U}(\rho)^{\otimes 3} \right]\\
    &=\tr\left[\ketbra{\vec{s}}{\vec{s}}\Phi^3(\rho^{\otimes 3})\right]\\
    &=\sum_{\pi,\sigma\in S_3} C_{\pi,\sigma} \tr( W_{\pi}\rho^{\otimes 3})W_{\sigma}(\vec{s}),
\end{aligned}
\end{equation}
where $\vec{s}:=(s,s',s'')$ is a 3-dit string, $W_\sigma(\vec{s}):=\bra{\vec{s}}W_{\sigma}\ket{\vec{s}}$,
and the final line is a direct application of the Weingarten integral in Eq.~\eqref{Eq:wein}. The term $W_\sigma(\vec{s})$ is just some delta function of the indices, for instance, for $\sigma=(1,2)$ and $(1,2,3)$, one has $\delta_{ss'}$ and $\delta_{ss's''}$, respectively. And the purity quantities can appear here, for instance, if $\pi=(1,2)$,~$(1,2,3)$, $\tr( W_{\pi}\rho^{\otimes 3})=\tr(\rho^2)$,~$\tr(\rho^3)$.  See Fig.~\ref{Fig:sketch} (a),~(b) for a diagrammatic illustration.

To extract the target permutations, for example
\begin{equation}\label{Eq:mainM+}
    \begin{aligned}
    M_+:=W_{(1,2,3)}+W_{(1,3,2)},
    \end{aligned}
\end{equation}
which corresponds to $\tr(\rho^3)$. One can further linearly combine the result $\Omega(\vec{s}, \rho)$ for different measurement outputs $\vec{s}$, described by a function of the indices $O(\vec{s})$.

\begin{equation}\label{Eq:mainPost}
    \begin{aligned}
    \sum_{\vec{s}}O(\vec{s})\Omega(\vec{s}, \rho)=&\tr\left[\sum_{\vec{s}} O(\vec{s})\ketbra{\vec{s}}{\vec{s}}\Phi^{3}(\rho^{\otimes 3})\right]\\
   =&\tr\left[\Phi^{3}(O)\rho^{\otimes 3}\right],
\end{aligned}
\end{equation}
where $O=O(\vec{s})\ketbra{\vec{s}}{\vec{s}}$ is the corresponding diagonal operator 
and the final line is due to $(\Phi^{k})^*=\Phi^k$. See Fig.~\ref{Fig:sketch} (c),~(d) for an illustration.

Note that the twirling channel is now on $O$. As a result, the goal of post-processing is to find proper $O$ such that $\Phi^{3}(O)$ outputs the target combination of permutations. For $M_+$ in Eq.~\eqref{Eq:mainM+}, one has $\Phi^3(O_+)=M_+$ with
\begin{equation}\label{Eq:MainSglPost}
    \begin{aligned}
     O_+(\vec{s})=\alpha\delta_{ss's''}+\beta(\delta_{ss'}+\delta_{s's''}+\delta_{ss''})+\gamma
     \end{aligned}
\end{equation}
and $\alpha=(d+1)(d+2), \beta=-(d+1), \gamma=2$.

\begin{figure}[htbp]
\centering
\includegraphics[width=0.48\textwidth]{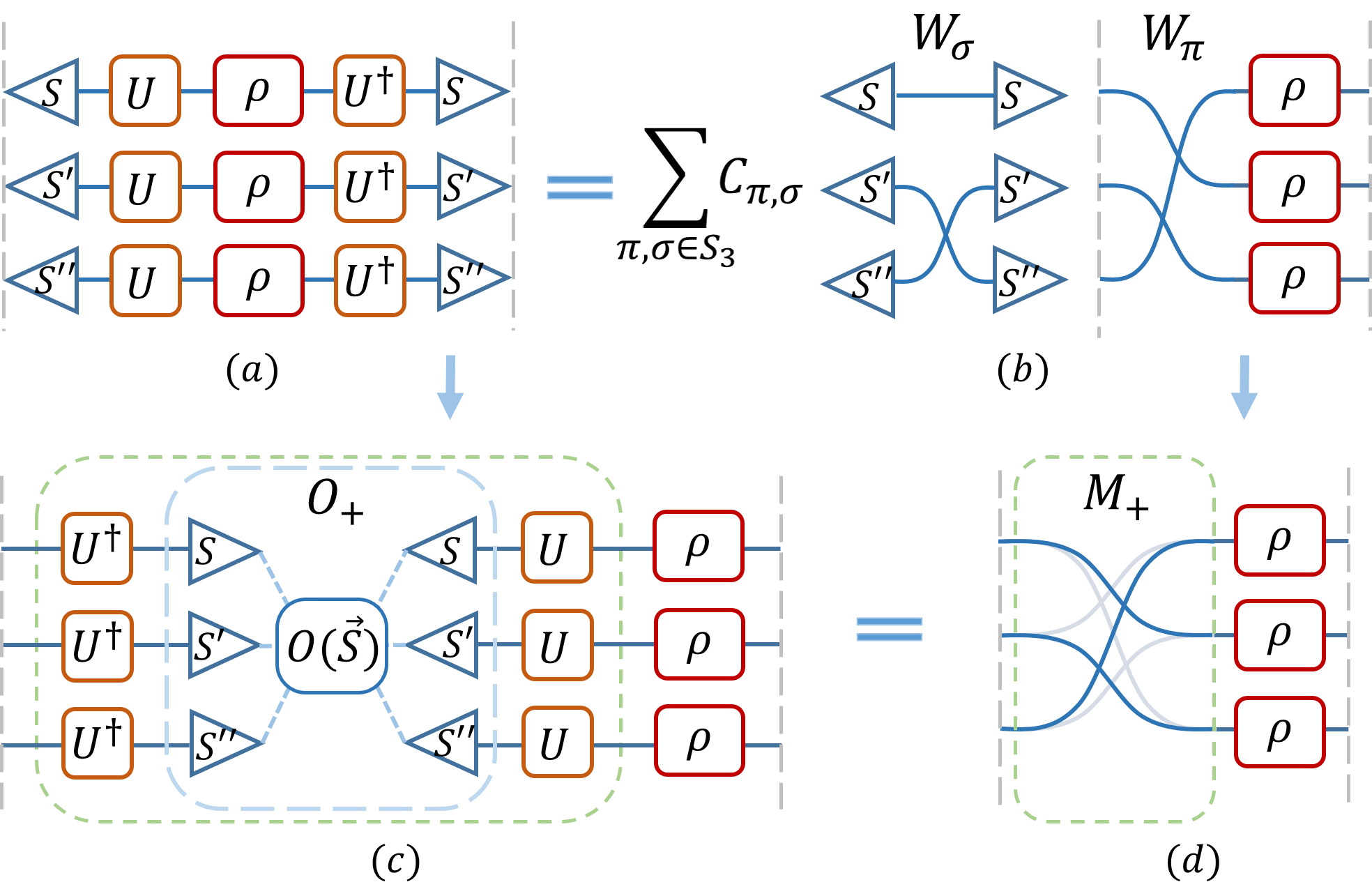}
\caption{The sketch of the single-copy evaluation, where $U$ and $U^{\dag}$ in the orange boxes denote the average on the unitary ensemble. (a) By multiplying the probability $P(s|U)$ on a single-copy three times, we equivalently twirl the 3-copy state to get (b), where there are a few of permutations. (c) In further post-processing, we linearly combine the results fro different outputs $\vec{s}$, and the effective operation is a twirling channel on the diagonal matrix $O$ with elements $O(\vec{s})$. By properly choosing $O$, one can get the target permutations, for instance, $M_+$ pf Eq.~\eqref{Eq:mainM+} shown in (d).}
\label{Fig:sketch}
\end{figure}

\section{Measuring Negativity-moment}

To extract the $W_{(1,2,3)}^A\otimes W_{(1,3,2)}^B$ type operator of $M_{neg}$ in Eq.~\eqref{Eq:PerNg}, one should effectively twirl on both subsystems $A$ and $B$. Similar to the single-party protocol, we still do projective measurement on the $\{\ket{a}\}$ and $\{\ket{b}\}$ of $A$ and $B$. But now the Step 2 is substituted with
the random unitaries $U_A\otimes U_B$ from the Haar measure (or the unitary 3-design) of $\mathcal{H}_A$ and $\mathcal{H}_B$ independently. We denote this by the bi-local unitary scheme.

For a given $\mathcal{U}_A \otimes \mathcal{U}_B(\rho_{AB})$, by repeating the measurements, one can have an estimation of the probability
\begin{equation}
    \begin{aligned}
    P(a,b|U_A,U_B)=\tr\left[\left(\ketbra{a}{a}\otimes\ketbra{b}{b}\right)\mathcal{U}_A\otimes \mathcal{U}_B(\rho_{AB})\right].
    \end{aligned}
\end{equation}
Similar to Eq.~\eqref{Eq:MainintSingle}, by multiplying $P(a,b|U_A,U_B)$ three times one has
\begin{equation}\label{Eq:Mainint}
    \begin{aligned}
    &\Omega(\vec{a},\vec{b},\rho_{AB})=
    \tr\left[\left(\ketbra{\vec{a}}{\vec{a}}\otimes \ketbra{\vec{b}}{\vec{b}}\right)\Phi^{3}_A\otimes\Phi^{3}_B(\rho_{AB}^{\otimes 3})\right]\\
    &=\sum_{\substack{\pi,\sigma, \\\pi',\sigma' \in S_3}} C_{\pi,\sigma}C_{\pi',\sigma'}
    \tr\left[(W_{\pi}^A\otimes W_{\pi'}^B) \rho_{AB}^{\otimes 3}\right]
    W_{\sigma}^A(\vec{a}) W_{\sigma'}^B(\vec{b}).
\end{aligned}
\end{equation}

There are totally possible $6^2=36$ combinations of $\{\pi,\pi'\}$ appearing in $\tr\left[(W_{\pi}^A\otimes W_{\pi'}^B) \rho_{AB}^{\otimes 3}\right]$.
\comments{
\red{polish}
There are totally possible $6^2=36$ combinations of $\{\pi_A,\pi'_B\}$, and the result of $\tr(W_{\pi_A}\otimes W_{\pi'_B} \rho_{AB}^{\otimes 3})$ are listed in the \red{table}. We focus on the terms relative to the negativity, and there are other terms which may be of independent interest such as $\tr[\rho_{AB}(\rho_A\otimes\rho_B)]$, quantifying the total correlation between $A$ and $B$.
}
To extract the target permutations, like the single-party case in Eq.~\eqref{Eq:mainPost} one can introduce the post-processing diagonal operator $O=O(\vec{a},\vec{b})\ketbra{\vec{a},\vec{b}}{\vec{a},\vec{b}}$, such that $\sum_{\vec{a},\vec{b}}O(\vec{a},\vec{b})\Omega(\vec{a},\vec{b}, \rho_{AB})=\tr\left[\Phi_A^{3}\otimes \Phi_B^{3}(O)\rho_{AB}^{\otimes 3}\right]$.

Here we are interested in the negativity-moment $\tr\left[ (\rho_{AB}^{T_B})^3 \right]$,
and the corresponding observable $M_{neg}$ in Eq.~\eqref{Eq:PerNg} can be decomposed into the following two terms
\begin{equation} \label{Eq:Decom1}
\begin{aligned}
   M_{neg}= \frac1{2}(M_+^A\otimes M_+^B-M_+^{AB}),
 \end{aligned}
\end{equation}
where $M_+$ is defined in Eq.~\eqref{Eq:mainM+}.
The first term can be realized locally shown as follows.
\begin{prop}
$M_+^A\otimes M_+^B$ in Eq.~\eqref{Eq:Decom1} can be realized with the bi-local random unitary scheme, such that $\Phi^{3}_A\otimes\Phi^{3}_B(O)=M_+^{A}\otimes M_+^B$. Specifically, one can find a product type $O=O_A\otimes O_B$, satisfying $\Phi^{3}_A(O_A)=M_+^{A}$ and $\Phi^{3}_B(O_B)=M_+^{B}$ respectively, with $O_A$ and $O_B$ given in Eq.~\eqref{Eq:MainSglPost}.
\end{prop}

\comments{
It is not hard to check that
\begin{equation}\label{}
   \begin{aligned}
 \tr\left[\frac1{2}(M_+^A\otimes M_+^B)\rho_{AB}^{\otimes 3}\right]=\tr(\rho_{AB}^{3})+\tr(\rho_{AB}^{T_B3}).
       \end{aligned}
\end{equation}
}

The second term $M_+^{AB}$ 
can be realized in a similar way, nevertheless with global random unitary scheme $U_{AB}$. Using the same post-processing operator $O_{AB}$ as in Eq.~\eqref{Eq:MainSglPost}, the global twirling makes $\Phi^{3}_{AB}(O_{AB})=M_+^{AB}$. The detailed construction of $M_{neg}$ is presented in Appendix~\ref{SSec:ZZnegativity}.

Although the evaluation scheme requires only single-copy quantum operations, it needs the global unitaries $U_{AB}$ which is not easily accessible in the experiments. One may ask if it is possible to evaluate $M_{neg}$ just using bi-local unitary. Unfortunately, in Appendix~\ref{SSec:ZZnegativity}, we prove the following no-go result.
\comments{
\begin{prop}\label{Pro:noLocal}
With the bi-local random unitaries $U_A\otimes U_B$, there is no computational basis observable $O_{neg}=\sum_{\vec{a},\vec{b}}O(\vec{a},\vec{b}) \ket{\vec{a}}\bra{\vec{a}} \otimes \ket{\vec{b}}\bra{\vec{b}}$ such that $\Phi^{3}_A\otimes\Phi^{3}_B(O_{neg})=M_{neg}$.
\end{prop}
}

\begin{prop}\label{Pro:noLocal}
Using the bi-local random unitary scheme, there is no post-processing strategy $O$ such that $\Phi^{3}_A\otimes\Phi^{3}_B(O)=M_{neg}$.
\end{prop}
Proposition \ref{Pro:noLocal} also indicates that $\tr(\rho_{AB}^{3})$ can not be measured in a bi-local manner, which answers an open question regarding higher-order moments \cite{Elben2019toolbox,Brydges2019Probing}.

\comments{
According to the Weingarten formula in Eq.~\eqref{Eq:wein}, after a bi-local twirling $\Phi^3_A\otimes \Phi^3_B$, an bi-partite observable becomes,
\begin{equation} \label{eq:bilocalO}
\begin{aligned}
\Phi^3_A\otimes \Phi^3_B(O) &= \sum_{\pi,\pi',\sigma,\sigma'\in S_3} \Big\{ C_{\pi,\sigma} C_{\pi',\sigma'} \\ &\tr\left[(W_\pi^A \otimes W_{\pi'}^B)O\right] (W_\sigma^A \otimes W_{\sigma'}^B) \Big\}.
\end{aligned}
\end{equation}
}

The essence of the no-go result is that one cannot discriminate the two 3-order cyclic permutations by local basis,
\begin{equation}
\tr[W_{(1,2,3)}^A \ket{\vec{a}}\bra{\vec{a}}] = \tr[W_{(1,3,2)}^A \ket{\vec{a}}\bra{\vec{a}}],
\end{equation}
similar for $B$. Consequently, the bi-local scheme always take $W_{(1,2,3)}^A\otimes W_{(1,2,3)}^B$ and $W_{(1,2,3)}^A\otimes W_{(1,3,2)}^B$ equally, which hiders our construction. 

Thus we further consider the Bell measurement on system $A$ and $B$. Note that for the Bell state $\ket{\Psi_+}:=\frac{1}{\sqrt{d}}\sum_{s=0}^{d-1}\ket{s,s}_{AB}$,
\begin{equation}
\begin{aligned}
\tr[(W_{(1,2,3)}^A\otimes W_{(1,2,3)}^B)\Psi_{+}^{\otimes 3}] &= d, \\
\tr[(W_{(1,2,3)}^A\otimes W_{(1,3,2)}^B)\Psi_{+}^{\otimes 3}] &= 1/d^2,
\end{aligned}
\end{equation}
which breaks the symmetry. Therefore, we construct an observable $O_{Bell}$ on the Bell basis, which represents a post-processing strategy using the Bell state measurement (BSM) on $\rho_{AB}$.
\comments{
Therefore, we may consider to construct an observable on the Bell basis
\begin{equation}
O_{Bell} = \sum_{\vec{u},\vec{v}} O(\vec{u},\vec{v})\Psi_{u_1,v_1}\otimes \Psi_{u_2,v_2}\otimes \Psi_{u_3,v_3},
\end{equation}
where $\ket{\Psi_{u,v}} :=\frac{1}{\sqrt{d}}\sum_{s=0}^{d-1}\exp\left(\frac{2\pi i}{d} lv\right)\ket{s,s+u}_{AB}$ forms the qudit Bell-state basis.} By decomposing $M_{neg}$ as follows,
\begin{equation} \label{Eq:Decom2}
\begin{aligned}
   M_{neg}= \frac{1}{4}(M_+^A\otimes M_+^B - M_-^A\otimes M_-^B),
\end{aligned}
\end{equation}
with $M_{-}^{A}:=(W_{(1,2,3)}^{A}-W_{(1,3,2)}^{A})$ and similar for $B$, we show that

\begin{prop}\label{prop:MmmBSM}
$M_-^A\otimes M_-^B$ in Eq.~\eqref{Eq:Decom2} can be realized with bi-local random unitary scheme, with the final measurement substituted by the BSM between $A$ and $B$, i.e., there exist Bell-basis observable $O_{Bell}$ such that $(\Phi^3_A\otimes \Phi^3_B)(O_{Bell})= M_-^A\otimes M_-^B$.
\end{prop}
The proof of the proposition and the detailed construction of $O_{Bell}$ is in Appendix~\ref{SSec:BSMnegativity}.

\comments{
Furthermore, we may extend the bi-local strategy to a general $m$-local case. Consider a system with $m$-qudits, if we are able to realize the BSM on each two of qudits, then with the estimated values of $\{M_{+}^{(i)}\}_i$ and $\{M_{--}^{(i,j)}\}_{i,j}$, we can construct the negativity observable $M_{neg}^{AB}$ over any bi-partition on the $m$-qudit system from a generalized version of Eq.~\eqref{Eq:Decom2}. We will present the explicit $m$-local results in a follow-up work \cite{Liu2020correlation}.
}

\section{quantifying total correlation}

Entanglement negativity quantifies the quantum correlation between the subsystems $A$ and $B$. Here, we extend the random circuit scheme to extract the total correlation with delicate post-processing.

The quantity used to quantify the total correlation is based on a fidelity measure between two (mixed) states $\mathcal{F}_2(\rho_1,\rho_2)$,
which is defined by the operator 2-norm \cite{Liang2019fidelity}, and also used to quantify the overlap of states \cite{Elben2020Cross}.

In our case, we are interested in the fidelity between $\rho_{AB}$ and the corresponding marginal $\rho_A\otimes\rho_B$,
\begin{equation} 
\begin{aligned}
  \mathcal{F}_2(\rho_{AB},\rho_A\otimes\rho_B)
  &=\frac{\tr[\rho_{AB}(\rho_A\otimes \rho_B)]}{\max\{\tr[\rho_{AB}^2], \tr[\rho^2_A]\tr[\rho_B^2]\}}.
\end{aligned}
\end{equation}
Note that the 2-order purity terms in the denominator can be measured with local random unitary scheme \cite{Elben2019toolbox,Brydges2019Probing}. Here,  we focus on the numerator $\tr[\rho_{AB}(\rho_A\otimes \rho_B)]$, and it is not hard to check that
\begin{equation}
  \tr[\rho_{AB}(\rho_A\otimes \rho_B)]=\tr[(W_{\pi}^A\otimes W_{\pi'}^B)\rho_{AB}^{\otimes 3}].
\end{equation}
for any $\pi\neq\pi' \in \{(1,2),(2,3),(1,3)\}$. Without loss of generality, we take $M_{c}=W_{(1,2)}^A\otimes W_{(2,3)}^B$. Recall that in Eq.~\eqref{Eq:Mainint} there are various possible combinations of local permutation operators $W_{\pi}^A\otimes W_{\pi'}^B$, similar as the negativity-moment, we have the following post-processing for the total correlation.
\comments{
\begin{prop}
$M_{c}=W_{(1,2)}^A\otimes W_{(2,3)}^B$ can be realized with bi-local random unitary scheme, such that $\Phi^{3}_A\otimes\Phi^{3}_B(O_c)=M_{c}$. Specifically, one can find a product type $O_c=O_c^A\otimes O_c^B$, satisfying $\Phi^{3}_A(O_c^A)=W_{(1,2)}^A$ and $\Phi^{3}_B(O_c^B)=W_{(2,3)}^B$, respectively, with
\begin{equation}\label{Eq:CorrPost}
\begin{aligned}
  O_c^A &= \sum_{\vec{a}} (-d)^{wt(a_1 a_2)-2} \ket{\vec{a}}\bra{\vec{a}},\\
  O_c^B &= \sum_{\vec{b}} (-d)^{wt(b_2 b_3)-2} \ket{\vec{b}}\bra{\vec{b}}.
\end{aligned}
\end{equation}
Here, \red{we assume $d_A=d_B=d$.} $\vec{a}=(a_1,a_2,a_3)$ and $\vec{b}=(b_1,b_2,b_3)$ are $3$-dit strings, $wt(a_1,a_2)$ denotes the number of the same values in $a_1$ and $a_2$.
\end{prop}
Note that the form of $O_c^A$ and $O_c^B$ are similar, but they act on different copies. We remark that one can generalize the construction of $O_c$ to detect the multipartite correlations, and extract the correlation hierarchy only using local unitary twirling, which will be studied further in a follow-up work \cite{Liu2020correlation}.
}

\begin{prop}
$M_{c}=W_{(1,2)}^A\otimes W_{(2,3)}^B$ can be realized with bi-local random unitary scheme, such that $\Phi^{3}_A\otimes\Phi^{3}_B(O)=M_{c}$. Specifically, one can find a product type $O=O_A\otimes O_B$, satisfying $\Phi^{3}_A(O_A)=W_{(1,2)}^A$ and $\Phi^{3}_B(O_B)=W_{(2,3)}^B$ respectively, with
\begin{equation}\label{Eq:CorrPost}
    \begin{aligned}
  O_A(\vec{a})&=\alpha_A\delta_{a,a'}+\beta_A,\\
  O_B(\vec{b})&=\alpha_B\delta_{b',b''}+\beta_B,
       \end{aligned}
\end{equation}
with $\alpha_A=(d_A+1)/d_A, \beta_A=-1/d_A$, similar for $B$.
\end{prop}
Note that $O_A$ and $O_B$ show a similar form but act on different copies. We remark that one can generalize the above discussion to multipartite correlations even with local unitary scheme \cite{Liu2020correlation}.
\section{Statistical error analyses} 
\comments{
Now we analyse the efficiency of the negativity-moment detection scheme with $ZZ$-basis measurement. Suppose we randomly choose $N_U'$ bi-local unitaries and $N_M'$ samples for each unitary to estimate $M^A_+\otimes M^B_+$; $N_U$ global unitaries and $N_M$ samples for each unitary to estimate $M^{AB}_+$. The overall sample number is then $N=N_U'N_M'+N_UN_M$. The statistical error arises from two aspects: (i) the finite $N_U(N_U')$ rounds of sampling from the bi-local (global) random unitary ensembles; (ii) the finite sample number $N_M(N_M')$ for a fixed bi-local (global) unitary.
}

Here we discuss the effect of finite number realization on the final result. In our scheme, the statistical error arises from two aspects: (i) the finite $N_U$ rounds of sampling from the random unitary ensemble; (ii) the finite shot number $N_M$ per one unitary round.

Here we assume that different rounds of random unitary and different shots for a given unitary are generated in an independent and identical distributed (i.i.d.) manner. Therefore, one can describe the $i$-th shot for a given unitary $U$ as a random variable $\hat{r}_U(i)$, which takes value $\ket{a}\bra{a}$ with the probability $P(a|\rho,U)=\tr[\ketbra{a}{a}U\rho U^\dag]$. Using these random variables, an unbiased estimator $\hat{M}_{neg}$ can be constructed for $M_{neg}$. Note that in Eq.~\eqref{Eq:Decom1} $M_{neg}$ can be written into two terms,  and here we take the estimator of $M_{+}^{AB}$ as an example, that is,
\begin{equation}
\hat{M}_{+}^{AB}(t) = N_3^{-1}\sum_{i<j<k} \tr\left[\left(\hat{r}_U(i)\otimes\hat{r}_U(j)\otimes\hat{r}_U(k)\right)O_+^{AB}\right]
\end{equation}
where $N_3:=\binom{N_M}{3}$ and $t$ denotes the $t$-th unitary round. It is an unbiased estimator in the sense that
\begin{equation}
\sub{\mbb{E}}{U\in\mc{E}}\sub{\mbb{E}}{a} \hat{M}_{+}^{AB}(t) = \tr[M_{+}^{AB}\rho^{\otimes 3}_{AB}].
\end{equation}
The overall estimator is generated by averaging over $N_U$ rounds
\begin{equation}
\hat{M}_{+}^{AB} = N_U^{-1}\sum_{t=1}^{N_U} \hat{M}_{+}^{AB}(t),
\end{equation}
which is clearly unbiased. The estimator $\hat{M}_{++}^{AB}$ of $M_{+}^{A}\otimes M_+^{B}$ and thus $\hat{M}_{neg}$ can be constructed in a similar way. See Appendix~\ref{Sec:suppStat} for the detailed construction.
When $D\gg 1$, the variance of $\hat{M}_{neg}$ has the following form
\begin{equation}\label{Eq:ErrAnaly}
\begin{aligned}
\mb{Var}[\hat{M}_{neg}] &\sim \frac{1}{N_U}\left[\frac{c_0}{D}+\frac{c_1}{N_M} + \frac{c_2D}{N_M^2} + \frac{c_3D^2}{N_M^3}\right], \\
\end{aligned}
\end{equation}
where $D=d_A d_B$ is the dimension of the total Hilbert space, and $\{c_i\}$ are some constants related to the state $\rho$.

\begin{figure}[tbh!]
\centering
\includegraphics[width=0.5\textwidth]{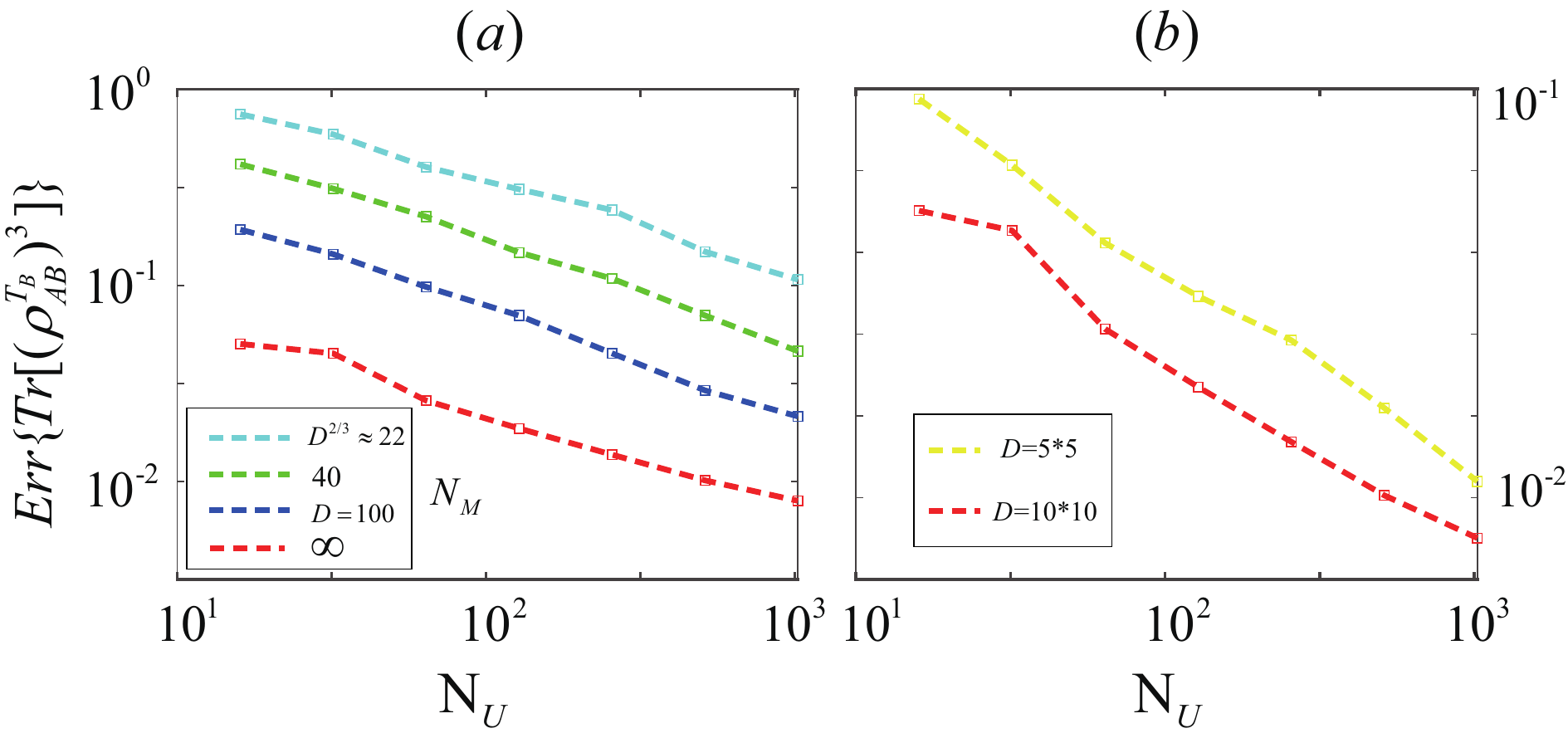}
\caption{Scaling of statistical errors. (a) Average statistical error of the estimated negativity-moment $\tr\left[ (\rho_{AB}^{T_B})^3 \right]$ as a function of $N_U$ for various
$N_M$ with D = 10*10; (b) for D=5*5 and 10*10, with $N_M=\infty$. The unitaries
are sampled from the Haar measure numerically, and the prepared state is Bell state mixed with white noise $p=0.3$, i.e., $\rho_{AB}=(1-p)\Psi_{+}+p\mathbb{I}/D$.
}\label{Fig:NegToT1}
\end{figure}

In the limit $D \gg N_M \gg 1$, which is the regime of practical interest, the variance behaves as $\mb{Var}[\hat{M}_{neg}]\sim D^2/[N_UN_M^3]$. In this case, to make the error less than $\epsilon$, one needs $N_M=D^{2/3}$ and $N_U=\mc{O}(1/\epsilon^2)$. As a result, the overall realizations of experiment $N$ scales like $\frac1{\epsilon^2}D^{2/3}$. Even though it scales polynomially with the dimension $D$, and thus exponentially with the system size, it is more efficient than the conventional tomography. Moreover, we also find that for mixed states and entangled states, which are actually the normal cases, the corresponding error decreases compared to the pure product states. 
Fig.~\ref{Fig:NegToT1} shows the numerical results of the statistical error for $\mc{H}_5\otimes\mc{H}_5$ and $\mc{H}_{10}\otimes\mc{H}_{10}$ systems. One can see that for different values of $N_M$, the error always decreases with slope $-0.5$ versus $N_U$ in the Log-Log plot; and the error decreases as the increase of the dimension $D$, which are both described by our analytical result in Eq.~\eqref{Eq:ErrAnaly}. See Appendix~\ref{Sec:suppNumer} for more numerical results.

\section{Concluding remarks}
In this letter, we proposed a scheme to estimate the 3-order moment related to the entanglement negativity, based on the random unitary evolution and projective measurements. The scheme can also be used to quantify the total correlation. 
Moreover, we propose a general method to construct the unbiased estimator and analyse the statistical error, which can also be applied to other quantum benchmarking tasks.

Due to its single-copy property, the proposed scheme is feasible with current quantum technology. Note that the whole scheme only requires unitary $3$-designs, which can be realized by the Clifford circuits that is widely used in the quantum information processing \cite{Webb2015design,Zhu2017designs}. 
These circuits can be implemented on various quantum platforms, such as superconducting circuits, ion trap, and linear optics.
Besides, there are proposals to realize the random circuits from quenched Hamiltonian evolution \cite{Nakata2017Pseudorandomness,Elben2018Random,Vermersch2018quenches}.


The integration of Bell measurement with the bi-local scheme can also used to measure the 3-order purity $\tr(\rho_{AB}^{3})$, which may be extended to higher order ones to identify the entanglement spectrum \cite{Haldane2008Spectrum}. For the total correlation quantified by the fidelity, it can be directly extended to multipartite scenario to characterize the correlation hierarchy \cite{Cirac2000Three,Bennett2011Postulates,Davide2017Quantifying}. 

It is intriguing to apply the proposed scheme to characterize other properties of a many-body wave function, such as the high order out-of-time-order correlators (especially the six-point one) \cite{Roberts2017Chaos,Vermersch2019Probing} and the topological invariants \cite{Elben2020topological}. Note that the Bell measurement strategy could  contribute to accessing these quantities with local unitaries. Moreover, it is also interesting to extend the current scheme to bosons and fermions in the quantum simulators \cite{Bloch2012ultracold,Vermersch2018quenches,Elben2018Random}.

\begin{acknowledgments}
We thank Xun Gao, Alioscia Hamma, Arthur Jaffe, Xiongfeng Ma, and Zhen-Sheng Yuan for helpful discussions. Y.~Zhou was supported in part by the Templeton Religion Trust under grant TRT 0159 and by the ARO under contract W911NF1910302. P.~Zeng and Z.~Liu were supported by the National Natural Science Foundation of China Grants No.~11875173 and No.~11674193, the National Key R\&D Program of China Grants No.~2017YFA0303900 and No.~2017YFA0304004, and the Zhongguancun Haihua Institute for Frontier Information Technology.
\end{acknowledgments}


%

\onecolumngrid
\newpage
\begin{appendix}
\maketitle

We provide detailed description of observable construction, statistical analysis and numerical results. In Sec.~\ref{Sec:preliminaries}, we introduce some essential knowledge about the random circuits. In Sec.~\ref{Sec:suppNeg}, we explicitly show how to construct the Negativity-moment observable using global random unitaries and local measurements or local random unitaries and Bell-state measurements. In Sec.~\ref{Sec:suppStat}, we analyze the finite-size performance of global random unitary protocol. Finally, in Sec.~\ref{Sec:suppProof} and \ref{Sec:suppNumer} , we present detailed proofs and more numerical results.

\section{Preliminaries} \label{Sec:preliminaries}

\subsection{Haar measure and unitary design}
In this section and the following one, we give a brief introduction to the integral of unitary according to Haar measure. And a more detailed review can be found in, e.g., \cite{Collins2016Random,gu2013moments,Roberts2017Chaos}

Haar measure is the unique measure of unitary $U\in \mc{H}_d$, which is invariant of left and right multiplying any unitary $V$ for any function $f(U)$. That is,
\begin{equation}
\begin{aligned}
    \int_{\Hr}dU=1,\ \ \int_{\Hr}dUf(U)=\int_{\Hr}dUf(VU)=\int_{\Hr}dUf(UV).
\end{aligned}
\end{equation}

In our work, we mainly focus on the integral on the $k$-copy Hilbert space $\mathcal{H}_d^{\otimes{k}}$,
\begin{equation}\label{}
    \begin{aligned}
    \Phi^k(X):=\int_{\mathrm{Haar}}dU U^{\otimes{k}} X U^{\dag\otimes{k}}
    \end{aligned}
\end{equation}
where $X$ is a linear operator and the quantum channel $\Phi^k(\cdot)$ is usually called the ``twirling'' operation. In the following section, we give the explicit formula for this integral.

Haar measure is a continuous measure on the Hilbert space, and it is not practical to realize. Alternatively, if one is just interested in the first $k$-moments of the integral, it is found that one can use other unitary ensemble. An unitary ensemble $\mathcal{E}$ is called an unitary $k$-design, if for any $X$ one has
\begin{equation}\label{}
    \begin{aligned}
    \Phi_{\mathcal{E}}^k(X):=\int_{\mathcal{E}}dU U^{\otimes{k}} X U^{\dag\otimes{k}}=\Phi^k(X),
    \end{aligned}
\end{equation}
i.e., the $k$-fold twirling channel of $\mathcal{E}$ is the same with the one of Haar. Note that $\mathcal{E}$ is an unitary $k$-design then it is also a unitary $k-1$-design by definition.
It is known that the Pauli group is unitary 1-design, and the Clifford group is unitary
3-design but fails to be a 4-design \cite{Webb2015design,Zhu2017designs,Zhu2013design}.


\subsection{Schur-Weyl duality and Weigartan formula}
%
In this section, we introduce the explicit result of the twirling operation referred as Weingarten formula, which can be derived using Schur-Weyl duality.

To this end, we first give the definition of the representation of the permutation element $\pi \in S_k$ on $\mathcal{H}_d^{\otimes{k}}$,

\begin{equation}\label{}
    \begin{aligned}
    W_{\pi}=\sum_{s_i\in[d]}\ket{s_{\pi(1)},s_{\pi(2)},\cdots,s_{\pi(k)}}\bra{s_1,s_2,\cdots,s_k}
    \end{aligned}
\end{equation}
where $[d]=\{0,1,2,\cdots d\}$.

It is not hard to see that $[W_{\pi},U^{\otimes k}]=0$, thus the permutation operator is invariant under the twirling channel $\Phi^k(W_{\pi})=W_{\pi}$. In fact, due to the
Schur–Weyl duality which makes connection between the irreducible representations (irreps.) of the permutation group $S_k$ and unitary group $U(d)$, the twirling result can be spanned by all $\{W_{\pi}\}$, i.e.,
\begin{equation}\label{}
    \begin{aligned}
    \Phi^k(X)=\sum_{\pi,\sigma \in S_k} C_{\pi,\sigma}\tr(W_{\pi}X)W_{\sigma},
    \end{aligned}
\end{equation}
where the real coefficients $C_{\pi,\sigma}$
constitute the symmetric Weingarten matrix $C$. The index of the Weingarten matrix $C_{\pi,\sigma}$ is the permutation operator, and it is the pseudo-inverse (can be inversed as $d\geq k$) of the Gram matrix $Q_{\pi,\sigma}=d^{\mathrm{cycles} (\pi\sigma)}$, $\mathrm{cycles} (\pi\sigma)$ counts the cycle number of $\pi\sigma$ depending on the conjugate class.

When one operates the k-fold twirling channel on any pure product symmetric state $\ket{\psi}^{\otimes k}$, the result is proportional to the symmetric subspace \cite{Harrow2013symmetric} showing,
\begin{equation}\label{}
    \begin{aligned}
    &\Phi^k(\psi^{\otimes k})=\frac{\mathbf{P}_{\mathrm{sym}}}{D_{\mathrm{sym}}},\\
    &\mathbf{P}_{\mathrm{sym}}=\frac1{k!}\sum_{\pi\in S_k}W_\pi,\ D_{\mathrm{sym}}=C_{k}^{d+k-1}.
    \end{aligned}
\end{equation}
where $\mathbf{P}_{\mathrm{sym}}$ is the projector of the symmetric subspace.
\subsection{Heisenberg-Weyl operator and Bell-state measurement}

For a qudit system $A$, we denote the computational basis as $\{\ket{l}\}_{l=0}^{d-1}$. The generalized Pauli generators are defined to be
\begin{equation}
\begin{aligned}
    Z &:= \sum_{l=0}^{d-1} \exp{\left(i\dfrac{2\pi}{d}l\right)} \ket{l}\bra{l}, \\
    X &:= \sum_{l=0}^{d-1} \ket{l+1}\bra{l}. \\
\end{aligned}
\end{equation}
Here, the addition operation $+$ on the computational basis is defined on the ring $\mathbb{Z}_d$. The Heisenberg-Weyl operator $P(u,v)$ is defined to be
\begin{equation}
    P(u,v) := X^u Z^v = \sum_{l=0}^{d-1} \exp{\left(i\dfrac{2\pi}{d}v l\right)} \ket{l+u}\bra{l},
\end{equation}
with $u,v= 0,1,...,d-1$.
It is easy to verify that
\begin{equation}\label{eq:Acomm}
\begin{aligned}
X^d =& Z^d=I, \quad (X^u)^\dagger = X^{-u}, \quad (Z^v)^\dagger = Z^{-v}, \\
X^u Z^v &= \exp\left(-i\frac{2\pi}{d}uv\right)Z^v X^u, \\
P(u,v)P(u',v') &= \exp\left(-i\frac{2\pi}{d}(uv' - vu')\right) P(u',v')P(u,v). \\
\end{aligned}
\end{equation}

Define $\Psi_{0,0} := \Psi = \dfrac{1}{\sqrt{d}}\sum_{j=0}^{d-1} \ket{jj}$. The generalized qudit Bell states \cite{Bennett1993Teleporting} are
\begin{equation}
\begin{aligned}
\ket{\Psi_{u,v}}_{AB} &:= P_B(u,v)\ket{\Psi}_{AB} \\
&= \frac{1}{\sqrt{d}}\sum_{l=0}^{d-1} \exp\left(\frac{2\pi i}{d}l v \right)\ket{l}_A\otimes \ket{l+u}_B,
\end{aligned}
\end{equation}
Denote $\Psi_{u,v}:=\ket{\Psi_{u,v}}\bra{\Psi_{u,v}}$. The qudit Bell states $\{\Psi_{u,v}\}_{u,v=0}^{d-1}$ form an orthonormal basis, $\braket{\Psi_{u,v}|\Psi_{u',v'}} = \delta_{u,u'} \delta_{v,v'}$. We define the Bell-state measurement (BSM) on two qudit system $A$ and $B$ by the projective measurement on $\{\Psi_{u,v}\}_{u,v=0}^{d-1}$.

\section{Negativity-moment observable construction} \label{Sec:suppNeg}

In this section, we discuss how to construct the negativity observables. We first construct the $3$-order moment observable and the negativity-moment observable with global random unitaries, and then construct the negativity-moment observable with local random unitaries and Bell-state measurement.

Here, we use $W_0:= W_{(1,2,3)}$ and $W_1:=W_{(1,3,2)}$ to simplify the notation for two cyclic operations. Moreover, we define $M_{x0} = M_{+} := (W_0 + W_1)$ and $M_{x1} = M_{-} := (W_0 - W_1)$.

\subsection{3-order purity observable based on local random unitaries} \label{SSec:supp3Purity}

As is stated in the main text, we want to construct an observable $O\in\mc{L}((\mc{H}^A)^{\otimes 3})$ on three copies of the state $\rho\in\mc{D}(\mc{H}^A)$ such that
\begin{equation}
\tr[\Phi^3(O)\rho^{\otimes 3}] = \tr[M_+ \rho^{\otimes 3}],
\end{equation}
where $\Phi^t(\cdot) = \sub{\mbb{E}}{U\in\mc{E}}((U^\dag)^{\otimes t} \cdot U^{\otimes t})$ is a $t$-copy unitary twirling, where $\mc{E}$ is a set of unitaries which forms a unitary $t$-design.

We now show that, to systematically construct $O$, we only need to consider the projection of $O$ onto the permutation operators $\{W_\pi\}_{\pi\in S_3}$.

\begin{prop} \label{prop:Wbasis}
For two operators $O,P\in\mc{L}((\mc{H}^A)^{\otimes t})$, the following two statements are equivalent,
\begin{enumerate}
\item $\Phi^t(O) = \Phi^t(P)$,
\item $\tr[OW_\pi] = \tr[PW_\pi], \forall \pi\in S_t$.
\end{enumerate}
\end{prop}
\begin{proof}
To prove $1\Rightarrow 2$, we have
\begin{equation}
\begin{aligned}
&\Phi^t(O) = \Phi^t(P), \\
\Rightarrow\; & \tr[\Phi^t(O)W_\pi] = \tr[\Phi^t(P)W_\pi], \forall \pi\in S_t,\\
\Rightarrow\; & \tr[O\Phi^t(W_\pi)] = \tr[P\Phi^t(W_\pi)], \forall \pi\in S_t, \\
\Rightarrow\; & \tr[O W_\pi] = \tr[P W_\pi], \forall \pi\in S_t.
\end{aligned}
\end{equation}
Here, the second $\Rightarrow$ is because $\Phi^t(\cdot)$ is a Hermitian-preserving map. The third $\Rightarrow$ is due to the invariance of $W_\pi$ under the twirling operation.

To prove $2\Rightarrow 1$, we have
\begin{equation}
\begin{aligned}
\Phi^t(O) &= \sum_{\pi,\sigma} c_{\pi,\sigma} \tr[OW_\pi] W_\sigma \\
&= \sum_{\pi,\sigma} c_{\pi,\sigma} \tr[PW_\pi] W_\sigma \\
&= \Phi^t(P).
\end{aligned}
\end{equation}
Here, in the second equality, we have used the statement $2$.
\end{proof}

Proposition~\ref{prop:Wbasis} implies that, the permutation operators $\{W_\pi\}$ forms a complete basis on the inner-projuct space with the Hilbert-Schmidt norm and non-singular gram matrix.

Therefore, to construct $O$ such that $\Phi^3(O)=M_+ = \Phi^3(M_+)$, we only need to construct $O_+$ that satisfies
\begin{equation} \label{eq:OpWp}
\tr[O_+ W_\pi] = \tr[M_+ W_\pi], \forall \pi\in S_3.
\end{equation}

Note that
\begin{equation} \label{eq:WpWpi}
\tr[M_+ W_\pi] =
\begin{cases}
2d,\quad \pi=(), \\
2d^2,\quad \pi=(12),(23),\text{ or } (31), \\
d(d^2+1),\quad \pi=(123) \text{ or }(132).
\end{cases}
\end{equation}
So the value of $\tr[M_+W_\pi]$ only depends on the cycle structure of $\pi$. As a result, $\tr[O_+W_\pi]$ should only depend on the cycle structure of $\pi$.

Without loss of generality, we set $O_+$ to be the following form
\begin{equation} \label{eq:Opform}
O_+ = \sum_{\vec{a}\in\mbb{Z}_d^3} O(\vec{a}) \ketbra{\vec{a}}{\vec{a}}  = \sum_{\vec{a}\in\mbb{Z}_d^3} O_{wt(\vec{a})} \ketbra{\vec{a}}{\vec{a}},
\end{equation}
where $wt(\vec{a})$ denotes the weight, i.e., number of the same elements, in $\vec{a}$. When $t=3$, the classification of elements in $\mbb{Z}_d^t$ by the weights is sufficient to describe the inner-product $\tr[\vec{a}W_\pi]$. In higher-order case, however, one has to introduce the partition number $\lambda(\vec{a})$, which will be discussed in Section~\ref{Sec:suppProof}.

From Eqs.~\eqref{eq:OpWp},\eqref{eq:WpWpi}, and \eqref{eq:Opform}, we can construct $3$ independent equations for three parameters $\{O_1,O_2,O_3\}$
\begin{equation}
\begin{aligned}
d O_3 + 3d(d-1) O_2 + d(d-1)(d-2)O_1 &= 2d, \\
d O_3 + d(d-1) O_2 &= 2d^2, \\
d O_3 &= d(d^2+1).
\end{aligned}
\end{equation}

Solving the equations, we have
\begin{equation}
O_1 = 2,\quad O_2= 1-d,\quad O_3= 1+d^2.
\end{equation}
To express it in a concise form, $O_{wt} = 1 + (-d)^{wt-1}$. Therefore
\begin{equation}
\begin{aligned}
O_+ &= \sum_{\vec{a}\in\mbb{Z}_d^3} \left[1 + (-d)^{wt(\vec{a})-1}\right] \ketbra{\vec{a}}{\vec{a}} \\
&= \sum_{\vec{a}\in\mbb{Z}_d^3} \left[ \alpha\delta_{a_1 a_2 a_3}+ \beta(\delta_{a_1 a_2} + \delta_{a_2 a_3}+\delta_{a_3 a_1})+\gamma \right] \ketbra{\vec{a}}{\vec{a}}.
\end{aligned}
\end{equation}
Here, $\alpha=(d+1)(d+2), \beta=-(d+1), \gamma=2$.

\subsection{Negativity-moment observable based on global random unitaries and local measurement} \label{SSec:ZZnegativity}

As a first trial, we try to construct the bi-partite observable $O_{neg}\in\mc{L}((\mc{H}^{AB})^{\otimes 3})$ on the local computational basis of $A$ and $B$,
\begin{equation} \label{eq:OnegABtrial}
O_{neg}^{AB} = \sum_{\vec{a},\vec{b}\in \mbb{Z}_d^3} O(\vec{a},\vec{b}) \ket{\vec{a}}_A\bra{\vec{a}} \otimes \ket{\vec{b}}_B\bra{\vec{b}},
\end{equation}
such that, with local independent random unitary twirling on system $A$ and $B$ independently, we obtain
\begin{equation} \label{eq:phiAphiBOneg}
(\Phi^3_A \otimes \Phi^3_B)(O_{neg}^{AB}) = M_{neg} := W^A_{0}\otimes W^B_{1} + W^A_{1}\otimes W^B_{0}.
\end{equation}
Note that, $M_{neg}$ is invariant under local unitary twirling, $(\Phi^3_A \otimes \Phi^3_B)(M_{neg}) = M_{neg}$. We can generalize Proposition~\ref{prop:Wbasis} to the following local unitary twirling form.
\begin{prop}\label{prop:WABbasis}
For two operators $O,P\in\mc{L}((\mc{H}^{AB})^{\otimes t})$, the following two statements are equivalent,
\begin{enumerate}
\item $(\Phi^t_A\otimes \Phi^t_B)(O) = (\Phi^t_A\otimes \Phi^t_B)(P)$,
\item $\tr[O(W_\pi^A\otimes W_\alpha^B)] = \tr[P(W_\pi^A\otimes W_\alpha^B)], \forall \pi,\alpha\in S_t$.
\end{enumerate}
\end{prop}
\begin{proof}
To prove $1\Rightarrow 2$, we have
\begin{equation}
\begin{aligned}
&(\Phi^t_A\otimes \Phi^t_B)(O) = (\Phi^t_A\otimes \Phi^t_B)(P), \\
\Rightarrow\; & \tr[(\Phi^t_A\otimes \Phi^t_B)(O)(W_\pi^A\otimes W_\alpha^B)] = \tr[(\Phi^t_A\otimes \Phi^t_B)(P)(W_\pi^A\otimes W_\alpha^B)], \forall \pi,\sigma\in S_t,\\
\Rightarrow\; & \tr[O(\Phi^t_A\otimes \Phi^t_B)(W_\pi^A\otimes W_\alpha^B)] = \tr[P(\Phi^t_A\otimes \Phi^t_B)(W_\pi^A\otimes W_\alpha^B)], \forall \pi,\sigma\in S_t, \\
\Rightarrow\; & \tr[O (W_\pi^A\otimes W_\alpha^B)] = \tr[P (W_\pi^A\otimes W_\alpha^B)], \forall \pi,\sigma\in S_t.
\end{aligned}
\end{equation}
Here, the second $\Rightarrow$ is because $(\Phi^t_A\otimes \Phi^t_B)(\cdot)$ is a Hermitian-preserving map. The third $\Rightarrow$ is due to the invariance of $W_\pi^{A(B)}$ under the twirling operation.

To prove $2\Rightarrow 1$, we have
\begin{equation}
\begin{aligned}
(\Phi^t_A\otimes\Phi^t_B)(O) &= \sum_{\pi,\sigma,\alpha,\beta} c_{\pi,\sigma} c_{\alpha,\beta} \tr[O(W_\pi^A\otimes W_\alpha^B)] W_\sigma^A\otimes W_\beta^B \\
&= \sum_{\pi,\sigma,\alpha,\beta} c_{\pi,\sigma} c_{\alpha,\beta} \tr[P(W_\pi^A\otimes W_\alpha^B)] W_\sigma^A\otimes W_\beta^B \\
&= (\Phi^t_A\otimes\Phi^t_B)(P).
\end{aligned}
\end{equation}
Here, in the second equality, we have used the statement $2$.
\end{proof}

It seems that, we only need to construct $O_{neg}$ such that
\begin{equation} \label{eq:OnegWnegbasis}
\tr[O_{neg}(W_\pi^A\otimes W_\alpha^B)] = \tr[M_{neg}(W_\pi^A\otimes W_\alpha^B)], \forall \pi,\alpha\in S_3.
\end{equation}

Nevertheless, in the following proposition, we show that the construction above is impossible.
\begin{prop} \label{prop:localOnogo} (Proposition~\ref{Pro:noLocal} in the main text)
It is impossible to construct a computational basis observable $O_{neg}\in\mc{L}((\mc{H}^{AB})^{\otimes 3})$ with the form of Eq.~\eqref{eq:OnegABtrial}, such that one can obtain $M_{neg}$ with independent local unitary twirling $(\Phi^3_A \otimes \Phi^3_B)$ on systems $A$ and $B$.
\end{prop}
\begin{proof}
Suppose we can construct $O_{neg} = \sum_{\vec{a},\vec{b}} O(\vec{a},\vec{b}) \ket{\vec{a}}_A\bra{\vec{a}} \otimes \ket{\vec{a}}_B\bra{\vec{a}}$ such that $(\Phi^t_A\otimes \Phi^t_B)(O_{neg}) = M_{neg}$, then Eq.~\eqref{eq:OnegWnegbasis} holds.

Now we consider the case when $\pi,\sigma$ is $(123)$ or $(132)$. We have
\begin{equation}
\begin{aligned}
\tr[O_{neg}(W_0^A\otimes W_0^B)] &= \sum_{\vec{a},\vec{b}\in\mbb{Z}_d^3} O(\vec{a},\vec{b}) \tr[\ket{\vec{a}}_A\bra{\vec{a}} W_0^A] \tr[\ket{\vec{b}}_A\bra{\vec{b}} W_0^B] \\
&= \sum_{\vec{a},\vec{b}\in\mbb{Z}_d^3} O(\vec{a},\vec{b}) \tr[\ket{\vec{a}}_A\bra{\vec{a}} W_0^A] \tr[\ket{\vec{b}}_A\bra{\vec{b}} W_0^B]^T \\
&= \sum_{\vec{a},\vec{b}\in\mbb{Z}_d^3} O(\vec{a},\vec{b}) \tr[\ket{\vec{a}}_A\bra{\vec{a}} W_0^A] \tr[\ket{\vec{b}}_A\bra{\vec{b}} W_1^B] \\
&= \tr[O_{neg}(W_0^A\otimes W_1^B)],
\end{aligned}
\end{equation}
where the second equality is due to the invariance of the transpose operation on a number. Similarly, we can prove that $\tr[O_{neg}(W_0^A\otimes W_1^B)] = \tr[O_{neg}(W_1^A\otimes W_0^B)] = \tr[O_{neg}(W_1^A\otimes W_1^B)]$. On the other hand, the projection values of $M_{neg}$ on the permutation basis $W_0^A\otimes W_0^B$ and $W_0^A\otimes W_1^B$ are
\begin{equation}
\begin{aligned}
\tr[M_{neg}W_0^A\otimes W_0^B] &= \tr[W_0^A W_0^A]\tr[W_1^B W_0^B] + \tr[W_1^A W_0^A]\tr[W_0^B W_0^B] = 2d^4, \\
\tr[M_{neg}W_0^A\otimes W_1^B] &= \tr[W_0^A W_0^A]\tr[W_1^B W_1^B] + \tr[W_1^A W_0^A]\tr[W_0^B W_1^B] = d^2 + d^6 \neq \tr[M_{neg}W_0^A\otimes W_0^B].
\end{aligned}
\end{equation}
Consequently, Eq.~\eqref{eq:OnegWnegbasis} cannot hold. Therefore, no legal $O_{neg}$ exists.
\end{proof}

Although $M_{neg}$ cannot be directly constructed by local random unitaries, we notice that
\begin{equation}
(\Phi^3_A\otimes \Phi^3_B)(O_+^A\otimes O_+^B) = M_+^A\otimes M_+^B = M_{neg}^{AB} + M_{+}^{AB},
\end{equation}
where $M_{+}^{AB}:= W_0^A\otimes W_0^B + W_1^A\otimes W_1^B$ is the global 3-order purity operator, which can be constructed by global random unitary twirling,
\begin{equation}
\Phi^3_{AB} (O_{+}^{AB}) = M_{+}^{AB},
\end{equation}
with
\begin{equation}
O_{+}^{AB} = \sum_{\vec{c}\in\mbb{Z}_{d^2}^3} O(\vec{c}) \ket{\vec{c}}_{AB} \bra{\vec{c}} = \sum_{\vec{c}\in\mbb{Z}_{d^2}^3} \left[1 + (-d)^{wt(\vec{c})-1} \right] \ket{\vec{c}}_{AB} \bra{\vec{c}}.
\end{equation}
Here, $\{\ket{\vec{c}}\} = \{\ket{\vec{a}}\otimes\ket{\vec{b}}\}$ is the relabelled computational basis of system $A$ and $B$. Then
\begin{equation}
M_{neg}^{AB} =  \Phi^3_A(O_+^A) \otimes \Phi^3_B(O_+^B)  - \Phi^{3}_{AB}(O_{+}^{AB}).
\end{equation}

In the experiment, we first estimate $O_+^A \otimes O_+^B$ with local random unitaries, and then estimate $O_{+}^{AB}$ with global random unitaries, and the estimation of negativity-moment is generated by the difference of them.

\subsection{Negativity observable based on local random unitaries and Bell-state measurement} \label{SSec:BSMnegativity}

From Proposition~\ref{prop:localOnogo}, we have already known that, with local computational basis measurement and local random unitaries, one cannot construct the $M_{neg}$ operator. This is due to the intrinsic parity symmetry of $W_0$ and $W_1$ on the computational basis
\begin{equation}
    \tr\left[W_0 \ketbra{s,s',s''}{s,s',s''}\right] = \tr\left[W_1 \ketbra{s,s',s''}{s,s',s''}\right] = \delta_{s,s',s''}, \forall s,s',s'' = 0,1,...,d-1.
\end{equation}

As a result, for bipartite system, the actions of $\{W_i^A \otimes W_j^B\}_{i,j=0,1}$ on the tensor-ed computational basis are always the same,
\begin{equation}
\begin{aligned}
    &\tr\left[(W_i^A \otimes W_j^B) \ket{s,s',s''}_A\bra{s,s',s''} \otimes \ket{t,t',t''}_B\bra{t,t',t''} \right] \\
    &= \tr\left[W_i^A \ket{s,s',s''}_A\bra{s,s',s''}\right] \tr\left[W_j^B \ket{t,t',t''}_B\bra{t,t',t''}\right]\\
    &= \delta_{s,s',s''}\delta_{t,t',t''}, \forall s,s',s'',t,t',t'' \in 0,1,...,d-1.
\end{aligned}
\end{equation}
which is irrelevant to the value of $i$ and $j$.

However, if we correlate the basis of $A$ and $B$, then the actions of  $W_i^A\otimes W_j^B$ can be dependent on cyclic direction $i$ and $j$.  For example, the action of $W_i^A\otimes W_j^B$ on the Bell state $\ket{\Psi_+}_{AB} := \dfrac{1}{\sqrt{d}}\sum_{s=0}^{d-1} \ket{s,s}_{AB}$ is given by
\begin{equation} \label{eq:WAWBbell}
\tr\left[ (W_i^A\otimes W_j^B) (\Psi_{AB} \otimes\Psi'_{AB} \otimes \Psi''_{AB} ) \right] =
    \begin{cases}
    d^3, &\quad i=j \\
    d,   &\quad i\neq j.
    \end{cases}
\end{equation}
Here $\Psi:= \ket{\Psi}\bra{\Psi}$. Eq.~\eqref{eq:WAWBbell} implies that, it is possible to realize the negativity measurement using a single-copy of $\rho_{AB}$, local random unitaries on $A$ and $B$, assisted with Bell-state measurement (BSM).

Suppose now we have three copies of a given two qudit system $\mc{H}^A\otimes \mc{H}^B$. Our aim is to construct an observable $O_{--}^{AB}$ on the tensor-ed Bell-diagonal basis,
\begin{equation} \label{eq:OBell}
    O_{--} = \sum_{u_1,u_2,u_3; v_1,v_2,v_3= 0}^{d-1} O(u_1,u_2,u_3;v_1,v_2,v_3) \Psi_{u_1,v_1} \otimes \Psi_{u_2,v_2} \otimes \Psi_{u_3,v_3},
\end{equation}
such that
\begin{equation} \label{eq:OBellWAWB}
    \Phi^3_A \otimes \Phi^3_B (O_{--}) = M_{--}^{AB} := M_-^A \otimes M_-^B=(W_0^A - W_1^A)\otimes (W_0^B - W_1^B).
\end{equation}
Note that $M_-^A \otimes M_-^B$ is a Hermitian operator, hence it can be constructed using an observable. As a comparison, $M_-$ on a single system is non-Hermitian.

Recall that we have introduced observable $O_+^A$ such that $\Phi^3_A(O_+)= M_+^A$. Combine this with $O_{--}$, we can construct the negativity operator,
\begin{equation}
    M^{AB}_{neg} := W^A_0 W^B_1 + W^A_1 W^B_0 = \frac{1}{2}\left( M^A_+ \otimes M^B_+ - M^A_- \otimes M^B_- \right) =\frac{1}{2} \Phi^3_A \otimes \Phi^3_B \left( O_+ \otimes O_+ - O_{--} \right).
\end{equation}

\begin{prop} (Proposition~\ref{prop:MmmBSM} in the main text)
$M_{--}^{AB}$ in Eq.~\eqref{eq:OBellWAWB} can be realized with bi-local random unitary twirling, if the final computational basis measurement is substituted with the Bell state measurement between $A$ and $B$, i.e., there exist Bell-basis observable $O_{--}$ such that $(\Phi^3_A\otimes \Phi^3_B)(O_{Bell})= M_{--}^{AB}$.
\end{prop}
\begin{proof}
Based on Proposition~\ref{prop:WABbasis}, Eq.~\eqref{eq:OBellWAWB} is equivalent to
\begin{equation} \label{eq:OBellWAWBterm}
    \tr\left[ O_{--} ( W^A_\pi \otimes W^B_\sigma ) \right] = \tr\left[ (M^A_- \otimes M^B_-)( W^A_\pi \otimes W^B_\sigma ) \right], \quad \forall \pi,\sigma \in S^3.
\end{equation}

The right-hand side (RHS) of Eq.~\eqref{eq:OBellWAWBterm} is easy to be solved,
\begin{equation} \label{eq:WAWBterm}
\begin{aligned}
    & \tr\left[ (M^A_- \otimes M^B_-)( W^A_\pi \otimes W^B_\sigma ) \right] = \tr\left[(W^A_0 - W^A_1) W^A_\pi \right] \tr\left[(W^B_0 - W^B_1) W^B_\sigma \right] \\
    &=
    \begin{cases}
    d^2 (d^2-1)^2, &\quad W_\pi^A = W_\sigma^B = W_0 \text{ or } W_\pi^A = W_\sigma^B= W_1, \\
    -d^2 (d^2-1)^2, &\quad W_\pi^A = W_0, W_\sigma^B = W_1 \text{ or } W_\pi^A = W_1, W_\sigma^B = W_0, \\
    0, &\quad \text{otherwise}.
    \end{cases}
\end{aligned}
\end{equation}

To construct a Bell-diagonal $O_{--}$, we first figure out the effect of $(W^A_\pi \otimes W^B_\sigma)$ on a given tensor-ed Bell state $\Psi_{u_1,v_1} \otimes \Psi_{u_2,v_2} \otimes \Psi_{u_3,v_3}$. We have
\begin{equation} \label{eq:innerWAWBBell}
\begin{aligned}
    &\phi(\vec{u},\vec{v};\pi,\sigma) = \tr\left[ (\Psi_{u_1,v_1} \otimes \Psi_{u_2,v_2} \otimes \Psi_{u_3,v_3}) ( W^A_\pi \otimes W^B_\sigma ) \right] \\
    &= \tr\left\{ (\Psi_+ \otimes \Psi_+ \otimes \Psi_+)\left[ P_{B1}^\dag \otimes P_{B2}^\dag \otimes P_{B3}^\dag \right]  ( W^A_\pi \otimes W^B_\sigma ) \left[ P_{B1} \otimes P_{B2} \otimes P_{B3} \right] \right\} \\
    &= \tr\left\{ (\Psi_+ \otimes \Psi_+ \otimes \Psi_+)\left[ \left( P_{B1}^\dag P_{B\sigma(1)} \right) \otimes \left( P_{B2}^\dag P_{B\sigma(2)} \right) \otimes \left( P_{B3}^\dag P_{B\sigma(3)} \right)\right] ( W^A_\pi \otimes W^B_\sigma ) \right\} \\
    &= \tr\left\{ (W^A_\pi \otimes I^B) (\Psi_+ \otimes \Psi_+ \otimes \Psi_+)\left[ \left( P_{B1}^\dag P_{B\sigma(1)} \right) \otimes \left( P_{B2}^\dag P_{B\sigma(2)} \right) \otimes \left( P_{B3}^\dag P_{B\sigma(3)} \right)\right] ( I^A \otimes W^B_\sigma ) \right\} \\
    &= \tr\left\{ (\Psi_+ \otimes \Psi_+ \otimes \Psi_+)\left[ \left( P_{B1}^\dag P_{B\sigma(1)} \right) \otimes \left( P_{B2}^\dag P_{B\sigma(2)} \right) \otimes \left( P_{B3}^\dag P_{B\sigma(3)} \right)\right] ( I^A \otimes W^B_{\sigma\pi^{-1}} ) \right\} \\
    &= \frac{1}{d^3}\tr\left[ \left( P_{B1}^\dag P_{B\sigma(1)} \right) \otimes \left( P_{B2}^\dag P_{B\sigma(2)} \right) \otimes \left( P_{B3}^\dag P_{B\sigma(3)} \right)  W^B_{\sigma\pi^{-1}} \right] \\.
\end{aligned}
\end{equation}
Here, $P_{Bi}:= P_{B}(u_i,v_i)$, $\vec{u}:=(u_1,u_2,u_3)$, and $\vec{v}:=(v_1,v_2,v_3)$. The forth equal sign is due to the transpose property of Bell state,
\begin{equation}
    (M^A \otimes I^B) \Psi_+ = (I^A \otimes (M^T)^B) \Psi_+.
\end{equation}

In Fig.~\ref{fig:simpPhi} we draw the simplification procedure based on the tensor network graph.

\begin{figure}[htbp]
\centering

\subfigure[]{
\begin{minipage}[t]{0.4\linewidth}
\centering
\includegraphics[width=2.5in]{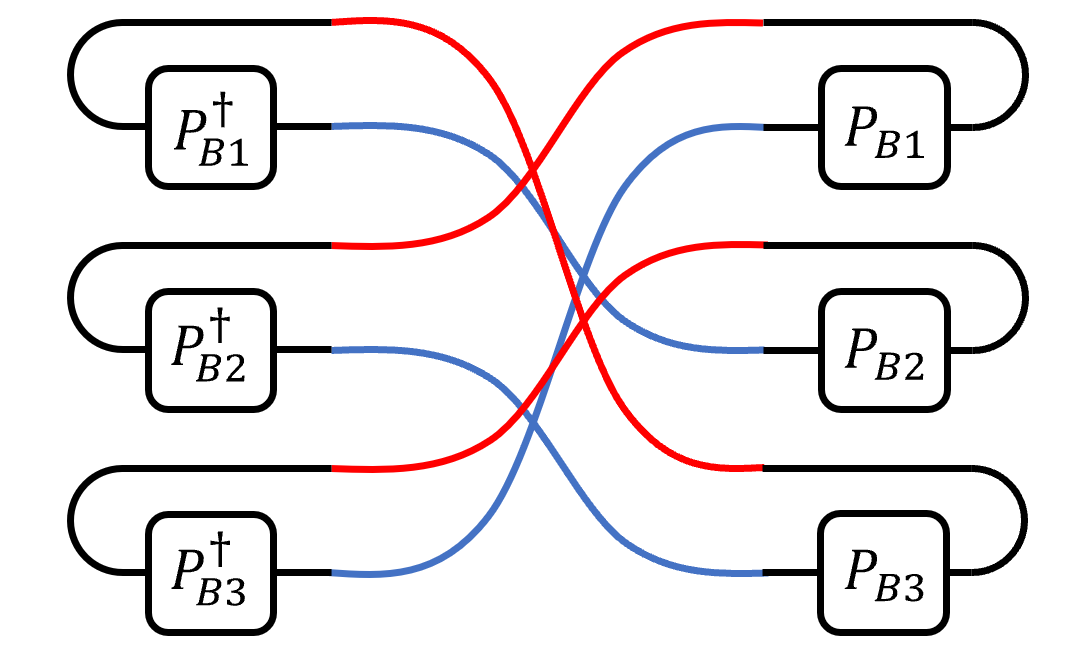}
\end{minipage}%
}%
\subfigure[]{
\begin{minipage}[t]{0.4\linewidth}
\centering
\includegraphics[width=2.5in]{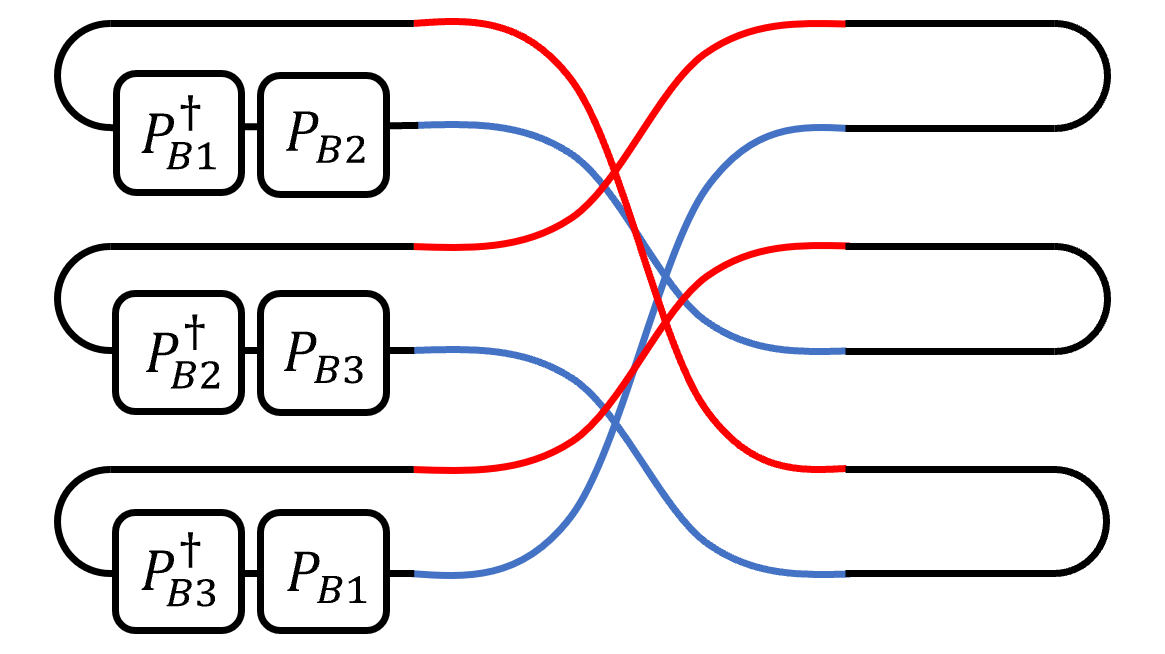}
\end{minipage}%
}%

\subfigure[]{
\begin{minipage}[t]{0.4\linewidth}
\centering
\includegraphics[width=2.5in]{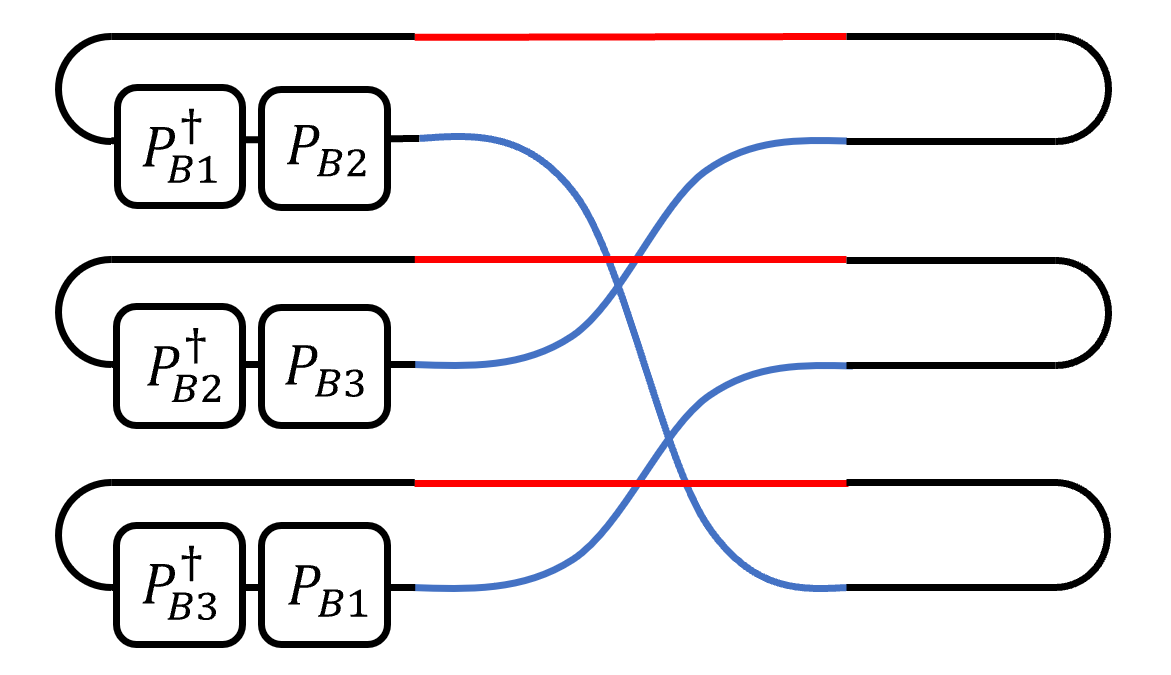}
\end{minipage}
}%
\subfigure[]{
\begin{minipage}[t]{0.4\linewidth}
\centering
\includegraphics[width=2.5in]{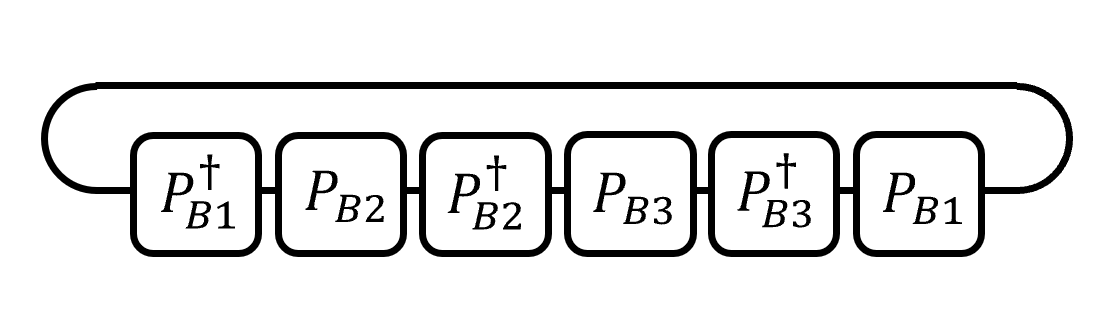}
\end{minipage}
}%

\centering
\caption{The simplification procedure of solving $\phi(\vec{u},\vec{v};\pi,\sigma)$ in Eq.~\eqref{eq:innerWAWBBell}. Here we take $\pi=(132)$ and $\sigma=(123)$ for example.} \label{fig:simpPhi}
\end{figure}

We now summarize the values of $d^3\,\phi(\vec{u},\vec{v};\pi,\sigma)$ in Table \ref{tab:coeffBell},
\begin{table}[htbp]
\begin{tabular}{|c|cccccc|}
\hline
 $W_\sigma^B \backslash W_\pi^A$  & () & (2,3) & (1,3) & (1,2) & (1,2,3) & (1,3,2) \\ \hline
() & $d^3$ & $d^2$ & $d^2$ & $d^2$ & $d$ & $d$ \\
(2,3) & $d^2$ & $\tr[P_2^\dag P_3]\tr[P_3^\dag P_2]$ & $d$ & $d$ & $\tr[P_2^\dag P_3]\tr[P_3^\dag P_2]$ & $\tr[P_2^\dag P_3]\tr[P_3^\dag P_2]$ \\
(1,3) & $d^2$ & $d$ & $\tr[P_1^\dag P_3]\tr[P_3^\dag P_1]$ & $d$ & $\tr[P_1^\dag P_3]\tr[P_3^\dag P_1]$ & $\tr[P_1^\dag P_3]\tr[P_3^\dag P_1]$ \\
(1,2) & $d^2$ & $d$ & $d$ & $\tr[P_1^\dag P_2]\tr[P_2^\dag P_1]$ & $\tr[P_1^\dag P_2]\tr[P_2^\dag P_1]$  & $\tr[P_1^\dag P_2]\tr[P_2^\dag P_1]$ \\
(1,2,3) & $d$ & $\tr[P_2^\dag P_3]\tr[P_3^\dag P_2]$ & $\tr[P_1^\dag P_3]\tr[P_3^\dag P_1]$ & $\tr[P_1^\dag P_2]\tr[P_2^\dag P_1]$ & $\tr[P_1^\dag P_2]\tr[P_2^\dag P_3]\tr[P_3^\dag P_1]$ & $\tr[P_1^\dag P_2 P_3^\dag P_1 P_2^\dag P_3]$ \\
(1,3,2) & $d$ & $\tr[P_2^\dag P_3]\tr[P_3^\dag P_2]$ & $\tr[P_1^\dag P_3]\tr[P_3^\dag P_1]$ & $\tr[P_1^\dag P_2]\tr[P_2^\dag P_1]$ & $\tr[P_1^\dag P_3 P_2^\dag P_1 P_3^\dag P_2]$ & $\tr[P_1^\dag P_3]\tr[P_2^\dag P_1]\tr[P_3^\dag P_2]$ \\ \hline
\end{tabular}
\caption{The coefficients of Bell diagonal state projection on permutation operators $d^3 \phi(\vec{u},\vec{v};\pi,\sigma)$. Here, $P_i := P(u_i,v_i)$.} \label{tab:coeffBell}
\end{table}

The five independent coefficients of $\phi(\vec{u},\vec{v};\pi,\sigma)$ in Table \ref{tab:coeffBell} are $\tr[P_1^\dag P_2], \tr[P_2^\dag P_3], \tr[P_3^\dag P_1], \tr[P_1^\dag P_3 P_2^\dag P_1 P_3^\dag P_2]$ and $\tr[P_1^\dag P_2 P_3^\dag P_1 P_2^\dag P_3]$. Based on the weight of three Pauli operators, i.e., the number of the same Pauli operators, we list the value of these coefficients in Table \ref{tab:coeffvaluequdit},
\begin{table}[htbp]
\begin{tabular}{|c|ccccc|}
\hline
 Terms  & $\tr[P_1^\dag P_2]$ & $\tr[P_2^\dag P_3]$ & $\tr[P_3^\dag P_1]$ & $\tr[P_1^\dag P_3 P_2^\dag P_1 P_3^\dag P_2]$ & $\tr[P_1^\dag P_2 P_3^\dag P_1 P_2^\dag P_3]$ \\ \hline
$P P P$ ($P_1 = P_2 = P_3$) & $d$ & $d$ & $d$ & $d$ & $d$ \\
$P P Q$ ($P_1 = P_2 \neq P_3$) & $d$ & $0$ & $0$ & $d$ & $d$  \\
$P Q P$ ($P_1 = P_3 \neq P_2$) & $0$ & $d$ & $0$ & $d$ & $d$  \\
$Q P P$ ($P_2 = P_3 \neq P_1$) & $0$ & $0$ & $d$ & $d$ & $d$  \\
$P Q R$ ($P_1, P_2, P_3$ all different) & $0$ & $0$ & $0$ & $d \exp{\left[-\frac{2\pi}{d} \vec{I}\cdot (\vec{u}\times \vec{v} )\right]}$ & $d \exp{\left[\frac{2\pi}{d} \vec{I}\cdot (\vec{u}\times \vec{v} )\right]}$  \\
\hline
\end{tabular}
\caption{The coefficients of $\phi(\vec{u},\vec{v};\pi,\sigma)$ depending on the weight of $P_i$. Here, $P_i := P(u_i,v_i), \vec{I}:= (1,1,1), \vec{u}:=(u_1,u_2,u_3), \vec{v}:=(v_1,v_2,v_3)$.} \label{tab:coeffvaluequdit}
\end{table}

Recall that the coefficients of $M^A_-\otimes M^B_-$ in Eq.~\eqref{eq:WAWBterm} is invariant under the cyclic operations $W^A_0, W^A_1, W^B_0$ and $W^B_1$, therefore
\begin{equation}
\begin{aligned}
    & \tr[(M^A_-\otimes M^B_-)(W^A_{(1,2)}\otimes W^B_{(1,2)})] \\
    &= \tr[(W^A_0 \otimes I)(M^A_-\otimes M^B_-)(W^A_0 \otimes I)^\dag (W^A_{(1,2)}\otimes W^B_{(1,2)})] \\
    &= \tr[(M^A_-\otimes M^B_-)(W^A_0 \otimes I)^\dag (W^A_{(1,2)}\otimes W^B_{(1,2)})(W^A_0 \otimes I)] \\
    &= \tr\left[(M^A_-\otimes M^B_-) \left(W^A_{(2,3)}\otimes W^B_{(1,2)}\right)\right].
\end{aligned}
\end{equation}
Similarly, we can show that the inner product of $(M^A_-\otimes M^B_-)$ and any swapping operations should be the same. Therefore, based on Table \ref{tab:coeffvaluequdit}, without loss of generality, we may assume all the terms with the form $PPQ$ ($P_1=P_2\neq P_3$) own the same coefficient $O(\vec{u},\vec{v})$ as the corresponding terms $PQP$ and $QPP$.

Define
\begin{equation}
wt(\vec{u},\vec{v}):= \{\text{numbers of the same index pairs in } (u_1,v_1), (u_2,v_2), (u_3,v_3)\}.
\end{equation}
Therefore, the element number of $(\vec{u},\vec{v})$ with $wt(\vec{u},\vec{v})=3,2,1$ is $d^2, 3 d^2(d^2-1)$ and $d^2(d^2-1)(d^2-2)$, respectively.

For the elements with $wt(\vec{u},\vec{v})=1$, we define the rotation angle
\begin{equation}
\theta(\vec{u},\vec{v}):= \vec{I}\cdot (\vec{u}\times \vec{v}) = (u_1 v_2 + u_2 v_3 + u_3 v_1) - (u_1 v_3 + u_2 v_1 + u_3 v_2),
\end{equation}
where the multiplication and addition is defined on the integer ring $\mathbb{Z}_d$. We also assume that the elements with the same value of $\theta$ share the same coefficients. The coefficients $\phi(\vec{u},\vec{v};\pi,\sigma)$ is now reduced to $\phi(wt,\theta;\pi,\sigma)$, which is only related to the weight $wt$ and rotation angle $\theta$. Note that, for the elements with $wt=3$ and $wt=2$, rotation angle $\theta=0$.

Based on the value of $wt$ and $\theta$, we simplify $O_{--}$ to be in the form
\begin{equation} \label{eq:OBellreduced}
\begin{aligned}
    O_{--} &= O(3) \sum_{wt(\vec{u},\vec{v})=3} \Psi_{u_1,v_1}\otimes \Psi_{u_2,v_2}\otimes \Psi_{u_3,v_3} + O(2) \sum_{wt(\vec{u},\vec{v})=2} \Psi_{u_1,v_1}\otimes \Psi_{u_2,v_2}\otimes \Psi_{u_3,v_3} \\
    &\quad +\sum_{wt(\vec{u},\vec{v})=1} O(1;\theta(\vec{u},\vec{v})) \Psi_{u_1,v_1}\otimes \Psi_{u_2,v_2}\otimes \Psi_{u_3,v_3}.
\end{aligned}
\end{equation}

Combined with Eqs.~\eqref{eq:WAWBterm}, ~\eqref{eq:innerWAWBBell} and \eqref{eq:OBellreduced}, Eq.~\eqref{eq:OBellWAWB} is now reduced to
\begin{equation} \label{eq:OBellreduced2}
\begin{aligned}
    \sum_{wt=2}^3 & \phi(wt,0;\pi,\sigma) O(wt,0)n(wt,0) + \sum_{\theta=0}^{d-1} \phi(1,\theta;\pi,\sigma) O(1,\theta)n(1,\theta) \\
    &=
    \begin{cases}
    d^2 (d^2-1)^2, &\quad W_\pi^A = W_\sigma^B = W_0 \text{ or } W_\pi^A = W_\sigma^B= W_1, \\
    -d^2 (d^2-1)^2, &\quad W_\pi^A = W_0, W_\sigma^B = W_1 \text{ or } W_\pi^A = W_1, W_\sigma^B = W_0, \\
    0, &\quad \text{otherwise}.
    \end{cases}
\end{aligned}
\end{equation}
Where $n(wt_0,\theta_0)$ is the number of index pairs $(\vec{u},\vec{v})$ with $wt(\vec{u},\vec{v})=wt_0$ and $\theta(\vec{u},\vec{v})=\theta_0$. We know that $n(3,0) = d^2, n(2,0) = 3d^2(d^2-1)$, and $\sum_\theta n(1,\theta) = d^2(d^2-1)(d^2-2)$. Moreover, since
\begin{equation}
\theta(\vec{u},\vec{v}) = -\theta(\vec{v},\vec{u}),
\end{equation}
we have $n(1,\theta) = n(1,-\theta)$.

Note that $\phi(1,\theta;\vec{u},\vec{v}) = \phi(1,-\theta;\vec{u},\vec{v})^*$. Hereafter, we set $O(1,\theta)=O(1,-\theta)$ to be a real number.

To further simplify Eq.~\eqref{eq:OBellreduced2}, we suppose that:
\begin{enumerate}
\item When $n$ is even, $O(1,\theta)=0$ if $\theta\neq 0$ or $d/2$.
\item When $n$ is odd, $O(1,\theta)=0$ if $\theta\neq 0, (d-1)/2$ or $(d+1)/2$.
\end{enumerate}

Then
\begin{equation}
\sum_{\theta\neq 0} \phi(1,\theta;\vec{u},\vec{v}) O(1,\theta)n(1,\theta) =
\begin{cases}
\phi(1,d/2;\vec{u},\vec{v})O(1,d/2)n(1,d/2), \quad d\text{ is even}, \\
2Re[\phi(1,(d+1)/2;\vec{u},\vec{v})]O(1,(d+1)/2)n(1,(d+1)/2), \quad d\text{ is odd}.
\end{cases}
\end{equation}

We denote
\begin{equation}
Q(1,d/2) =
\begin{cases}
O(1,d/2)n(1,d/2), \quad d\text{ is even}, \\
2O(1,(d+1)/2)n(1,(d+1)/2), \quad d\text{ is odd},
\end{cases}
\end{equation}
Then there are only four unknown parameters $Q(2,0), Q(3,0), Q(1,0)$ and $Q(1,d/2)$.

When $d$ is even, after eliminating all the redundant terms in Eq.~\eqref{eq:OBellreduced2}, we obtain the following linear equation,
\begin{equation} \label{eq:OBellequation}
    \begin{pmatrix}
    d  &  d  &  d  &  d \\
    0  &  0  &  d^2/3  & d^2  \\
    0  &  0  &  0  & d^3  \\
    d  &  -d  &  d  & d  \\
    \end{pmatrix}
    \begin{pmatrix}
    Q(1,0) \\
    Q(1,d/2)  \\
    Q(2,0)  \\
    Q(3,0)  \\
    \end{pmatrix}
    =
    \begin{pmatrix}
    0 \\
    0 \\
    d^5 (d^2 -1)^2  \\
    -d^5 (d^2 -1)^2 \\
    \end{pmatrix}.
\end{equation}

Solve Eq.~\eqref{eq:OBellequation}, we have
\begin{equation} \label{eq:Qresult}
    \begin{pmatrix}
    Q(1,d/2) \\
    Q(1,1)  \\
    Q(2,0)  \\
    Q(3,0)  \\
    \end{pmatrix}
    = d^{2} (d^2 - 1)^2
    \begin{pmatrix}
    -\frac{1}{2}(d^2-4) \\
    \frac{1}{2}d^2 \\
    -3  \\
    1 \\
    \end{pmatrix}.
\end{equation}

Therefore, the observable $O_{--}$ is
\begin{equation} \label{eq:OBellresult}
\begin{aligned}
    O_{--} &= \sum_{wt(\vec{u},\vec{v})=2}^3 \frac{Q(wt,0)}{n(wt,0)}  \Psi_{u_1,v_1}\otimes \Psi_{u_2,v_2}\otimes \Psi_{u_3,v_3}  +\sum_{wt(\vec{u},\vec{v})=1, \theta(\vec{u},\vec{v})=0} \frac{Q(1,0)}{n(1,0)} \Psi_{u_1,v_1}\otimes \Psi_{u_2,v_2}\otimes \Psi_{u_3,v_3} \\
    &\quad + \sum_{wt(\vec{u},\vec{v})=1, \theta(\vec{u},\vec{v})=d/2} \frac{Q(1,d/2)}{n(1,d/2)} \Psi_{u_1,v_1}\otimes \Psi_{u_2,v_2}\otimes \Psi_{u_3,v_3}.
\end{aligned}
\end{equation}

When $d$ is odd, after eliminating all the redundant terms in Eq.~\eqref{eq:OBellreduced2}, we obtain the following linear equation,
\begin{equation} \label{eq:OBellequationOdd}
    \begin{pmatrix}
    d  &  d  &  d  &  d \\
    0  &  0  &  d^2/3  & d^2  \\
    0  &  0  &  0  & d^3  \\
    d  &  d\cos(\frac{d+1}{d}\pi)  &  d  & d  \\
    \end{pmatrix}
    \begin{pmatrix}
    Q(1,0) \\
    Q(1,d/2)  \\
    Q(2,0)  \\
    Q(3,0)  \\
    \end{pmatrix}
    =
    \begin{pmatrix}
    0 \\
    0 \\
    d^5 (d^2 -1)^2  \\
    -d^5 (d^2 -1)^2 \\
    \end{pmatrix}.
\end{equation}

Solve Eq.~\eqref{eq:OBellequation}, we have
\begin{equation} \label{eq:QresultOdd}
    \begin{pmatrix}
    Q(1,d/2) \\
    Q(1,1)  \\
    Q(2,0)  \\
    Q(3,0)  \\
    \end{pmatrix}
    = d^{2} (d^2 - 1)^2
    \begin{pmatrix}
    -\frac{d^2-2-2\cos(\pi/d)}{\cos(\frac{d+1}{d}\pi)-1} \\
    \frac{1}{2}d^2 \sec(\pi/2d)^2 \\
    -3  \\
    1 \\
    \end{pmatrix},
\end{equation}
when $d\gg 1$, the solution in Eq.~\eqref{eq:QresultOdd} becomes the one in Eq.~\eqref{eq:Qresult}.

Therefore, the observable $O_{--}$ is
\begin{equation} \label{eq:OBellresultOdd}
\begin{aligned}
    O_{--} &= \sum_{wt(\vec{u},\vec{v})=2}^3 \frac{Q(wt,0)}{n(wt,0)}  \Psi_{u_1,v_1}\otimes \Psi_{u_2,v_2}\otimes \Psi_{u_3,v_3}  +\sum_{wt(\vec{u},\vec{v})=1, \theta(\vec{u},\vec{v})=0} \frac{Q(1,0)}{n(1,0)} \Psi_{u_1,v_1}\otimes \Psi_{u_2,v_2}\otimes \Psi_{u_3,v_3} \\
    &\quad + \sum_{wt(\vec{u},\vec{v})=1, \theta(\vec{u},\vec{v})=(d\pm 1)/2} \frac{Q(1,d/2)}{n(1,d/2)} \Psi_{u_1,v_1}\otimes \Psi_{u_2,v_2}\otimes \Psi_{u_3,v_3}.
\end{aligned}
\end{equation}

\end{proof}

\section{Statistical analysis} \label{Sec:suppStat}

In this section, we analyze the statistical fluctuation of the estimation process. We start from the variance analysis of the 3-order purity estimation $\tr[\rho^3]$. Then, with similar methods, we consider the statistical fluctuation in the negativity estimation.

Here, we use the symbols with hat $\hat{r}$ to denote a random variable, and corresponding normal font $r$ to denote its value.

\subsection{Statistical analysis of 3-order moment estimation} \label{SSec:stat3mom}

We start from the estimation of 3-moment $\tr[\rho^3]$ for the quantum state on system $A$, i.e., $\rho\in\mc{D}(\mc{H}^A)$. Recall that
\begin{equation} \label{eq:trrho3}
\begin{aligned}
\tr[\rho^3] &= \frac{1}{2} \tr[M_+ (\rho\otimes \rho\otimes \rho) ] \\
&=\frac{1}{2}\sub{\mbb{E}}{U\in \mc{E}} \sum_{\vec{a}\in \mbb{Z}_d^3} \left[1 + (-d)^{wt(\vec{a})-1} \right] P(a_1|U) P(a_2|U) P(a_3|U),
\end{aligned}
\end{equation}
where $\mc{E}$ is a unitary group which forms the unitary 3-design, $\vec{a}=[a_1, a_2, a_3]$ is a 3-dit vector in the ring $\mbb{Z}_d^3:= (\mbb{Z}_d)^{\otimes 3}$, $wt(\vec{a})$ denotes the number of the same values in $\vec{a}$, e.g., $wt([0,2,3])=1$ while $wt([0,1,1])=2$. $P(a|U)$ denotes the probability
\begin{equation}
P(a|U) := \braket{a|U \rho U^\dag|a}.
\end{equation}

To estimate the $3$-order purity, we perform $N_U N_M$ times of experiments in total, and the estimation procedure is as follows,

\underline{$\hat{M}_+ (\mc{E},N_U, N_M):$ Estimator function for $3$-order moment $\tr[\rho^3]$.}
\begin{enumerate}
\item Set $\hat{M}_+ :=0$.
\item For $t=1,2,...,N_U$, Do
  \begin{enumerate}
    \item Randomly choose a unitary $U_i$ from group $\mc{E}$,
    \item Set $N_M$-dit measurement register to zero vector $\hat{V}:=[0,0,...,0]$,
    \item For $i=1,2,...,N_M$, Do
    \begin{enumerate}
      \item Generate a sample: first prepare $\rho$, then apply $U_i$ on it, and then measure it on the computational basis,
      \item Record the measurement result $a$ to register $\hat{V}(i) = a$.
    \end{enumerate}
    \item For $i,j,k=1,2,...,N_M$ and $i<j<k$, Do
    \begin{enumerate}
      \item Set $\hat{M}_+(t) := \frac{1}{2} \binom{N_M}{3}^{-1} \left[1 + (-d)^{wt([\hat{V}(i),\hat{V}(j),\hat{V}(k)])-1} \right]$.
      \item Set $\hat{M}_+ \mathrel{+}= \frac{1}{N_U} \hat{M}_+(t)$.
    \end{enumerate}
  \end{enumerate}
\item Return $\hat{M}_+$.
\end{enumerate}

In the $\hat{M}_+$ experiment, we generate $N_U$ independent estimators $\hat{M}_+(t)$. For each estimator, we first randomly sample $\hat{U}$ from the set $\mc{E}$, and then generate independent variables $\{\hat{a}_{i}\}_{i=1}^{N_M}$ with the conditional probability $P(a|U)$.

For the simplicity of later discussion, we describe the whole experiment process as generating $N_U N_M$ independent estimators $\{\hat{r}_U(i)\}$. The $i$-th estimator is
\begin{equation} \label{eq:rUi}
\hat{r}_U(i) = \ketbra{\hat{a}_i}{\hat{a}_i} := \ketbra{a}{a}, \quad \text{with probability  } P(a|U):= \tr[\ketbra{a}{a}U\rho U^\dag],
\end{equation}
where $U$ is randomly chosen from $\mc{E}$. Therefore $\hat{r}_U(i)$ is the matrix expression of variable $\hat{a}_i$, which is dependent on the preset variable $\hat{U}$.

For a diagonal operator $Q\in\mc{L}(\mc{H}^A)$ such that $Q = \sum_{a\in\mbb{Z}_d} Q(a) \ketbra{a}{a}$, we have
\begin{equation}
\begin{aligned}
\mbb{E}\;\tr[\hat{r}_U(i) Q] &= \sub{\mbb{E}}{U\in\mc{E}} \sum_{a\in\mbb{Z}_d} Q(a) P(a|U) \\
&= \sub{\mbb{E}}{U\in\mc{E}} \sum_{a\in\mbb{Z}_d} Q(a) \tr[U^\dag \ketbra{a}{a} U \rho] \\
&= \sub{\mbb{E}}{U\in\mc{E}} \sum_{a\in\mbb{Z}_d} \tr[U^\dag Q U \rho] \\
&= \tr[\Phi^1(Q)\rho].
\end{aligned}
\end{equation}
Similarly, for $Q\in\mc{L}((\mc{H}^A)^{\otimes 3})$ and three independent estimators $\{\hat{r}_U(i), \hat{r}_U(j), \hat{r}_U(k)\}$, we have
\begin{equation} \label{eq:ErUrUrUQ}
\begin{aligned}
&\mbb{E}\;\tr[\left(\hat{r}_U(i)\otimes \hat{r}_U(j)\otimes \hat{r}_U(k)\right) Q] \\
&= \tr[\Phi^3(Q)\rho\otimes \rho \otimes \rho].
\end{aligned}
\end{equation}

To express the $3$-moment estimator $\hat{M}_+$ with $\{\hat{r}_U(i)\}$, we first split it to several independent estimators
\begin{equation}
\hat{M}_+ = \frac{1}{N_U} \sum_{t=1}^{N_U} \hat{M}_+(t),
\end{equation}
with each estimator a chosen fixed random unitary $U(t)$. Each independent estimator can be further written as
\begin{equation}
\hat{M}_+(t) = \frac{1}{2}\binom{N_M}{3}^{-1} \sum_{i<j<k} \tr[\left(\hat{r}_U(i)\otimes \hat{r}_U(j)\otimes \hat{r}_U(k)\right) O_+].
\end{equation}

The expectation value of $\hat{M}_+(t)$ is
\begin{equation}
\begin{aligned}
\mbb{E}(\hat{M}_+(t)) &= \frac{1}{2}\binom{N_M}{3}^{-1} \mbb{E}\left( \sum_{i<j<k} \tr[\left(\hat{r}_U(i)\otimes \hat{r}_U(j)\otimes \hat{r}_U(k)\right) O_+] \right) \\
&= \frac{1}{2}\binom{N_M}{3}^{-1} \sum_{i<j<k} \tr[\Phi^3(O_+)\rho\otimes \rho\otimes \rho] \\
&= \tr[\rho^3].
\end{aligned}
\end{equation}
The second and third equality are due to Eq.~\eqref{eq:ErUrUrUQ} and \eqref{eq:trrho3}, respectively. Therefore, $\{\hat{M}_+(t)\}$ are unbiased estimators for $\tr[\rho^3]$.

We now calculate the variances of the estimators $\hat{M}_+(t)$ and $\hat{M}_+$.

\begin{prop} \label{prop:variance3orderpurity}
For the estimator $\hat{M}_+(\mc{E},N_U,N_M)$, the variances of $\hat{M}_+(t)$ and $\hat{M}_+$ are
\begin{equation}
\begin{aligned}
&\mb{Var}[\hat{M}_+(t)] = \nu(N_M,d), \\
&\mb{Var}[\hat{M}_+] = \frac{1}{N_U}\nu(N_M,d),
\end{aligned}
\end{equation}
where
\begin{equation} \label{eq:nuNMd}
\begin{aligned}
\nu(N_M,d) := \frac{1}{4}\Gamma_6 + \frac{9}{4}\frac{1}{N_M}\Gamma_5 + \frac{9}{2}\frac{1}{N_M^2}\Gamma_4 + \frac{3}{2}\frac{1}{N_M(N_M-1)(N_M-2)}\Gamma_3 - \tr[\rho^3]^2.
\end{aligned}
\end{equation}
Here the variance terms $\Gamma_3, \Gamma_4, \Gamma_5, \Gamma_6$ are
\begin{equation} \label{eq:Gamma3456}
\begin{aligned}
\Gamma_3(\rho,O_{+}^2,\mc{E}) &:= \tr[\Phi^3(O_{+}^2)\rho^{\otimes 3}], \\
\Gamma_4(\rho,O_{123,124},\mc{E}) &:= \tr[\Phi^4(O_{123,124})\rho^{\otimes 4}], \\
\Gamma_5(\rho,O_{123,145},\mc{E}) &:= \tr[\Phi^5(O_{123,145})\rho^{\otimes 5}], \\
\Gamma_6(\rho,O_{+}^{\otimes 2},\mc{E}) &:= \tr[\Phi^6(O_{+}^{\otimes 2})\rho^{\otimes 6}],
\end{aligned}
\end{equation}
and
\begin{equation} \label{eq:O123145}
\begin{aligned}
O_{123,145} &:= \sum_{\vec{a}\in \mbb{Z}_d^5} O_+(a_1a_2a_3)O_+(a_1a_4a_5) \ketbra{\vec{a}}{\vec{a}}, \\
O_{123,124} &:= \sum_{\vec{a}\in \mbb{Z}_d^4} O_+(a_1a_2a_3)O_+(a_1a_2a_4) \ketbra{\vec{a}}{\vec{a}}.
\end{aligned}
\end{equation}
\end{prop}

\begin{proof}
By the total variance law, the variance $\mb{Var}[\hat{M}_+(t)]$ is
\begin{equation} \label{eq:VarO3t}
\begin{aligned}
\mb{Var}[\hat{M}_+(t)] &= \sub{\mbb{E}}{U} \left[\sub{\mb{Var}}{a} (\hat{M}_+(t)|U) \right] + \sub{\mb{Var}}{U} \left[\sub{\mbb{E}}{a}(\hat{M}_+(t)|U) \right] \\
&= \sub{\mbb{E}}{U} \left[ \sub{\mbb{E}}{a}(\hat{M}_+^2(t)|U) -  \sub{\mbb{E}}{a}(\hat{M}_+(t)|U)^2 \right] +  \sub{\mbb{E}}{U}\left[ \sub{\mbb{E}}{a}(\hat{M}_+(t)|U)^2\right] - \left[\sub{\mbb{E}}{U}\sub{\mbb{E}}{a}(\hat{M}_+(t)|U) \right]^2 \\
&= \sub{\mbb{E}}{U}\left[ \sub{\mbb{E}}{a}(\hat{M}_+^2(t)|U) \right] - \left[\sub{\mbb{E}}{U}\sub{\mbb{E}}{a}(\hat{M}_+(t)|U) \right]^2.
\end{aligned}
\end{equation}
Here, the second term is just $\tr[\rho^3]^2$. We now focus on the calculation of the first term,
\begin{equation} \label{eq:EUEaO2bound1}
\sub{\mbb{E}}{U}\left[ \sub{\mbb{E}}{a}(\hat{M}_+^2(t)|U) \right] = \frac{1}{4}\binom{N_M}{3}^{-2} \sum_{\substack{i<j<k\\ l<m<n}} \left\{ \sub{\mbb{E}}{U}\sub{\mbb{E}}{a}\left\{\tr[\left(\hat{r}_U(i)\otimes \hat{r}_U(j)\otimes \hat{r}_U(k)\right) O_+] \tr[\left(\hat{r}_U(l)\otimes \hat{r}_U(m)\otimes \hat{r}_U(n)\right) O_+] \right\} \right\}.
\end{equation}
Here $\{\hat{r}_U(i),\hat{r}_U(j),\hat{r}_U(k)\}$ ($\{\hat{r}_U(l),\hat{r}_U(m),\hat{r}_U(n)\}$) is a group of independent estimators from the set $\{\hat{r}_U(p)\}_{p=1}^{N_M}$. However, there may be collision (i.e., the same indices) between estimator group $(i,j,k)$ and $(l,m,n)$. Based on the collision number (denoted as $Co[(i,j,k);(l,m,n)]$), the expectation value will be different. We then group the estimators and reduce Eq.~\eqref{eq:EUEaO2bound1} to
\begin{align*} 
&\quad \sub{\mbb{E}}{U}\left[ \sub{\mbb{E}}{a}(\hat{M}_+^2(t)|U) \right] \\
&=\frac{1}{4}\binom{N_M}{3}^{-2} \bigg\{ \sum_{\substack{Co[(i,j,k);\\ (l,m,n)]=0}} \mbb{E}\left\{\tr[\left(\hat{r}_U(i)\otimes \hat{r}_U(j)\otimes \hat{r}_U(k)\right) O_+] \tr[\left(\hat{r}_U(l)\otimes \hat{r}_U(m)\otimes \hat{r}_U(n)\right) O_+] \right\} \\
& \qquad\qquad\qquad\qquad+ \sum_{\substack{Co[(i,j,k);\\ (l,m,n)]=1}} \mbb{E}\left\{\tr[\left(\hat{r}_U(i)\otimes \hat{r}_U(j)\otimes \hat{r}_U(k)\right) O_+] \tr[\left(\hat{r}_U(l)\otimes \hat{r}_U(m)\otimes \hat{r}_U(n)\right) O_+] \right\} \\
& \qquad\qquad\qquad\qquad+ \sum_{\substack{Co[(i,j,k);\\ (l,m,n)]=2}} \mbb{E}\left\{\tr[\left(\hat{r}_U(i)\otimes \hat{r}_U(j)\otimes \hat{r}_U(k)\right) O_+] \tr[\left(\hat{r}_U(l)\otimes \hat{r}_U(m)\otimes \hat{r}_U(n)\right) O_+] \right\}  \\
& \qquad\qquad\qquad\qquad+ \sum_{\substack{Co[(i,j,k);\\ (l,m,n)]=3}} \mbb{E}\left\{\tr[\left(\hat{r}_U(i)\otimes \hat{r}_U(j)\otimes \hat{r}_U(k)\right) O_+] \tr[\left(\hat{r}_U(l)\otimes \hat{r}_U(m)\otimes \hat{r}_U(n)\right) O_+] \right\}  \bigg\} \stepcounter{equation}\tag{\theequation}\label{eq:EUEaO2bound2} \\
= &\frac{1}{4}\binom{N_M}{3}^{-2} \bigg\{ \binom{N_M}{6}\binom{6}{0}\binom{6}{3} \sub{\mbb{E}}{U}\{\tr[(\rho_U\otimes\rho_U\otimes\rho_U) O_+]^2 \} + \binom{N_M}{5}\binom{5}{1}\binom{4}{2} \mbb{E}\left\{\tr[\left(\hat{r}_U(1)\otimes \rho_U \otimes \rho_U \right) O_+]^2 \right\} \\
& \quad+ \binom{N_M}{4}\binom{4}{2}\binom{2}{1} \mbb{E}\left\{\tr[\left(\hat{r}_U(1)\otimes \hat{r}_U(2) \otimes \rho_U \right) O_+]^2 \right\} + \binom{N_M}{3}\binom{3}{3}\binom{0}{0} \mbb{E}\left\{\tr[\left(\hat{r}_U(1)\otimes \hat{r}_U(2) \otimes \hat{r}_U(3) \right) O_+]^2 \right\} \bigg\} \\
= &\frac{1}{4}\binom{N_M}{3}^{-2} \bigg\{ \binom{N_M}{6}\binom{6}{0}\binom{6}{3} \tr[\Phi^6(O_+\otimes O_+) \rho^{\otimes 6}] \} + \binom{N_M}{5}\binom{5}{1}\binom{4}{2} \tr[\Phi^5(O_{123,145})\rho^{\otimes 5}] \\
& \quad+ \binom{N_M}{4}\binom{4}{2}\binom{2}{1} \tr[\Phi^4(O_{123,124})\rho^{\otimes 4}] + \binom{N_M}{3}\binom{3}{3}\binom{0}{0} \tr[\Phi^3(O^2_+)\rho^{\otimes 3}] \bigg\} \\
= &\frac{1}{4}\binom{N_M}{3}^{-1} \bigg\{ \frac{(N_M-3)(N_M-4)(N_M-5)}{6} \Gamma_6 + \frac{3}{2}(N_M -3)(N_M -4)\Gamma_5 + 3(N_M-3)\Gamma_4 + \Gamma_3 \bigg\} \\
= & \frac{1}{4}\Gamma_6 + \frac{9}{4}\frac{1}{N_M}\Gamma_5 + \frac{9}{2}\frac{1}{N_M^2}\Gamma_4 + \frac{3}{2}\frac{1}{N_M(N_M-1)(N_M-2)}\Gamma_3.
\end{align*}
Here, in the second equality, we calculate the collision number of each case and utilize the property that the expectation value only depends on the collision number, and $\rho_U:= U \rho U^\dag$.

The variance of the overall estimator is then upper bounded by
\begin{equation}
\begin{aligned}
\mb{Var}[\hat{M}_+] &= \mbb{E}[\hat{M}_+^2] - \mbb{E}[\hat{M}_+]^2 \\
&= \frac{1}{N_U^2} \sum_{t,q=1}^{N_U} \left\{ \mbb{E}[\hat{M}_+(t)\hat{M}_+(q)] - \mbb{E}[\hat{M}_+(1)]^2 \right\} \\
&= \frac{1}{N_U^2} \sum_{t=1}^{N_U} \left\{ \mbb{E}[\hat{M}_+(t)\hat{M}_+(t)] - \mbb{E}[\hat{M}_+(1)]^2 \right\} \\
&= \frac{1}{N_U} \mb{Var}[\hat{M}_+(1)] = \frac{1}{N_U} \nu(N_M,d).
\end{aligned}
\end{equation}
\end{proof}

Applying the Bernstein's inequality \cite{bernstein1924modification}, we obtain Proposition~\ref{Prop:concentraterho3}, as a concrete concentration result.
\begin{prop} \label{Prop:concentraterho3}
Using the estimator $\hat{M}_+(\mc{E},N_U,N_M)$,
the probability that the deviation $\epsilon$ of estimated $3$-moment $O_3$ from $\tr[\rho^3]$ is bounded by
\begin{equation} \label{eq:concentrationO3}
\Pr[|O_3 - \tr(\rho^3)|\geq \epsilon] \leq 2 \exp\left[-\frac{N_U \epsilon^2}{2\nu(N_M,d) + 2\epsilon/3}\right],
\end{equation}
where $\nu(N_M,d)$ is defined in Eq.~\eqref{eq:nuNMd}.
\end{prop}
\begin{proof}
The Bernstein's inequality states that, for i.i.d. variables $\mb{x}_1,\mb{x}_2,...,\mb{x}_N$ with $\mbb{E}(\mb{x}_i)=0, |\mb{x}_i|\leq \tau$ for all $i=1,2,...,N$, then
\begin{equation} \label{eq:Bernstein}
\begin{aligned}
\Pr\left[\left(\frac{1}{N}\sum_{i=1}^{N} x_i\right)\geq \epsilon \right] &\leq \exp\left[ - \frac{N\epsilon^2}{2\sigma^2 + 2\tau \epsilon/3}\right], \\
\Pr\left[\left(\frac{1}{N}\sum_{i=1}^{N} x_i\right)\leq -\epsilon \right] &\leq \exp\left[ - \frac{N\epsilon^2}{2\sigma^2 + 2\tau \epsilon/3}\right], \\
\end{aligned}
\end{equation}
where $\sigma^2:= \frac{1}{N}\sum_{i=1}^{N}\mb{Var}[\mb{x}_i]$ and $\epsilon>0$.

Now we set $\mb{x}_t := \mb{O}_t - \tr[\rho^3]$, then $\mbb{E}(\mb{x}_t)=0$. We have $|\mb{x}_t|\leq 1$, and $\sigma^2 = \nu(N_M,d)$ from Proposition~\ref{prop:variance3orderpurity}. Apply Eq.~\eqref{eq:Bernstein} twice, we obtain Eq.~\eqref{eq:concentrationO3}.
\end{proof}

From Proposition~\ref{prop:variance3orderpurity} we know that, as long as we can direct solve (or upper bound) the value of $\Gamma_3,\Gamma_4,\Gamma_5,\Gamma_6$, we can then calculate the variance $\mb{Var}[\hat{M}_+]$ directly. In Section~\ref{Sec:suppProof}, we will calculate $\Gamma_3,\Gamma_4,\Gamma_5,\Gamma_6$ in different cases of system dimension $d$, trial number $N_M$ and the rank of state $\rho$. We summarize the results as follows.

\textbf{Values of variance terms} $\{\Gamma_t\}$: (assuming $\mc{E}$ is a unitary $6$-design)
\begin{enumerate}
\item (Proposition~\ref{Prop:variancebound}) When the state $\rho\in\mc{D}(\mc{H}^{A})$ is pure, we have
\begin{equation}
\begin{aligned}
&\Gamma_3 = \frac{6d^2 - 2d + 8}{d+2} < 6d^2, \\
&\Gamma_4 = 4\,\frac{3d^3+ 5d^2 -d +5}{d^2+5d+6} < 14d,  \\
&\Gamma_5 = \frac{48d^3 + 68 d^2 + 60d + 64}{d^3 + 9d^2 + 26d +24} < 48, \\
&\Gamma_6 = 4\frac{d^4 + 59 d^3 + 107 d^2 + 109 d +84}{d^4 + 14 d^3 + 71 d^2 +154 d +120} < 10, \\
\end{aligned}
\end{equation}

\item (Proposition~\ref{prop:trQ3456}) When $d\gg 1$, 
the leading terms of $\Gamma_t$ with respect to $d$ are
\begin{equation}
\begin{aligned}
\Gamma_3 &\sim \frac{1}{d^3} \left\{ d^5 + 3d^5 \tr[\rho^2] + 2d^5 \tr[\rho^3] \right\}, \\
\Gamma_4 &\sim \frac{1}{d^4} \left\{ d^5 \tr[\rho^2] + 3d^5 \tr[\rho^2]^2 + 4d^5 \tr[\rho^3] + 6d^5\tr[\rho^4] \right\}, \\
\Gamma_5 &\sim \frac{1}{d^5} \left\{ 2d^5 \tr[\rho^2]^2 + 16d^5 \tr[\rho^2]\tr[\rho^3] + 6d^5 \tr[\rho^4] + 24d^5\tr[\rho^5] \right\}, \\
\Gamma_6 &\sim\frac{1}{d^6} \left\{ 4d^6\tr[\rho^3]^2 \right\}.
\end{aligned}
\end{equation}

\item (Proposition~\ref{prop:Q3456order} and \ref{Prop:3varianceExactBound}) For all $\rho\in\mc{D}(\mc{H}^{A})$, when $d\gg 1$, the $d$-orders of $\Gamma_t$ are
\begin{equation}
\begin{aligned}
\Gamma_3 &= \mc{O}(d^2), \Gamma_4 = \mc{O}(d), \\
\Gamma_5 &= \mc{O}(1), \Gamma_6 = \mc{O}(1),
\end{aligned}
\end{equation}
and the exact value of $\Gamma_3$ is
\begin{equation}
\Gamma_3 = (d+2)^{-1} \{ (d-1)(d^2+3d+4) + 3d(d-1)(d+1)\tr[\rho^2] + 2(d^3-d^2+6)\tr[\rho^3] \}.
\end{equation}
\end{enumerate}

Therefore, when the underlying $\rho$ is a pure state, we can directly calculate the variance $\mb{Var}[\hat{M}_+](t)$ from Proposition~\ref{prop:variance3orderpurity}; when $d\gg 1$, we can estimate the leading term of $\mb{Var}[\hat{M}_+](t)$, whose behavior is similar to the pure state case. In both cases, the variance $\mb{Var}[\hat{M}_+]$ can be upper bounded by
\begin{equation}
\nu(N_M,d) \leq \frac{5}{2} + \frac{108}{N_M} + \frac{54d}{N_M^2} + \frac{9d^2}{N_M(N_M-1)(N_M-2)} - \tr[\rho^3]^2.
\end{equation}
When $N_M\gg d$, the variance will approximately be $\nu = \frac{5}{2} - \tr[\rho^3]^2$.

On the other hand, in the regime of $d\gg N_M\gg 1$, the variance $\nu(N_M,d)$ is mainly determined by $\Gamma_3$, then
\begin{equation}
\nu(N_M,d) = \frac{3(d+2)^{-1}}{2(N_M-2)^3}\{(d+1)(d^2+3d+4) + 3d(d-1)(d+1)\tr[\rho^2] + 2(d^3-d^2+6)\tr[\rho^3]\} + \mc{O}(d).
\end{equation}
In this case, $\nu(N_M,d)\sim d^2$, which means that  the 3-order purity measurement is asymptotically more efficient than tomography, which requires $\Omega(d^3)$ times of experiments in general.

\subsection{Statistical analysis for Negativity detection} \label{SSec:statNega}

Recall that the negativity operator can be constructed by
\begin{equation} \label{eq:NegaObs}
\begin{aligned}
M_{neg} &= \frac{1}{2}\left( W^A_0\otimes W^B_1 + W^A_1\otimes W^B_0 \right) = \frac{1}{2} \left( M^A_+ \otimes M^B_+ - M^{AB}_{+} \right) \\
&= \frac{1}{2}\left[ (\Phi^3_A\otimes\Phi^3_B) (O^A_+ \otimes O^B_+) - \Phi^3_{AB}(O^{AB}_{+}) \right].
\end{aligned}
\end{equation}
We denote $O^{AB}_{++}:= O^A_+ \otimes O^B_+$. To finish the estimate of negativity-moment $\tr[(\rho_{AB}^{T_B})^3]$, we construct estimator $\hat{M}_{++}^{AB}$ and $\hat{M}_{+}^{AB}$ with local independent random unitaries on $A$ and $B$ and global random unitaries on $A,B$, respectively. The final estimator for negativity is then given by $\hat{M}_{neg}:= \hat{M}_{++}^{AB} - \hat{M}_{+}^{AB}$.

Suppose we perform $N_U' N_M'$ and $N_U N_M$ times of experiments for estimator $\hat{M}_{++}^{AB}$ and $\hat{M}_{+}^{AB}$, respectively. The negativity estimation procedure is as follows,

\underline{$\hat{M}_{neg} (\mc{E}^{AB},N_U', N_M', N_U, N_M):$ Estimator function for negativity-moment $\tr[(\rho_{AB}^{T_B})^3]$.}

\begin{enumerate}
\item Set $\hat{M}_{neg} :=0, \hat{M}_{++}^{AB}:=0, \hat{M}_{+}^{AB}:=0$.
\item Perform experiment with estimator $\hat{M}_{++}^{AB}:= \hat{M}_+(\mc{E}^A\otimes \mc{E}^B, N_U', N_M')$.
\item Perform experiment with estimator $\hat{M}_{+}^{AB}:= \hat{M}_+(\mc{E}^{AB}, N_U, N_M)$.
\item Set $\hat{M}_{neg} := \hat{M}_{++}^{AB} - \hat{M}_{+}^{AB}$.
\end{enumerate}

The estimator $\hat{M}_{+}^{AB}$ is simply a 3-moment estimator. From Proposition~\ref{prop:variance3orderpurity}, we have
\begin{equation}
\mb{Var}[\hat{M}_{+}^{AB}] \leq \frac{1}{N_U} \nu(N_M,d^2).
\end{equation}
Here we suppose $d_A=d_B =d$.

Similar to the discussion in Sec.~\ref{SSec:stat3mom}, we introduce $N_U'N_M'$ estimators $\{\hat{r}_{U,V}(i)\}$. The $i$-th estimator is
\begin{equation} \label{eq:rUVi}
\hat{r}_{U,V}(i) := \ketbra{a,b}{a,b}, \quad \text{with probability  } P(a,b|U,V):= \tr[(\ketbra{a}{a}\otimes\ketbra{b}{b}) (U\otimes V)\rho_{AB}(U^\dag\otimes V^\dag)],
\end{equation}
where $U\otimes V$ is randomly chosen from $\mc{E}^{A}\times \mc{E}^{B}$. For a diagonal operator $Q\in\mc{L}((\mc{H}^{AB})^{\otimes n})$ and $n$ independent estimators $\{\hat{r}_{U,V}(i)\}_{i=1}^n$, it is easy to prove that
\begin{equation}
\mbb{E}\;\tr\left[\left(\bigotimes_{i=1}^{n}\hat{r}_{U,V}(i)\right) Q \right] = \tr[\left(\Phi^n_A\otimes \Phi^n_B \right)(Q) \rho^{\otimes n}_{AB}].
\end{equation}

The estimators $\{\hat{M}_{++}^{AB}(t)\}_{t=1}^{N_U'}$ and $\hat{M}_{++}^{AB}$ can be expressed as
\begin{equation}
\begin{aligned}
\hat{M}_{++}^{AB}(t) &= \frac{1}{2}\binom{N_M'}{3}^{-1} \sum_{i<j<k} \tr\left[\left(\hat{r}_{U,V}(i)\otimes \hat{r}_{U,V}(j)\otimes \hat{r}_{U,V}(k)\right)O_{+}^{AB}\right], \\
\hat{M}_{++}^{AB} &= \frac{1}{N_U}\sum_{t=1}^{N_U'}\hat{M}_{++}^{AB}(t).
\end{aligned}
\end{equation}

We now calculate the variances of $\{\hat{M}_{++}^{AB}(t)\}_{t=1}^{N_U'}$ and $\hat{M}_{++}^{AB}$.

\begin{prop} \label{Prop:varianceOppAB}
For the estimator $\hat{M}_{++}(\mc{E}^A\otimes\mc{E}^B,N_U,N_M)$ with $\mc{E}^A$ and $\mc{E}^B$ the Haar measure on $\mc{L}(\mc{H}^A)$ and $\mc{L}(\mc{H}^B)$, respectively, the variances of $\hat{M}_{++}^{AB}(t)$ and $\hat{M}_{++}^{AB}$ are
\begin{equation}
\begin{aligned}
&\mb{Var}[\hat{M}_{++}(t)] = \mu(N_M,d), \\
&\mb{Var}[\hat{M}_{++}] = \frac{1}{N_U}\mu(N_M,d),
\end{aligned}
\end{equation}
where
\begin{equation} \label{eq:muNMd}
\mu(N_M,d):= \frac{1}{4}\Delta_6 + \frac{9}{4}\frac{1}{N_M}\Delta_5 + \frac{9}{2}\frac{1}{N_M^2}\Delta_4 + \frac{3}{2}\frac{1}{N_M(N_M-1)(N_M-2)}\Delta_3 - (\tr[\rho^3]+ \tr[(\rho^{T_B})^3])^2.
\end{equation}
here we assume $d_A = d_B = d$. The variance terms $\Delta_3, \Delta_4, \Delta_5, \Delta_6$ are
\begin{equation} \label{eq:Delta3456}
\begin{aligned}
\Delta_3(\rho,O_{++}^2,\mc{E}^A\otimes\mc{E}^B) &:= \tr[(\Phi^3_A\otimes \Phi^3_B)(O_{++}^2)\rho^{\otimes 3}], \\
\Delta_4(\rho,O_{123,124},\mc{E}^A\otimes\mc{E}^B) &:= \tr[(\Phi^4_A\otimes \Phi^4_B)(O_{123,124})^{AB}\rho^{\otimes 4}], \\
\Delta_5(\rho,O_{123,145},\mc{E}^A\otimes\mc{E}^B) &:= \tr[(\Phi^5_A\otimes \Phi^5_B)(O_{123,145})^{AB}\rho^{\otimes 5}], \\
\Delta_6(\rho,O_{++}^{\otimes 2},\mc{E}^A\otimes\mc{E}^B) &:= \tr[(\Phi^6_A\otimes \Phi^6_B)(O_{++}^{\otimes 2})\rho^{\otimes 6}],
\end{aligned}
\end{equation}
and
\begin{equation} \label{eq:O123145AB}
\begin{aligned}
O_{123,145}^{AB} &:= O_{123,145}^{A}\otimes O_{123,145}^{B}, \\
O_{123,124}^{AB} &:= O_{123,124}^{A}\otimes O_{123,124}^{B}.
\end{aligned}
\end{equation}
\end{prop}

\begin{proof}
The calculation is similar to the one in Proposition~\ref{prop:variance3orderpurity}. The total variance can be decomposed into two terms
\begin{equation}
\mb{Var}[\hat{M}_{++}^{AB}(t)] = \sub{\mbb{E}}{U}\left[ \sub{\mbb{E}}{a}(\hat{M}_{++}^2(t)|U) \right] - \left[\sub{\mbb{E}}{U}\sub{\mbb{E}}{a}(\hat{M}_{++}(t)|U) \right]^2,
\end{equation}
the second term is just $(\tr[\rho^3]+ \tr[(\rho^{T_B})^3])^2$. Now we focus on the first term
\begin{equation} \label{eq:NegVarOppbound}
\begin{aligned}
& \sub{\mbb{E}}{U}\left[ \sub{\mbb{E}}{a}(\hat{M}_{++}^2(t)|U) \right] \\
=& \frac{1}{4}\binom{N_M}{3}^{-2} \sum_{\substack{i<j<k\\ l<m<n}} \mbb{E}\left\{\tr[\left(\hat{r}_{U,V}(i)\otimes \hat{r}_{U,V}(j)\otimes \hat{r}_{U,V}(k)\right) O_{++}] \tr[\left(\hat{r}_{U,V}(l)\otimes \hat{r}_{U,V}(m)\otimes \hat{r}_{U,V}(n)\right) O_{++}] \right\} \\
= &\frac{1}{4}\binom{N_M}{3}^{-2} \bigg\{ \binom{N_M}{6}\binom{6}{0}\binom{6}{3}\sub{\mbb{E}}{U} \{ \tr[\rho_{U,V}^{\otimes 3} O_{++} ]^2 \} + \binom{N_M}{5}\binom{5}{1}\binom{4}{2} \mbb{E}\left\{\tr[\left(\hat{r}_{U,V}(1)\otimes \rho_{U,V} \otimes \rho_{U,V} \right) O_{++}]^2 \right\} \\
& \quad+ \binom{N_M}{4}\binom{4}{2}\binom{2}{1} \mbb{E}\left\{\tr[\left(\hat{r}_{U,V}(1)\otimes \hat{r}_{U,V}(2) \otimes \rho_{U,V} \right) O_{++}]^2 \right\}  \\
& \quad+ \binom{N_M}{3}\binom{3}{3}\binom{0}{0} \mbb{E}\left\{\tr[\left(\hat{r}_{U,V}(1)\otimes \hat{r}_{U,V}(2) \otimes \hat{r}_{U,V}(3) \right) O_{++}]^2 \right\} \bigg\} \\
= &\frac{1}{4}\binom{N_M}{3}^{-2} \bigg\{ \binom{N_M}{6}\binom{6}{0}\binom{6}{3} \tr\left[\rho^{\otimes 6} (\Phi^6_{A}\otimes\Phi^6_{B}) (O_{++}^{\otimes 2}) \right]  + \binom{N_M}{5}\binom{5}{1}\binom{4}{2} \tr\left[\rho^{\otimes 5}(\Phi^5_{A}\otimes\Phi^5_{B}) (O_{123,145})\right]  \\
& \quad+ \binom{N_M}{4}\binom{4}{2}\binom{2}{1} \tr\left[\rho^{\otimes 4} (\Phi^4_{A}\otimes\Phi^4_{B}) (O_{123,124})\right] + \binom{N_M}{3}\binom{3}{3}\binom{0}{0} \tr\left[\rho^{\otimes 3} (\Phi^3_{A}\otimes\Phi^3_{B}) (O_{++}^2)\right] \bigg\} \\
= &\frac{1}{4}\binom{N_M}{3}^{-1} \bigg\{ \frac{1}{6}(N_M-3)(N_M-4)(N_M-5)\cdot \Delta_6 + \frac{3}{2}(N_M -3)(N_M -4)\cdot \Delta_5 + 3(N_M-3)\cdot \Delta_4 + 1\cdot \Delta_3 \bigg\} \\
= & \frac{1}{4}\Delta_6  + \frac{9}{4}\frac{1}{N_M}\Delta_5 + \frac{9}{2}\frac{1}{N_M}\Delta_4 +\frac{3}{2}\frac{1}{N_M(N_M-1)(N_M-2)}\Delta_3.
\end{aligned}
\end{equation}

The variance of the overall estimator $\hat{M}_{++}$ is then upper bounded by
\begin{equation}
\begin{aligned}
\mb{Var}[\hat{M}_{++}] &= \mbb{E}[\hat{M}_{++}^2] - \mbb{E}[\hat{M}_{++}]^2 \\
&= \frac{1}{N_U^2} \sum_{t,q=1}^{N_U} \left\{ \mbb{E}[\hat{M}_{++}(t)\hat{M}_{++}(q)] - \mbb{E}[\hat{M}_{++}(1)]^2 \right\} \\
&= \frac{1}{N_U^2} \sum_{t=1}^{N_U} \left\{ \mbb{E}[\hat{M}_{++}(t)\hat{M}_{++}(t)] - \mbb{E}[\hat{M}_{++}(1)]^2 \right\} \\
&= \frac{1}{N_U} \mb{Var}[\hat{M}_{++}(1)] \leq \frac{1}{N_U} \mu(N_M,d).
\end{aligned}
\end{equation}
\end{proof}

The variance of $\hat{M}_{neg}$ is then
\begin{equation}
\mb{Var}[\hat{M}_{neg}] = \mb{Var}[\hat{M}_{++}^{AB}] + \mb{Var}[\hat{M}_{+}^{AB}] = \frac{1}{N_U'}\mu(N_M',d) + \frac{1}{N_U}\nu(N_M,d^2).
\end{equation}


The values of $\{\Delta_t\}$, however, is much hard to be reduced to a simple form. Here we list their value in some simple cases.

\textbf{Values of variance terms} $\{\Delta_t\}$: (assuming $\mc{E}_A, \mc{E}_B$ are unitary $6$-designs)
\begin{enumerate}
\item (Proposition~\ref{Prop:varianceboundAB}) When the state $\rho\in\mc{D}(\mc{H}^{A})$ is a pure tensor state, we have
\begin{equation}
\begin{aligned}
\Delta_6 = &\tr[(\Phi^6_A\otimes \Phi^6_B)(O_{++}^{\otimes 2}) \rho^{\otimes 6}] = \Gamma_6^2(\psi,O_{+}^{\otimes 2},\mc{E}) < 10^2, \\
\Delta_5 = &\tr[(\Phi^5_A\otimes \Phi^5_B)(O_{123,145}^{AB}) \rho^{\otimes 5}] = \Gamma_5^2(\psi,O_{123,145},\mc{E}) < 48^2, \\
\Delta_4 = &\tr[(\Phi^4_A\otimes \Phi^4_B)(O_{123,124}^{AB}) \rho^{\otimes 4}] = \Gamma_4^2(\psi,O_{123,124},\mc{E}) < (14d)^2,  \\
\Delta_3 = &\tr[(\Phi^3_A\otimes \Phi^3_B)(O_{++}^2) \rho^{\otimes 3}] = \Gamma_3^2(\psi,O_{+}^2,\mc{E}) < (6d^2)^2.
\end{aligned}
\end{equation}

\item (Proposition~\ref{prop:Q3456orderAB} and \ref{prop:Q3456orderABBell}) When $d\gg 1$, the asymptotic relation of $\{\Delta_t\}$ with respect to $d$ is
\begin{equation}
\begin{aligned}
\Delta_3 &= \mc{O}(d^4), \Delta_4 = \mc{O}(d^2), \\
\Delta_5 &= \mc{O}(1), \Delta_6 = \mc{O}(1).
\end{aligned}
\end{equation}
Moreover, when the state $\rho$ is maximally entangled state, the asymptotic relation of $\{\Delta_t\}$ with respect to $d$ is
\begin{equation}
\begin{aligned}
\Delta_3 &= \Theta(d^4), \Delta_4 = \Theta(d^2), \\
\Delta_5 &= \Theta(1), \Delta_6 = \Theta(1).
\end{aligned}
\end{equation}

\item (Proposition~\ref{Prop:3varianceExactBoundAB}) For all $\rho\in\mc{D}(\mc{H}^{A})$, when $d\gg 1$, the exact value of $\Delta_3$ is
\begin{equation}
\begin{aligned}
\Delta_3 =\frac{1}{(d+2)^2} &\Big\{ \left[\tr(\rho^2_A) + \tr(\rho^2_B)\right] [3d(d-1)^2(d+1)(d^2+3d+4)] \\
&\quad + \left[\tr(\rho_{A}^3)+\tr(\rho_{B}^3)\right][2(d-1) (6 + (d-1) d^2) (d^2+3d+4)] \\
&\quad + \tr(\rho_{AB}^2)[3 d^2 (d^2-1)^2] + \tr(\rho_{AB}^3)[2 (6 + (d-1) d^2)^2] \\
&\quad + \tr(\rho_{AB}\rho_{A}\otimes\rho_{B})[6 d^2 (d^2-1)^2] \\
&\quad + [\tr(\rho_{AB}^2\rho_{A}) +\tr(\rho_{AB}^2\rho_{B})] [6d (d-1)(d+1)((d-1)d^2 + 6)] \\
&\quad + 2\tr[(\rho_{AB}^{T_A})^3][(d-1)d^2 + 6)^2] + (d(d+1)^2 - 4)^2 \Big\}
\end{aligned}
\end{equation}
for all states $\rho\in\mc{D}(\mc{H}^{AB})$.
\end{enumerate}

Therefore, when the underlying $\rho$ is a pure tensor state, we can directly calculate the variance $\mb{Var}[\hat{M}_+](t)$ from Proposition~\ref{Prop:varianceboundAB},
\begin{equation} \label{eq:VarPureTensor}
\mu(N_M,d)\leq 25 + 5184\frac{1}{N_M} + 882\frac{d^2}{N_M^2} + 54\frac{d^4}{N_M(N_M-1)(N_M-2)} - (\tr[\rho^3]+\tr[(\rho^{T_B})^3])^2.
\end{equation}
Similarly, when $\rho$ is a mixed product state, we can apply the method adopted in the $3$-order purity analysis in Proposition~\ref{prop:trQ3456}. We have $\Delta_t = \Gamma_t^2$. The general separable state is just a convex mixture. From the 3-order purity analysis we know that, the more mixed the state is, the smaller variance is.

In general, the state $\rho_{AB}$ is entangled. When $d\gg 1$, we can estimate the leading term of $\mb{Var}[\hat{M}_{++}](t)$, whose asymptotic behavior is similar to the pure tensor state case,
\begin{equation}
\nu(N_M,d) = c_1 \mc{O}(1)+ c_2\frac{\mc{O}(1)}{N_M} + c_3\frac{\mc{O}(d^2)}{N_M^2} + c_4\frac{\mc{O}(d^4)}{N_M(N_M-1)(N_M-2)} - (\tr[\rho^3]+\tr[(\rho^{T_B})^3])^2.
\end{equation}
When $N_M\gg d$, the variance will be a constant. Note that, the variance of maximally entangled state has exactly the same asymptotic property as the one of pure tensor state in Eq.~\eqref{eq:VarPureTensor},
\begin{equation}
\nu(N_M,d) = c_1 + c_2\frac{1}{N_M} + c_3\frac{d^2}{N_M^2} + c_4\frac{d^4}{N_M(N_M-1)(N_M-2)} - (\tr[\rho^3]+\tr[(\rho^{T_B})^3])^2.
\end{equation}

On the other hand, in the regime of $d\gg N_M\gg 1$, the variance $\nu(N_M,d)$ is mainly determined by $\Gamma_3$, then
\begin{equation}
\begin{aligned}
\mu(N_M,d) =\frac{(d+2)^{-2}}{N_M(N_M-1)(N_M-2)} &\Big\{ \left[\tr(\rho^2_A) + \tr(\rho^2_B)\right] [3d(d-1)^2(d+1)(d^2+3d+4)] \\
&\quad + \left[\tr(\rho_{A}^3)+\tr(\rho_{B}^3)\right][2(d-1) (6 + (d-1) d^2) (d^2+3d+4)] \\
&\quad + \tr(\rho_{AB}^2)[3 d^2 (d^2-1)^2] + \tr(\rho_{AB}^3)[2 (6 + (d-1) d^2)^2] \\
&\quad + \tr(\rho_{AB}\rho_{A}\otimes\rho_{B})[6 d^2 (d^2-1)^2] \\
&\quad + [\tr(\rho_{AB}^2\rho_{A}) +\tr(\rho_{AB}^2\rho_{B})] [6d (d-1)(d+1)((d-1)d^2 + 6)] \\
&\quad + 2\tr[(\rho_{AB}^{T_A})^3][(d-1)d^2 + 6)^2] + (d(d+1)^2 - 4)^2 \Big\} + \mc{O}(d^2)
\end{aligned}
\end{equation}
In this case, $\nu(N_M,d)\sim d^4$.

\section{Detailed proofs of statistical bounds} \label{Sec:suppProof}

In this section, we provide detailed proofs of statistical bounds presented in Sec.~\ref{Sec:suppStat}.

We first introduce some notations. Note that, all the permutation elements $\pi\in S_t$ can be classified by their cycle structures. We will use partition number to denote the cycle structures of elements in $S_t$. A partition of $t$ is a weakly decreasing sequence of positive integers $[\xi_1, \xi_2, ..., \xi_k]$ where $\xi_1 \geq \xi_2 \geq ... \geq \xi_k =0$. We denote $|\xi|=t$ where $t:=\sum_i \xi_i$. The values $\{\xi_i\}$ are called the parts of $\xi$, which denotes the cycle length occurred in $\pi$. Two elements $\pi,\sigma \in S_t$ belong to the same conjugate class if and only if $\xi(\pi)=\xi(\sigma)$. The partition number can also be expressed as Young diagram. For example, the conjugate class of permutation element $\pi=(1,4,6)(3,5)\in S_6$ is
\begin{equation}
\xi((1,4,6)(3,5)) = [3,2,1] = \yng(3,2,1).
\end{equation}

In the later discussion, we will classify the $t$-dit strings $\vec{a}\in\mbb{Z}_d^t$ by the numbers of same values occurred in $\vec{a}=(a_1, a_2, ..., a_t)$. As a simple notation, we correlate a specific $t$-dit string $\vec{a}$ to a given permutation element $\omega(\vec{a})\in S_t$, whose the cyclic notation reflects the same values occurred in $\vec{a}$.
\begin{definition} \label{def:omega}
For $t$-dit string $\vec{a}$, we group all its entries with the same values, and record the indices as
\begin{equation}
(i(1)_1 i(1)_2 ... i(1)_{\lambda_1}) (i(2)_1 i(2)_2 ... i(2)_{\lambda_2}) ... (i(k)_1 i(k)_2 ... i(k)_{\lambda_k}),
\end{equation}
where in each group of indices $(i(j)_1 i(j)_2 ... i(j)_{\lambda_j})$, the value of corresponding entries are the same; and for each two indices in different groups, the values are different. Then we say the permutation element
\begin{equation}
\omega(\vec{a}) = (i(1)_1 i(1)_2 ... i(1)_{\lambda_1}) (i(2)_1 i(2)_2 ... i(2)_{\lambda_2}) ... (i(k)_1 i(k)_2 ... i(k)_{\lambda_k})
\end{equation}
is the corresponding permutation element of $\vec{a}$.
\end{definition}
For example, when $t=6$, the string $\vec{a}_0=(5,3,5,5,3,0)$ will be related to the permutation element
\begin{equation}
\omega(\vec{a}_0)= \omega((5,3,5,5,3,0)) =(1,3,4)(2,5)\in S_6.
\end{equation}
Note that, there is a freedom of determining the sequence of numbers in the same cycle. In the example above, we may equivalently set $\omega(\vec{a}_0) = (1,4,3)(2,5)$. For the convenience of following discussion, we choose the representation of $\omega(\vec{a})$ with the element where numbers are in the increasing order in the same cycle.

After introducing $\omega(\vec{a})$, we can also classify $t$-dit strings by the partition numbers. The class of $\vec{a}$ is defined to be the class of $\omega(\vec{a})$, denoted by the partition number $\lambda(\vec{a})$,
\begin{equation}
\lambda(\vec{a}):= \xi(\omega(\vec{a})).
\end{equation}
In the example above, the class of string $\vec{a}_0=(5,3,5,5,3,0)$ is
\begin{equation}
\lambda(\vec{a}_0) = \xi((1,3,4)(2,5)) = [3,2,1] = \yng(3,2,1).
\end{equation}

In the following discussion, we will frequently use the term $\tr[W_\pi \ketbra{\vec{a}}{\vec{a}}]$. To calculate the value of it, we introduce the following definition.
\begin{definition}
For two permutation operators $\pi, \omega\in S_t$, we say that $\pi$ can be \textit{embedded} into $\omega$, $\pi\subseteq \omega$, if one can create $\omega$ by merging the cycles in $\pi$.
\end{definition}

For example, for three permutation elements $\pi=(12)(46), \sigma=(13)$ and $\omega=(123)(46)$ in $S_6$, we have $\pi\subseteq \omega$, $\sigma\subseteq \omega$, and $\pi\subsetneq \sigma$.

With the definition of embedding, we have
\begin{equation}
\tr[W_\pi \ketbra{\vec{a}}{\vec{a}}] = \mathbbm{1}[\pi\subseteq \omega(\vec{a})].
\end{equation}
Where $\mathbbm{1}[s]$ is the indicating function which takes the value of $1$ when $s$ is true and $0$ otherwise.

We further define a embedding constant $\gamma_{\xi,\lambda}$ for two $t$-partitions $\xi$ and $\lambda$.
\begin{definition}
For two $t$-partitions $\xi$ and $\lambda$, the number of permutation elements $\pi\in \xi$ which can be embedded into a given element $\sigma\in\lambda$, denoted as $\gamma_{\xi,\lambda}(\sigma)$, is irrelevant of the choice of $\sigma$. We call $\gamma_{\xi,\lambda}$ the embedding constant from $\xi$ to $\lambda$.
\end{definition}

\begin{prop} \label{prop:gamma}
We have,
\begin{equation}
\sum_{\pi\in\xi} \tr[W_\pi \ketbra{\vec{a}}{\vec{a}}] = \gamma_{\xi,\lambda(\vec{a})}.
\end{equation}
\end{prop}
\begin{proof}
\begin{equation}
\begin{aligned}
\sum_{\pi\in\xi} \tr[W_\pi \ketbra{\vec{a}}{\vec{a}}] &= \sum_{\pi\in\xi} \mathbbm{1}[\pi\subseteq \omega(\vec{a})],
\end{aligned}
\end{equation}
by the definition of $\gamma_{\xi,\lambda}$ we know the proposition holds.
\end{proof}

Define
\begin{equation} \label{eq:Ta}
T(\vec{a}) = T_{\lambda(\vec{a})} := \sum_{\xi} \gamma_{\xi,\lambda(\vec{a})},
\end{equation}
then
\begin{equation} \label{eq:Talambda}
T(\vec{a}) = \sum_{\xi} \gamma_{\xi,\lambda(\vec{a})} = \sum_{\pi\in S_t} \tr[W_\pi\ketbra{\vec{a}}{\vec{a}}] = \prod_{i=1}^k (\lambda_i)!.
\end{equation}

\begin{prop} \label{Prop:variancebound}
For observable $O_+ \in \mc{L}((\mc{H}^A)^{\otimes 3})$ with the form $O_+ = \sum_{\vec{a}\in \mbb{Z}_d^3} [1 + (-d)^{wt(\vec{a})-1}] \ketbra{\vec{a}}{\vec{a}}$, when the random unitaries are chosen in unitary $6$-design, when the underlying state $\rho\in\mc{D}(\mc{H}^{A})$ is a pure state, we have
\begin{equation}
\begin{aligned}
\Gamma_6 = &\tr[\Phi^6(O_+^{\otimes 2}) \rho^{\otimes 6}] = 4\,\frac{d^4 + 59 d^3 + 107 d^2 + 109 d +84}{d^4 + 14 d^3 + 71 d^2 +154 d +120} < 10, \\
\Gamma_5 = &\tr[\Phi^5(O_{123,145}) \rho^{\otimes 5}] = \frac{48d^3 + 68 d^2 + 60d + 64}{d^3 + 9d^2 + 26d +24} < 48, \\
\Gamma_4 = &\tr[\Phi^4(O_{123,124}) \rho^{\otimes 4}] = 4\,\frac{3d^3+ 5d^2 -d +5}{d^2+5d+6} < 14d,  \\
\Gamma_3 = &\tr[\Phi^3(O_+^2) \rho^{\otimes 3}] = \frac{6d^2 - 2d + 8}{d+2} < 6d^2,
\end{aligned}
\end{equation}
where $O_{123,145}$ and $O_{123,124}$ is defined in Eq.~\eqref{eq:O123145}.
\end{prop}

\begin{proof}

First we note that, for any operators $Q\in\mc{L}((\mc{H}^A)^{\otimes t})$ with $t\leq 6$, we have
\begin{equation} \label{eq:PhitQrhot}
\begin{aligned}
    \tr[\Phi^t(Q)\rho^{\otimes t}] &= \sum_{\pi,\sigma\in S_t} C_{\pi,\sigma} \tr[W_\pi Q]\tr[W_\sigma \rho^{\otimes t}] \\
    &= \sum_{\pi \in S_t} \tr[W_\pi Q] \sum_{\sigma \in S_t} C_{\pi,\sigma} \\
    &= \frac{(d-1)!}{(d+t-1)!} \sum_{\pi \in S_t} \tr[W_\pi Q].
\end{aligned}
\end{equation}
Here, the first equality is simply the application of Weingarten integral. The second equality is because $\tr[\rho^m]= 1, \forall m=0,1,2,...$ for pure $\rho$. The third equality is due to the property of Weingarten matrix $\sum_{\pi\in S_t} C_{\pi,\sigma} = \frac{(d-1)!}{(d+t-1)!}$. Eq.~\eqref{eq:PhitQrhot} is in fact proportional to the projection onto the symmetric subspace for this pure state case.

In our problem, $Q$ is diagonal in the $\ket{\vec{a}}$ basis, $Q= \sum_{\vec{a}\in \mbb{Z}_d^{t}} Q(\vec{a}) \ketbra{\vec{a}}{\vec{a}}$. Then
\begin{equation} \label{eq:PhitQdiag}
    \tr[\Phi^t(Q)\rho^{\otimes t}] \leq \frac{(d-1)!}{(d+t-1)!} \sum_{\pi \in S_t} \tr[W_\pi Q] = \frac{(d-1)!}{(d+t-1)!} \sum_{\vec{a}\in \mbb{Z}_d^t} Q(\vec{a}) \sum_{\pi\in S_t} \bra{\vec{a}} W_\pi \ket{\vec{a}} = \sum_{\vec{a}\in \mbb{Z}_d^t} Q(\vec{a})T(\vec{a}).
\end{equation}
Here, $T(\vec{a})$ is defined in Eq.~\eqref{eq:Ta}. Note that, the value of $T(\vec{a})$ only depends on how many values in $\vec{a}$ are the same, or, the partition of $\vec{a}$. This can be easily characterized by the corresponding permutation element $\omega(\vec{a})$ defined in Definition \ref{def:omega}. The partition of $\vec{a}$ is then the partition of corresponding permutation element $\omega(\vec{a})$.

For a $t$-dit string $\vec{a}$ with partition $[\lambda_1, \lambda_2, ..., \lambda_k]$, from Eq.~\eqref{eq:Talambda} we know that
\begin{equation}
T(\vec{a}) = T_{\lambda(\vec{a})} = \prod_{i=1}^{k} (\lambda_i)!.
\end{equation}

In our case, the $Q$ function only depends on the weight of $\vec{a}$, i.e., the number of same values in $\vec{a}$. More specifically,
\begin{equation} \label{eq:Q3Q4Q5Q6}
\begin{aligned}
    Q_3 &:= O_+^2 = \sum_{\vec{a}\in \mbb{Z}_d^3} O_+^2(wt(a_1 a_2 a_3)) \ket{\vec{a}}\bra{\vec{a}}, \\
    Q_4 &:= O_{123,124} = \sum_{\vec{a}\in \mbb{Z}_d^4} O_+(wt(a_1 a_2 a_3))O_+(wt(a_1 a_2 a_4)) \ket{\vec{a}}\bra{\vec{a}}, \\
    Q_5 &:= O_{123,145} = \sum_{\vec{a}\in \mbb{Z}_d^5} O_+(wt(a_1 a_2 a_3))O_+(wt(a_1 a_4 a_5)) \ket{\vec{a}}\bra{\vec{a}}, \\
    Q_6 &:= O_+^{\otimes 2} = \sum_{\vec{a}\in \mbb{Z}_d^6} O_+(wt(a_1 a_2 a_3))O_+(wt(a_4 a_5 a_6)) \ket{\vec{a}}\bra{\vec{a}},
\end{aligned}
\end{equation}
with $O_+(wt=1) = 2$, $O_+(wt=2) = 1-d$, and $O_+(wt=2) = 1+d^2$.

Therefore, to calculate $\sum_{\vec{a}\in \mbb{Z}_d^t} Q(\vec{a}) \sum_{\pi\in S_t} \bra{\vec{a}} W_\pi \ket{\vec{a}}$, we first classify all the $t$-dit strings $\vec{a}\in \mbb{Z}_d^t$ by their partitions $\lambda{\vec{a}}$, and then futher divide them by the weight of the subsystems. By counting the weight of the subsystems, we define the ``sub-types'' $\{j_\lambda\}$ of a given partition class $\lambda$ of $\vec{a}$.
The partition $\lambda$ and sub-type $j_\lambda$ determine the value of $T(\vec{a})=T_{\lambda(\vec{a})}$ and $Q(\vec{a}) = Q(j_\lambda)$, respectively. We then count the number of elements $\vec{a}$ in all partition classes and subtypes, and finally figure out the results.

To be more specifically,
\begin{equation} \label{eq:sumQaTa}
\begin{aligned}
\sum_{\vec{a}\in \mbb{Z}_d^{t}} Q_t(\vec{a}) T(\vec{a})&= \sum_{\lambda} T_\lambda \sum_{\vec{a}\in \lambda} Q_t(\vec{a})\\
&= \sum_{\lambda} T_\lambda \sum_{(j_\lambda)\in\lambda}\#\{j_\lambda\} Q_t(j_\lambda).
\end{aligned}
\end{equation}

We start from the simplest $Q_3$ case, i.e., to estimate $\sum_{\vec{a}\in \mbb{Z}_d^3} Q_3(\vec{a}) T(\vec{a}) $. When $t=3$, the partition class of $\mbb{Z}_d^3$ determines the subsystem weight in $\mbb{Z}_d^3$. We classify the elements by $\lambda$ and list the values of $T_\lambda$ and $Q_3(j_\lambda)$ in Table~\ref{tab:classQ3}.
\begin{table}[htbp]
\begin{tabular}{c|cc||c|ccc}
\hline\hline
Partition classes $\lambda$ & $\#\{\lambda\}$ & $T_\lambda$ & Sub-type $j_\lambda: a_1 a_2 a_3$ & $\#\{j_\lambda\}$ & $wt(a_1 a_2 a_3)$ & $Q_3(j_\lambda)$ \\ \hline
$[111]$  & $1\cdot A_d^3$ & $1$ & $abc$ & $1\cdot A_d^3$ & $1$ & $4$ \\ \hline
$[21]$  & $3\cdot A_d^2$ & $2$ & $abb$ & $3\cdot A_d^2$ & $2$ & $(1-d)^2$ \\ \hline
$[3]$  & $1\cdot A_d^1$ & $2$ & $aaa$ & $1\cdot A_d^1$ & $3$ & $(1+d^2)^2$ \\ \hline
\hline
\end{tabular}
\caption{The classes and elements number of $\vec{a}$ for $T_\lambda$ and $Q_3(j_\lambda)$.} \label{tab:classQ3}
\end{table}

Therefore, from Eq.~\eqref{eq:sumQaTa} we have
\begin{equation}
\begin{aligned}
&\tr[\Phi^3(Q_3)\rho^{\otimes 3}]\leq \frac{(d-1)!}{(d+2)!}\sum_{\vec{a}\in \mbb{Z}_d^3} Q_3(\vec{a}) T(\vec{a}) \\
&=\frac{(d-1)!}{(d+2)!} \left\{ 4\cdot 1\cdot A_d^3 + (1-d)^2\cdot 2\cdot 3A_d^2 + (1+d^2)^2 \cdot 6 \cdot A_d^1 \right\} \\
&=\frac{6d^3 - 2d +8}{d+2} < 6d^2.
\end{aligned}
\end{equation}

For the $Q_4$ case, the sub-type of $\vec{a}$ depends on the weight of subsystem $wt(a_1 a_2 a_3)$ and $wt(a_1 a_2 a_4)$. We classify the elements by $\lambda$ and $j_\lambda$ in Table~\ref{tab:classQ4}.
\begin{table}[htbp]
\begin{tabular}{c|cc||c|ccc}
\hline\hline
Partition classes $\lambda$ & $\#\{\lambda\}$ & $T_\lambda$ & Sub-type $j_\lambda:a_1 a_2 |a_3|a_4$ & $\#\{j_\lambda\}$ & $(wt(a_1 a_2 a_3),wt(a_1 a_2 a_4)) $ & $Q_4(j_\lambda)$ \\ \hline
$[1111]$  & $1\cdot A_d^4$ & $1$ & $ab|c|d$ & $1\cdot A_d^4$ & $(1,1)$ & $4$ \\ \hline
$[211]$  & $6\cdot A_d^3$ & $2$ & $aa|b|c$ & $1 \cdot A_d^3$ & $(2,2)$ & $(1-d)^2$ \\
& & & $ab|a|c$ & $4 \cdot A_d^3$ & $(2,1)$ & $2(1-d)$ \\
& & & $bc|a|a$ & $1 \cdot A_d^3$ & $(1,1)$ & $4$ \\\hline
$[22]$  & $3\cdot A_d^2$ & $4$ & $aa|b|b$ & $1 \cdot A_d^2$ & $(2,2)$ & $(1-d)^2$ \\
& & & $ab|a|b$ & $2 \cdot A_d^2$ & $(2,2)$ & $(1-d)^2$ \\  \hline
$[31]$  & $4\cdot A_d^2$ & $6$ & $aa|a|b$ & $2\cdot A_d^2$ & $(3,2)$ & $(1+d^2)(1-d)$ \\
& & & $ab|a|a$ & $1 \cdot A_d^2$ & $(2,2)$ & $(1-d)^2$ \\\hline
$[4]$  & $1\cdot A_d^1$ & $24$ & $aa|a|a$ & $1\cdot A_d^1$ & $(3,3)$ & $(1+d^2)^2$ \\ \hline
\hline
\end{tabular}
\caption{The classes and elements number of $\vec{a}$ for $T_\lambda$ and $Q_4(j_\lambda)$. The sub-type $j_\lambda$ is determined by the weight of subsystem $a_1a_2a_3$ and $a_1a_2a_4$. $\#\{j_\lambda\}$ denote the number of elements contained in the sub-type $j_\lambda$. } \label{tab:classQ4}
\end{table}

Therefore,
\begin{equation}
\begin{aligned}
&\tr[\Phi^4(Q_4)\rho^{\otimes 4}]\leq \frac{(d-1)!}{(d+3)!}\sum_{\vec{a}\in \mbb{Z}_d^4} Q_4(\vec{a}) T(\vec{a}) \\
&=\frac{(d-1)!}{(d+3)!} \{ 4\cdot( 1\cdot A_d^4 + 2\cdot A_d^3  ) + 2(1-d)\cdot 2\cdot 4A_d^3 + (1-d)^2\cdot( 2\cdot A_d^3 + 4\cdot 3A_d^2 +  6\cdot 2A_d^2 ) \\
&\quad+ (1+d^2)(1-d)\cdot 6\cdot 2A_d^2 + (1+d^2)^2 \cdot 24 \cdot A_d^1 \} \\
&=2\,\frac{7d^3+6d^2+3d+8}{d^2 + 5d + 6} < 14d.
\end{aligned}
\end{equation}

For the $Q_5$ case, the sub-type of $\vec{a}$ depends on the weight of subsystem $wt(a_1 a_2 a_3)$ and $wt(a_1 a_4 a_5)$. We classify the elements by $\lambda$ and $j_\lambda$ in Table~\ref{tab:classQ5}.
\begin{table}[htbp]
\begin{tabular}{c|cc||c|ccc}
\hline\hline
Partition classes $\lambda$ & $\#\{\lambda\}$ & $T_\lambda$ & Sub-type $j_\lambda:a_1 | a_2 a_3 |a_4 a_5$ & $\#\{j_\lambda\}$ & $(wt(a_1 a_2 a_3),wt(a_1 a_4 a_5)) $ & $Q_5(j_\lambda)$ \\ \hline
$[11111]$  & $1\cdot A_d^5$ & $1$ & $a|bc|de$ & $1\cdot A_d^5$ & $(1,1)$ & $4$ \\ \hline
$[2111]$  & $10\cdot A_d^4$ & $2$ & $a|ab|cd$ & $4\cdot A_d^4$ & $(2,1)$ & $2(1-d)$ \\
& & & $b|aa|cd$ & $2 \cdot A_d^4$ & $(2,1)$ & $2(1-d)$ \\
& & & $b|ac|ad$ & $4 \cdot A_d^4$ & $(1,1)$ & $4$ \\\hline
$[221]$  & $15\cdot A_d^3$ & $4$ & $a|ab|bc$ & $8 \cdot A_d^3$ & $(2,1)$ & $2(1-d)$ \\
& & & $a|ac|bb$ & $4 \cdot A_d^3$ & $(2,2)$ & $(1-d)^2$ \\
& & & $c|aa|bb$ & $1 \cdot A_d^3$ & $(2,2)$ & $(1-d)^2$ \\
& & & $c|ab|ab$ & $2 \cdot A_d^3$ & $(1,1)$ & $4$ \\
\hline
$[311]$  & $10\cdot A_d^3$ & $6$ & $a|aa|bc$ & $2\cdot A_d^3$ & $(3,1)$ & $2(1+d^2)$ \\
& & & $a|ab|ac$ & $4 \cdot A_d^3$ & $(2,2)$ & $(1-d)^2$ \\
& & & $b|aa|ac$ & $4 \cdot A_d^3$ & $(2,1)$ & $2(1-d)$ \\
\hline
$[32]$  & $10\cdot A_d^2$ & $12$ & $a|aa|bb$ & $2\cdot A_d^2$ & $(3,2)$ & $(1-d)(1+d^2)$ \\
& & & $a|ab|ab$ & $4 \cdot A_d^2$ & $(2,2)$ & $(1-d)^2$ \\
& & & $b|ba|aa$ & $4 \cdot A_d^2$ & $(2,2)$ & $(1-d)^2$ \\
\hline
$[41]$  & $5\cdot A_d^2$ & $24$ & $a|aa|ab$ & $4\cdot A_d^2$ & $(3,2)$ & $(1-d)(1+d^2)$ \\
& & & $b|aa|aa$ & $1 \cdot A_d^2$ & $(2,2)$ & $(1-d)^2$ \\
\hline
$[5]$ & $1\cdot A_d^1$ & $120$ & $a|aa|aa$ & $1\cdot A_d^1$ & $(3,3)$ & $(1+d^2)^2$ \\
\hline
\hline
\end{tabular}
\caption{The classes and elements number of $\vec{a}$ for $T_\lambda$ and $Q_5(j_\lambda)$. The sub-type is determined by the weight of subsystem $a_1a_2a_3$ and $a_1a_4a_5$. $\#\{j_\lambda\}$ denote the number of elements contained in the sub-type $j_\lambda$.} \label{tab:classQ5}
\end{table}

Therefore,
\begin{equation}
\begin{aligned}
&\tr[\Phi^5(Q_5)\rho^{\otimes 5}]\leq \frac{(d-1)!}{(d+4)!}\sum_{\vec{a}\in \mbb{Z}_d^5} Q_5(\vec{a}) T(\vec{a}) \\
&=\frac{(d-1)!}{(d+4)!} \{ 4\cdot( A_d^5 + 2\cdot 4A_d^4 + 4\cdot 2A_d^3 ) + 2(1-d)\cdot( 2\cdot 6A_d^4 + 4\cdot 8A_d^3 + 6\cdot 4A_d^3) \\
&\quad+ (1-d)^2\cdot( 4\cdot 5 A_d^3 + 6\cdot 4A_d^3 + 12\cdot 8A_d^2 + 24 A_d^2 ) + 2(1+d^2)\cdot 6\cdot 2A_d^3\\
&\quad+ (1+d^2)(1-d)\cdot (12\cdot 2A_d^2 + 24\cdot 4A_d^2) + (1+d^2)^2 \cdot 120 \cdot A_d^1 \} \\
&=\frac{48d^3+68d^2+60d+64}{d^3 + 9d^2 + 26d + 24} < 48.
\end{aligned}
\end{equation}

For the $Q_6$ case, the sub-type $j_\lambda$ of $\vec{a}$ depends on the weight of subsystem $wt(a_1 a_2 a_3)$ and $wt(a_4 a_5 a_6)$. We classify the elements by $\lambda$ and $j_\lambda$ in Table~\ref{tab:classQ6}.
\begin{table}[htbp]
\begin{tabular}{c|cc||c|ccc}
\hline\hline
Partition classes $\lambda$ & $\#\{\lambda\}$ & $T_\lambda$ & Sub-type $j_\lambda:a_1 a_2 a_3 |a_4 a_5 a_6$ & $\#\{j_\lambda\}$ & $(wt(a_1 a_2 a_3),wt(a_4 a_5 a_6)) $ & $Q_6(j_\lambda)$ \\ \hline
$[111111]$  & $1\cdot A_d^6$ & $1$ & $abc|def$ & $1\cdot A_d^6$ & $(1,1)$ & $4$ \\
\hline
$[21111]$  & $15\cdot A_d^5$ & $2$ & $aab|cde$ & $6\cdot A_d^5$ & $(2,1)$ & $2(1-d)$ \\
& & & $abc|ade$ & $9 \cdot A_d^5$ & $(1,1)$ & $4$ \\
\hline
$[2211]$  & $45\cdot A_d^4$ & $4$ & $aac|bbd$ & $9 \cdot A_d^4$ & $(2,2)$ & $(1-d)^2$ \\
& & & $aab|bcd$ & $18 \cdot A_d^4$ & $(2,1)$ & $2(1-d)$ \\
& & & $abc|abd$ & $18 \cdot A_d^4$ & $(1,1)$ & $4$ \\
\hline
$[222]$  & $15\cdot A_d^3$ & $8$ & $aab|bcc$ & $9 \cdot A_d^3$ & $(2,2)$ & $(1-d)^2$ \\
& & & $abc|abc$ & $6 \cdot A_d^3$ & $(1,1)$ & $4$ \\
\hline
$[3111]$ & $20\cdot A_d^4$ & $6$ & $aaa|bcd$ & $2\cdot A_d^4$ & $(3,1)$ & $2(1+d^2)$ \\
& & & $aab|acd$ & $18 \cdot A_d^4$ & $(2,1)$ & $2(1-d)$ \\
\hline
$[321]$  & $60\cdot A_d^3$ & $12$ & $aaa|bbc$ & $6\cdot A_d^3$ & $(3,2)$ & $(1-d)(1+d^2)$ \\
& & & $aab|abc$ & $36 \cdot A_d^3$ & $(2,1)$ & $2(1-d)$ \\
& & & $aac|abb$ & $18 \cdot A_d^3$ & $(2,2)$ & $(1-d)^2$ \\
\hline
$[33]$  & $10\cdot A_d^2$ & $36$ & $aaa|bbb$ & $1\cdot A_d^2$ & $(3,3)$ & $(1+d^2)^2$ \\
& & & $aab|abb$ & $9 \cdot A_d^2$ & $(2,2)$ & $(1-d)^2$ \\
\hline
$[411]$  & $15\cdot A_d^3$ & $24$ & $aaa|abc$ & $6\cdot A_d^3$ & $(3,1)$ & $2(1+d^2)$ \\
& & & $aab|aac$ & $9 \cdot A_d^3$ & $(2,2)$ & $(1-d)^2$ \\
\hline
$[42]$  & $15\cdot A_d^2$ & $48$ & $aaa|abb$ & $6\cdot A_d^2$ & $(3,2)$ & $(1-d)(1+d^2)$ \\
& & & $aab|aab$ & $9 \cdot A_d^2$ & $(2,2)$ & $(1-d)^2$ \\
\hline
$[51]$ & $6\cdot A_d^2$ & $120$ & $aaa|aab$ & $6\cdot A_d^2$ & $(3,2)$ & $(1-d)(1+d^2)$ \\
\hline
$[6]$ & $1\cdot A_d^1$ & $720$ & $aaa|aaa$ & $1\cdot A_d^1$ & $(3,3)$ & $(1+d^2)^2$ \\
\hline
\hline
\end{tabular}
\caption{The classes and elements number of $\vec{a}$ for $T_\lambda$ and $Q_6(j_\lambda)$. The sub-type is determined by the weight of subsystem $a_1a_2a_3$ and $a_4a_5a_6$. $\#\{j_\lambda\}$ denote the number of elements contained in the sub-type $j_\lambda$.} \label{tab:classQ6}
\end{table}

Therefore,
\begin{equation}
\begin{aligned}
&\tr[\Phi^6(Q_6)\rho^{\otimes 6}]\leq \frac{(d-1)!}{(d+5)!}\sum_{\vec{a}\in \mbb{Z}_d^6} Q_6(\vec{a}) T(\vec{a}) \\
&=\frac{(d-1)!}{(d+5)!} \{ 4\cdot( A_d^6 + 2\cdot 9A_d^5 + 4\cdot 18A_d^4 + 8\cdot 6A_d^3 ) + 2(1-d)\cdot( 2\cdot 6A_d^5 + 4\cdot 18A_d^4 + 6\cdot 18A_d^4 + 12\cdot 36A_d^3) \\
&\quad+ (1-d)^2\cdot( 4\cdot 9 A_d^4 + 8\cdot 9A_d^3 + 12\cdot 18A_d^3 + 36\cdot 9 A_d^2 + 24\cdot 9 A_d^3 + 48\cdot 9 A_d^2) + 2(1+d^2)(6\cdot 2 A_d^4 + 24\cdot 6 A_d^3) \\
&\quad+ (1+d^2)(1-d)\cdot (12\cdot 6A_d^3 + 48\cdot 6A_d^2 + 120\cdot 6A_d^2) + (1+d^2)^2 \cdot (36\cdot A_d^2 + 720 \cdot A_d^1) \} \\
&=4\,\frac{d^4 + 59d^3+107d^2+109d+84}{d^4 + 14d^3 + 71d^2 + 154d + 120} < 10.
\end{aligned}
\end{equation}
\end{proof}

\begin{prop} \label{prop:Q3456gg0}
When $d\gg 1$, for the $Q_t$ defined in Eq.~\eqref{eq:Q3Q4Q5Q6}, we have
\begin{equation}
\sum_{\pi\in\xi} \tr[W_\pi Q_t] > 0,
\end{equation}
for all conjugacy classes $\xi$ in $S_t$.
\end{prop}

\begin{proof}
We have
\begin{equation}
\begin{aligned}
\sum_{\pi\in \xi} \tr[W_\pi Q_t] &= \sum_{\vec{a}\in\mbb{Z}_d^t} \sum_{\pi\in \xi} \tr[W_\pi \ketbra{\vec{a}}{\vec{a}}]  Q_t(\vec{a}) \\
&= \sum_{\lambda} \sum_{\vec{a}\in\lambda} \sum_{\pi\in \xi} \tr[W_\pi \ketbra{\vec{a}}{\vec{a}}]  Q_t(\vec{a}) \\
&= \sum_{\lambda} \gamma_{\xi,\lambda} \sum_{\vec{a}\in\lambda}Q_t(\vec{a}) \\
&= \sum_{\lambda} \gamma_{\xi,\lambda} F_\lambda(Q_t),
\end{aligned}
\end{equation}
where the second equality is due to Proposition~\ref{prop:gamma}, $F_{\lambda}(Q_t):=\sum_{\vec{a}\in\lambda} Q_t(\vec{a})$.

The values of $\{F_\lambda(Q_t)\}$ for all $\lambda$ and $Q_t$ have been listed in Table~\ref{tab:classQ3},\ref{tab:classQ4},\ref{tab:classQ5}, and \ref{tab:classQ6}. Below in Table~\ref{tab:gammaQ3}, \ref{tab:gammaQ4}, \ref{tab:gammaQ5}, and \ref{tab:gammaQ6}, we list all the embedding constants $\gamma_{\xi,\lambda}$, as the intrinsic property of the permutation group $S_t$. We also list the leading term with respect to $d$ in $\{F_\lambda(Q_t)\}$ when $t=3,4,5$ and $6$.

From the tables we can calculate the values of $\sum_{\pi\in\xi}\tr[W_\pi Q_t]$, which are all positive.
\end{proof}

\begin{table}[htbp]
\begin{tabular}{c|ccc}
\hline\hline
\diagbox{$\xi$}{$\lambda$} & $[111]$ & $[21]$ & $[3]$ \\ \hline
$[111]$  & $1$ & $1$ & $1$ \\ \hline
$[21]$  & $0$ & $1$ & $3$ \\ \hline
$[3]$  & $0$ & $0$ & $2\cdot 1$ \\ \hline\hline
Leading term $\sum_{\vec{a}\in \lambda} Q_3(\vec{a})$ & $4 d^3$ & $3d^4$ & $d^5$ 
\\
\hline
\end{tabular}
\caption{The embedding constants $\gamma_{\xi,\lambda}$ for permutation group $S_3$. } \label{tab:gammaQ3}
\end{table}

\begin{table}[htbp]
\begin{tabular}{c|ccccc}
\hline\hline
\diagbox{$\xi$}{$\lambda$} & $[1111]$ & $[211]$ & $[22]$ & $[31]$ & $[4]$ \\ \hline
$[1111]$  & $1$ & $1$ & $1$ & $1$ & $1$ \\ \hline
$[211]$  & $0$ & $1$ & $2$ & $3$ & $6$\\ \hline
$[22]$  & $0$ & $0$ & $1$ & $0$ & $3$ \\ \hline
$[31]$  & $0$ & $0$ & $0$ & $2\cdot 1$ & $2 \cdot 4$ \\ \hline
$[4]$  & $0$ & $0$ & $0$ & $0$ & $6 \cdot 1$ \\ \hline
\hline
Leading term $\sum_{\vec{a}\in \lambda} Q_4(\vec{a})$ & $4d^4$ & $d^5$ & $3d^4$ & $-2d^5$ & $d^5$
\\ \hline
\end{tabular}
\caption{The embedding constants $\gamma_{\xi,\lambda}$ for permutation group $S_4$.} \label{tab:gammaQ4}
\end{table}

\begin{table}[htbp]
\begin{tabular}{c|ccccccc}
\hline\hline
\diagbox{$\xi$}{$\lambda$} & $[11111]$ & $[2111]$ & $[221]$ & $[311]$ & $[32]$ & $[41]$ & $[5]$ \\ \hline
$[11111]$  & $1$ & $1$ & $1$ & $1$ & $1$ & $1$ & $1$\\ \hline
$[2111]$  & $0$ & $1$ & $2$ & $3$ & $4$ & $6$ & $10$\\ \hline
$[221]$  & $0$ & $0$ & $1$ & $0$ & $3$ & $3$ & $15$ \\ \hline
$[311]$  & $0$ & $0$ & $0$ & $2\cdot 1$ & $2 \cdot 1$ & $2\cdot 4$ & $2\cdot 10$ \\ \hline
$[32]$  & $0$ & $0$ & $0$ & $0$ & $2 \cdot 1$ & $0$ & $2\cdot 10$ \\ \hline
$[41]$ & $0$ & $0$ & $0$ & $0$ & $0$ & $6 \cdot 1$ & $6\cdot 5$\\ \hline
$[5]$ & $0$ & $0$ & $0$ & $0$ & $0$ & $0$ & $24\cdot 1$\\ \hline \hline
Leading term $\sum_{\vec{a}\in \lambda} Q_5(\vec{a})$ & $4d^5$ & $-12d^5$ & $5 d^5$ & $8 d^5$ & $-2 d^5$ & $-4 d^5$ & $d^5$ \\
\hline
\end{tabular}
\caption{The embedding constants $\gamma_{\xi,\lambda}$ for permutation group $S_5$.} \label{tab:gammaQ5}
\end{table}

\begin{table}[htbp]
\begin{tabular}{c|ccccccccccc}
\hline\hline
\diagbox{$\xi$}{$\lambda$} & $[111111]$ & $[21111]$ & $[2211]$ & $[222]$ & $[3111]$ & $[321]$ & $[33]$ & $[411]$ & $[42]$ & $[51]$ & $[6]$ \\ \hline
$[111111]$  & $1$ & $1$ & $1$ & $1$ & $1$ & $1$ & $1$ & $1$ & $1$ & $1$ & $1$ \\ \hline
$[21111]$  & $0$ & $1$ & $2$ & $3$ & $3$ & $4$ & $6$ & $6$ & $7$ & $10$ & $15$ \\ \hline
$[2211]$  & $0$ & $0$ & $1$ & $3$ & $0$ & $3$ & $9$ & $3$ & $9$ & $15$ & $45$ \\ \hline
$[222]$  & $0$ & $0$ & $0$ & $1$ & $0$ & $0$ & $0$ & $0$ & $3$ & $0$ & $15$ \\ \hline
$[3111]$  & $0$ & $0$ & $0$ & $0$ & $2 \cdot 1$ & $2\cdot 1$ & $2\cdot 2$ & $2\cdot 4$ & $2\cdot 4$ & $2\cdot 10$ & $2\cdot 20$\\ \hline
$[321]$  & $0$ & $0$ & $0$ & $0$ & $0$ & $2 \cdot 1$ & $2 \cdot 6$ & $0$ & $2 \cdot 4$ &  $2\cdot 10$ & $2 \cdot 60$ \\ \hline
$[33]$  & $0$ & $0$ & $0$ & $0$ & $0$ & $0$ & $4\cdot 1$ & $0$ & $0$ &  $0$ & $4 \cdot 10$ \\ \hline
$[411]$ & $0$ & $0$ & $0$ & $0$ & $0$ & $0$ & $0$ & $6 \cdot 1$ & $6 \cdot 1$ & $6\cdot 5$ & $6\cdot 15$\\ \hline
$[42]$ & $0$ & $0$ & $0$ & $0$ & $0$ & $0$ & $0$ & $0$ & $6 \cdot 1$ & $0$ & $6\cdot 15$\\ \hline
$[51]$ & $0$ & $0$ & $0$ & $0$ & $0$ & $0$ & $0$ & $0$ & $0$ & $24\cdot 1$ & $24\cdot 6$\\ \hline
$[6]$  & $0$ & $0$ & $0$ & $0$ & $0$ & $0$ & $0$ & $0$ & $0$ & $0$ & $120\cdot 1$\\ \hline\hline
Leading term $\sum_{\vec{a}\in \lambda} Q_6(\vec{a})$ & $4d^6$ & $-12d^6$ & $9 d^6$ & $9d^5$ & $4d^6$ & $-6d^6$ & $d^6$ & $12d^5$ & $-6d^5$ & $-6d^5$ & $d^5$ \\
\hline
\end{tabular}
\caption{The embedding constants $\gamma_{\xi,\lambda}$ for permutation group $S_6$.} \label{tab:gammaQ6}
\end{table}

With the tables above and the values of $F_\lambda(Q_t)$, we can bound the values of $\tr[\Phi^t(Q_t)\rho^{\otimes t}]$ tighter, using the information of $\rho$.

\begin{prop} \label{prop:trQ3456}
When $d\gg 1$, and the rank of $\rho$ is constant with respect to $d$, for the $Q_t$ defined in Eq.~\eqref{eq:Q3Q4Q5Q6}, we have
\begin{equation}
\begin{aligned}
\tr[\Phi^3(Q_3)\rho^{\otimes 3}] &\sim \frac{1}{d^3} \left\{ d^5 + 3d^5 \tr[\rho^2] + 2d^5 \tr[\rho^3] \right\}, \\
\tr[\Phi^4(Q_4)\rho^{\otimes 4}] &\sim \frac{1}{d^4} \left\{ d^5 \tr[\rho^2] + 3d^5 \tr[\rho^2]^2 + 4d^5 \tr[\rho^3] + 6d^5\tr[\rho^4] \right\}, \\
\tr[\Phi^5(Q_5)\rho^{\otimes 5}] &\sim \frac{1}{d^5} \left\{ 2d^5 \tr[\rho^2]^2 + 16d^5 \tr[\rho^2]\tr[\rho^3] + 6d^5 \tr[\rho^4] + 24d^5\tr[\rho^5] \right\}, \\
\tr[\Phi^6(Q_6)\rho^{\otimes 6}] &\sim \frac{1}{d^6} \left\{ 4d^6\tr[\rho^3]^2 \right\}.
\end{aligned}
\end{equation}
for all the possible partition $\xi$ of $S_t$.
\end{prop}

\begin{proof}

When $d\gg 1$, we have
\begin{equation} \label{eq:PhitQrhotdgg1}
\begin{aligned}
    \tr[\Phi^t(Q)\rho^{\otimes t}] &= \sum_{\pi,\sigma\in S_t} C_{\pi,\sigma} \tr[W_\pi Q]\tr[W_\sigma \rho^{\otimes t}] \\
    &= \sum_{\pi,\alpha\in S_t} C_{\pi,\alpha\pi^{-1}} \tr[W_\pi Q]\tr[W_{\alpha\pi^{-1}} \rho^{\otimes t}] \\
    &= \sum_{\alpha\in S_t} Wg(\alpha,d) \sum_{\pi\in S_t} \tr[W_\pi Q]\tr[W_{\alpha\pi^{-1}} \rho^{\otimes t}]
\end{aligned}
\end{equation}
Here, $\alpha = \sigma\pi$. $Wg(\alpha,d)$ is the Weingarten function defined in Section~\ref{Sec:preliminaries}. When $d\gg 1$, the Weingarten function can be expanded as
\begin{equation} \label{eq:WgExpanded}
Wg(\alpha,d) = d^{k(\alpha)-2t} \prod_{i=1}^{k(\alpha)} (-1)^{\xi_i -1} Ca_{\xi_i - 1} + \mc{O}(d^{k(\alpha)-2t-2}),
\end{equation}
where $k(\alpha)$ is the cycle number of $\alpha$, $\xi(\alpha) = (\xi_1, \xi_2, ... \xi_k)$ is the partition of $\alpha$. $Ca_q:= \frac{(2q)!}{q!(q+1)!}$ is the Catalan number.

From Eq.~\eqref{eq:WQ3456order} in the proof of Proposition~\ref{prop:Q3456order}, we know that the highest rank of $\sum_{\pi\in S_t} \tr[W_\pi Q]\tr[W_{\alpha\pi} \rho^{\otimes t}$ is $\mc{O}(d^5)$ when $t=3,4,5$ and less than $\mc{O}(d^6)$ when $t=6$, regardless of $\alpha$. Later we will show that when $\alpha = I$, the highest order can be reached. Therefore, without loss of generality, we consider the leading term when $\alpha=I$.

From Eq.~\eqref{eq:WgExpanded}, when $d\gg 1$, the leading term in Eq.~\eqref{eq:PhitQrhotdgg1} is
\begin{equation} \label{eq:PhitQrhotdgg1sim}
\begin{aligned}
    \tr[\Phi^t(Q)\rho^{\otimes t}] &\sim Wg(I,d) \sum_{\pi\in S_t} \tr[W_\pi Q]\tr[W_{\pi^{-1}} \rho^{\otimes t}] \\
    &= d^{-t} \sum_{\pi\in S_t} \tr[W_\pi Q]\tr[ \rho^{\otimes t}W_{\pi}] \\
    &= d^{-t} \sum_{\xi} H_{\xi}(\rho) \sum_{\pi\in \xi} \tr[W_\pi Q].
\end{aligned}
\end{equation}
Here, the first equation is due to Eq.~\eqref{eq:WgExpanded}. The second equation is because the value of $\tr[W_{\pi} \rho^{\otimes t}]$ only depends on the partition of $\pi$,
\begin{equation} \label{eq:Cxirho}
\tr[W_{\pi} \rho^{\otimes t}] = \prod_{i=1}^{k(\pi)} \tr[\rho^{\xi_i(\pi)}] := H_{\xi}(\rho).
\end{equation}

In our case, the $Q$ observables we care about are the ones defined in Eq.~\eqref{eq:Q3Q4Q5Q6}. From Proposition~\ref{prop:Q3456gg0}, we know that $\sum_{\pi\in \xi} \tr[W_\pi Q] >0$ for all possible partition $\xi$. Then from Eq.~\eqref{eq:PhitQrhotdgg1sim} we have
\begin{equation}
\begin{aligned}
\tr[\Phi^t(Q)\rho^{\otimes t}] &\sim d^{-t} \sum_{\xi} H_{\xi}(\rho) \sum_{\pi\in \xi} \tr[W_\pi Q] \\
&\leq d^{-t} \sum_{\xi} \sum_{\pi\in \xi}\tr[W_\pi Q] \\
&= \frac{1}{d^t}\sum_{\xi} H_\xi(\rho) \sum_{\lambda} \gamma_{\xi,\lambda} \sum_{\vec{a}\in\lambda} Q_t(\vec{a}) \\
&= \frac{1}{d^t}\sum_{\xi,\lambda} \gamma_{\xi,\lambda} H_{\xi}(\rho) F_{\lambda}(Q_t),
\end{aligned}
\end{equation}
Here, $F_{\lambda}(Q_t):=\sum_{\vec{a}\in\lambda} Q_t(\vec{a})$. In the first inequality, we use the fact that $H_{\xi}(\rho)\leq 1$. The first equality is because
\begin{equation}
\sum_{\pi\in \xi} \tr[W_\pi Q_t] = \sum_{\lambda} \gamma_{\xi,\lambda} \sum_{\vec{a}\in\lambda}Q_t(\vec{a}).
\end{equation}
Note that the coefficient $H_\xi(\rho)\leq 1$ for all $\xi$, which is irrelevant of the dimension $d$. Therefore, when $d\gg 1$, we only need to consider the leading terms in $F_\lambda(Q_t)$.

From the Tables~\ref{tab:gammaQ3}, \ref{tab:gammaQ4}, \ref{tab:gammaQ5}, and \ref{tab:gammaQ6}, we can calculate the leading term of $\tr[\Phi^t(Q_t)\rho^{\otimes t}]$ with respect to $d$.
\end{proof}

In the discussion above, we assume the rank of $\rho$ to be small and independent of $d$. This requirement simplifies the discussion, since $\tr[W_\pi\rho^{\otimes t}]$ is then not related to $d$. In the general case when the rank of $\rho$ is not a constant, one can still bound the order of $d$ of each variance term.

\begin{prop} \label{prop:Q3456order}
When $d\gg 1$, for the $Q_t$ defined in Eq.~\eqref{eq:Q3Q4Q5Q6}, the asymptotic relation with respect to $d$ is
\begin{equation}
\begin{aligned}
\tr[\Phi^3(Q_3)\rho^{\otimes 3}] &= \mc{O}(d^2), \\
\tr[\Phi^4(Q_4)\rho^{\otimes 4}] &= \mc{O}(d), \\
\tr[\Phi^5(Q_5)\rho^{\otimes 5}] &= \mc{O}(1), \\
\tr[\Phi^6(Q_6)\rho^{\otimes 6}] &= \mc{O}(1).
\end{aligned}
\end{equation}
\end{prop}

\begin{proof}
Recall that
\begin{equation}
    \tr[\Phi^t(Q_t)\rho^{\otimes t}] = \sum_{\alpha\in S_t} Wg(\alpha,d) \sum_{\pi\in S_t} \tr[W_\pi Q_t]\tr[W_{\alpha\pi} \rho^{\otimes t}],
\end{equation}
with $Wg(\alpha,d) = d^{k(\alpha)-2t} \prod_{i=1}^{k(\alpha)} (-1)^{\xi_i -1} Ca_{\xi_i - 1} + \mc{O}(d^{k(\alpha)-2t-2})$. The highest $d$-order of $\tr[\Phi^t(Q_t)\rho^{\otimes t}]$ is bounded by the multiplication of the highest $d$-order of $Wg(\alpha,d)$, $\tr[W_\pi Q_t]$, and $\tr[W_\sigma \rho^{\otimes t}]$. We have already known that
\begin{equation}
Wg(\alpha,d) = \mc{O}(d^{-t}), \quad \tr[W_\sigma \rho^{\otimes t}] = \mc{O}(1).
\end{equation}
Therefore, we only need to bound the highest $d$-order of $\tr[W_\pi Q_t]$. Note that
\begin{equation}
\begin{aligned}
\tr[W_\pi Q_t] &= \sum_{\vec{a}\in\mbb{Z}_d^t} Q_t(\vec{a}) \tr[W_\pi \ketbra{\vec{a}}{\vec{a}}] \\
&= \sum_{\vec{a}\in\mbb{Z}_d^t} Q_t(\vec{a}) \mathbbm{1}[\pi\subseteq \omega(\vec{a})] \\
&= \sum_{\lambda}\sum_{j_\lambda} Q_t(j_\lambda) \sum_{\vec{a}\in j_\lambda} \mathbbm{1}[\pi\subseteq \omega(\vec{a})] \\
&\leq \sum_{\lambda}\sum_{j_\lambda} \#\{j_\lambda\} Q_t(j_\lambda) \\
&= \sum_{\lambda} F_\lambda(Q_t).
\end{aligned}
\end{equation}
Here, $F_\lambda(Q_t):= \sum_{\vec{a}\in\lambda} Q_t(\vec{a})$. From Table~\ref{tab:gammaQ3},\ref{tab:gammaQ4},\ref{tab:gammaQ5}, and \ref{tab:gammaQ6}, we have
\begin{equation}  \label{eq:WQ3456order}
\begin{aligned}
\tr[W_\pi Q_3] \leq \sum_{\lambda} F_\lambda(Q_3) &= \mc{O}(d^5), \\
\tr[W_\pi Q_4] \leq \sum_{\lambda} F_\lambda(Q_4) &= \mc{O}(d^5), \\
\tr[W_\pi Q_5] \leq \sum_{\lambda} F_\lambda(Q_5) &= \mc{O}(d^5), \\
\tr[W_\pi Q_6] \leq \sum_{\lambda} F_\lambda(Q_6) &= \mc{O}(d^6).
\end{aligned}
\end{equation}
Combining this with $Wg(\alpha,d)=\mc{O}(d^{-t})$ and $\tr[W_\sigma \rho^{\otimes t}]$, we finish the proof.
\end{proof}

\begin{prop} \label{Prop:3varianceExactBound}
For observable $O_+ \in \mc{L}((\mc{H}^A)^{\otimes 3})$ with the form $O_+ = \sum_{\vec{a}\in \mbb{Z}_d^3} [1 + (-d)^{wt(\vec{a})-1}] \ketbra{\vec{a}}{\vec{a}}$, when the random unitaries are chosen within unitary $3$-design, we have
\begin{equation}
\tr[ \Phi^3(O_+^2) \rho^{\otimes 3} ] = (d+2)^{-1} \{ (d+1)(d^2+3d+4) + 3d(d-1)(d+1)\tr[\rho^2] + 2(d^3-d^2+6)\tr[\rho^3] \}
\end{equation}
for all states $\rho\in\mc{D}(\mc{H}^{A})$.
\end{prop}
\begin{proof}
By applying the Weingarten integral, we have
\begin{equation}
\begin{aligned}
    \tr[\Phi^t(O_+^2)\rho^{\otimes t}] &= \sum_{\pi,\sigma\in S_t} C_{\pi,\sigma} \tr[W_\pi O_+^2]\tr[W_\sigma \rho^{\otimes t}],
\end{aligned}
\end{equation}
where $C_{\pi,\sigma}$ is the Weingarten matrix of $S_3$ group. With a direct calculation, we finish the proof.
\end{proof}

In the negativity detection, we also need the twirling results for local random Clifford gates $U_A\otimes V_B$ for $Q\in\mc{L}((\mc{H}^{AB})^{\otimes t})$. In general, it is much harder to calculate the variance terms $\{\Delta_t\}$ defined in Eq.~\eqref{eq:Delta3456}. Here, we mainly consider two cases: 1) the underlying state is a pure tensor state. 2) the asymptotic case $d\gg N_M \gg 1$. We have the following propositions.

\begin{prop} \label{Prop:varianceboundAB}
For observable $O_{++} \in \mc{L}((\mc{H}^{AB})^{\otimes 3})$ with the form
\begin{equation}
O_{++} = O_{+}^A \otimes O_{+}^B = \sum_{\vec{a},\vec{b}\in \mbb{Z}_d^3} [1 + (-d)^{wt(\vec{a})-1}][1 + (-d)^{wt(\vec{b})-1}] \ketbra{\vec{a}}{\vec{a}}\otimes\ketbra{\vec{b}}{\vec{b}},
\end{equation}
when the random unitaries are chosen in $\mc{E}_A \times \mc{E}_B$, where $\mc{E}_A,\mc{E}_B$ are unitary $6$-designs, when the underlying state $\rho\in\mc{D}(\mc{H}^{A})$ is a pure tensor state, we have
\begin{equation}
\begin{aligned}
\Delta_6 = &\tr[(\Phi^6_A\otimes \Phi^6_B)(O_{++}^{\otimes 2}) \rho^{\otimes 6}] = \Gamma_6^2(\psi,O_{+}^{\otimes 2},\mc{E}) < 10^2, \\
\Delta_5 = &\tr[(\Phi^5_A\otimes \Phi^5_B)(O_{123,145}^{AB}) \rho^{\otimes 5}] = \Gamma_5^2(\psi,O_{123,145},\mc{E}) < 48^2, \\
\Delta_4 = &\tr[(\Phi^4_A\otimes \Phi^4_B)(O_{123,124}^{AB}) \rho^{\otimes 4}] = \Gamma_4^2(\psi,O_{123,124},\mc{E}) < (14d)^2,  \\
\Delta_3 = &\tr[(\Phi^3_A\otimes \Phi^3_B)(O_{++}^2) \rho^{\otimes 3}] = \Gamma_3^2(\psi,O_{+}^2,\mc{E}) < (6d^2)^2,
\end{aligned}
\end{equation}
where $O_{123,145}^{AB}$ and $O_{123,124}^{AB}$ is defined in Eq.~\eqref{eq:O123145AB}, and $\{\Gamma_t\}$ are defined in Eq.~\eqref{eq:Gamma3456}. $\psi$ is a pure state.
\end{prop}

\begin{proof}
Here, we slightly modify the derivation in Eq.~\eqref{eq:PhitQrhot},
\begin{equation} \label{eq:PhiABtQrhoABt}
\begin{aligned}
    \tr[(\Phi^t_A \otimes \Phi^t_B)(Q_t^A\otimes Q_t^B)\rho_{AB}^{\otimes t}] &= \sum_{\pi,\sigma,\alpha,\beta\in S_t} C_{\pi,\sigma} C_{\alpha,\beta} \tr[(W_\pi^A \otimes W_\alpha^B) (Q_t^A\otimes Q_t^B)]\tr[(W_\sigma^A \otimes W_\beta^B)\rho_{AB}^{\otimes t}] \\
    &= \sum_{\pi,\alpha \in S_t} \tr[(W_\pi^A \otimes W_\alpha^B)(Q_t^A\otimes Q_t^B)] \left(\sum_{\sigma \in S_t} C_{\pi,\sigma}\right)\left(\sum_{\beta \in S_t} C_{\alpha,\beta}\right)  \\
    &=  \left(\frac{(d-1)!}{(d+t-1)!}\right)^2 \sum_{\pi \in S_t} \tr[(W_\pi^A \otimes W_\alpha^B)(Q_t^A\otimes Q_t^B)] \\
    &= \Gamma_t^2(\psi,Q_t,\mc{E}).
\end{aligned}
\end{equation}
Here, the second equality is because $\tr[(W_\sigma^A \otimes W_\beta^B)\rho_{AB}^{\otimes t}]=1$ when $\rho_{AB} = \psi_A \otimes \psi_B$.
\end{proof}
In the proposition above, we consider the pure separable state. For a mixed product state $\rho_{AB}=\rho_A\otimes \rho_B$, the term $\tr[(W_\pi^A \otimes W_\alpha^B)(Q_t^A\otimes Q_t^B)]$ is still decoupled, same as the mixed state in the 3-order purity case. The general separable state is just a convex mixture. Thus, similar to the $3$-order purity results, the more mixed the state is, the smaller variance is.

In general, the state $\rho_{AB}$ is entangled. In this case, the absolute value of $\tr[(W_\sigma^A \otimes W_\beta^B)\rho_{AB}^{\otimes t}]$ is still bounded by $1$.
\begin{prop} \label{prop:WABbounded1}
For any $\rho_{AB}\in\mc{D}(\mc{H}^{AB})$ and $\pi,\sigma\in S_t$,
\begin{equation} 
\begin{aligned}
    \left|\tr[(W_\pi^A \otimes W_\sigma^B)\rho_{AB}^{\otimes t}]\right| \leq 1.
\end{aligned}
\end{equation}
\end{prop}

\begin{proof}
Any bipartite mixed state can be written in the convex decomposion $\rho_{AB}=\sum_j p_j \Psi_j$, where $\Psi_j$ is pure state and $p_j$ is the corresponding probability. For $t$-copies of $\rho_{AB}$, we have
\begin{equation}
\rho_{AB}^{\otimes t} = \sum_{\vec{j}} p_{\vec{j}} \bigotimes_{k=1}^{t} \Psi_{j[k]}.
\end{equation}

As a result, we only need to prove that any term $\tr[(W_\pi^A \otimes W_\sigma^B) \bigotimes_{k=1}^t\Psi_k] \leq 1$, where $\Psi_k$ can be any bipartite pure state. For a given pure state $\Psi_k$, the Schmidt decomposition shows
\begin{equation}
\Psi = \sum_{i=1}^d g_i\ket{\psi_i}^A\ket{\phi_i}^B=\sum_{i=1}^d g_i U_{\psi}^A\ket{i}^A U_{\phi}^B\ket{i}^B,
\end{equation}
where $g_i$ is the positive coefficient and $\{\ket{\psi_i}^A\},\{\ket{\phi_i}^B\}$ are orthogonal bases on $A,B$ respectively, which can be transformed from computational bases by $U_{\psi}^A, U_{\phi}^B$. In fact, the state can be written in a more compact form $\Psi=\sum_{i=1}^d g_i U^B\ket{i}^A\ket{i}^B$, where $U=U_{\phi}U_{\psi}^T$ operating on subsystem $B$.
We denote $D=GU$ and $D'=D^\dag = U^{\dag}G$.

$\bigotimes_{k=1}^t\Psi_k$ is a tensor-ed pure state on the $t$-copy Hilbert space. Similar to Eq.~\eqref{eq:innerWAWBBell}, we may use the Bell-state trick to simplify the equation,
\begin{equation}
\begin{aligned}
    \tr[(W_\pi^A \otimes W_\sigma^B) \bigotimes_{k=1}^t\Psi_k] &= \tr[(W_\pi^A \otimes W_\sigma^B)(\bigotimes_{k=1}^t D_k)(\tilde{\Psi}_+^{\otimes t})(\bigotimes_{k=1}^t D_k')] \\
    &= \tr\left[ W_{\alpha} \left(D'_1 D_{\sigma(1)} \otimes D'_2 D_{\sigma(2)} \cdots \otimes D'_t D_{\sigma(t)} \right) \right]
\end{aligned}
\end{equation}
where $\alpha=\sigma\pi^{-1}$ together is some permutation on the $t$-copy space, $\tilde{\Psi}_+ = \sum_{i,j=1}^{d}\ket{ii}\bra{jj}$ is the unnormalized Bell state. The final result depends on the cycle structure of $\alpha$. For example, for $t=3$ and $\alpha=(12)(3)$, the result is $\tr[D'_1 D_{\sigma(1)} D'_2 D_{\sigma(2)}]\tr[D'_3 D_{\sigma(3)}]$. Thus, to prove that the total value is less than $1$, one only need to prove the the absolute value of each term in a cycle, e.g., $\tr[D'_1 D_{\sigma(1)}D'_2D_{\sigma(2)}]$ is less than $1$.

Here we show that $|\tr[D_1' D_2 D_3' D_4]|\leq 1$ and other terms can be proved similarly.
\begin{equation}
\begin{aligned}
    |\tr[D_1' D_2 D_3' D_4]| &= |\tr[U_1^{\dag} G_1 G_2 U_2 U_3^{\dag} G_3 G_4 U_4]| \\
    &=|\tr[U G_1 G_2 V G_3 G_4]| \\
    &=\left|\sum_{i,j} \bra{j} U \ket{i} G_1^i G_2^i\bra{i} V \ket{j}G_3^j G_4^j  \right|\\
    &\leq \sum_{i,j} G_1^i G_2^i G_3^j G_4^j |\bra{i}U\ket{j}| |\bra{j}V\ket{i}|\\
    &\leq \sum_{i,j} G_1^i G_2^i G_3^j G_4^j \\
    &\leq  \sqrt{\sum_i{G_1^i}^2}\sqrt{\sum_i{G_2^i}^2} \sqrt{\sum_j{G_3^j}^2}\sqrt{\sum_j{G_4^j}^2}\leq 1.
\end{aligned}
\end{equation}
Here in the seond line, we denote $U=U_4U_1^{\dag}$, $V=U_2U_3^{\dag}$, the second inequality is beacuse of the transition probability $|\bra{i}U\ket{j}| |\bra{j}V\ket{i}|\leq 1$, the third one is based on Cauchy-Schwarz inequality, and the last one is just the sub-normalization requirement on the coefficient matrices $\{G_k\}$.
\end{proof}

Before we get further to estimate $\{\Delta_t\}$ for general $\rho_{AB}$, we first study the term $\tr[W_\pi^A\otimes W_{\sigma}^B \rho_{AB}^{\otimes t}]$. Here we would like to consider how the entanglement of the state affects the variance and focus on a general bipartite pure state $\Psi_{AB}$.

\begin{prop} \label{prop:PureABSchmidt}
For a pure bipartite state $\Psi_{AB}$, the Schmidt decomposition is
\begin{equation}
\ket{\Psi}_{AB} = \sum_{i=1}^{d} g_i U^B \ket{i}_A \ket{i}_B = G U^B \sum_{i=1}^d \ket{ii}_{AB} = D \ket{\tilde{\Psi}_+}_{AB},
\end{equation}
where $G = diag\{g_1,g_2,...\}$ is the Schmidt coefficient matrices, $U^B$ is a local unitary on system $B$, and $\ket{\tilde{\Psi}_+}_{AB}:= \sum_{i=1}^d \ket{ii}_{AB}$ is the unnormalized Bell state. Then
\begin{equation}
\tr[(W^A_\pi\otimes W^B_\sigma) \Psi^{\otimes t}_{AB}] = \prod_{l=1}^{\# cycles(\beta)} \tr[\Lambda^{\xi_l(\beta)}]:= \chi(\Psi,\beta),
\end{equation}
where $\beta = \sigma\pi^{-1}$, $\Lambda = G^\dag G = diag\{p_1,p_2,...\}$ is the Schmidt probabilities, and $\xi(\beta)$ is the cycle structure of $\beta\in S_t$.

Moreover, $\chi(\Psi,\beta)$ is bounded by the \renyi-entropy,
\begin{equation} \label{eq:renyibound}
d_0(\Psi)^{\#cycles(\beta)-t} \leq \chi(\Psi,\beta) \leq d_t(\Psi)^{\#cycles(\beta)-t},
\end{equation}
where $d_\alpha(\Psi) := 2^{S_\alpha(\Psi)}$, and $S_\alpha(\Psi)$ is the $\alpha$-\renyi entropy,
\begin{equation} \label{eq:renyi}
\begin{aligned}
    S_{\alpha}(\Psi)= \frac1{1-\alpha} \log \sum_i {p_i^{\alpha}}, \quad \alpha>0; \alpha\neq 1.
\end{aligned}
\end{equation}
\end{prop}

\begin{proof}
We have
\begin{equation} 
\begin{aligned}
    \tr[(W_\pi^A\otimes W_{\sigma}^B) \Psi_{AB}^{\otimes t}] &= \tr[(W_\pi^A \otimes W_\sigma^B)(\bigotimes_{k=1}^t D_k)(\tilde{\Psi}_+^{\otimes t})(\bigotimes_{k=1}^t D_k')] \\
    &= \tr\left[ W_{\beta} \left(D'_1 D_{\sigma(1)} \otimes D'_2 D_{\sigma(2)} \cdots \otimes D'_t D_{\sigma(t)} \right) \right]
    &= \tr[W_{\beta} D^{\otimes t}]\\
    &= \prod_{l=1}^{\# cycles(\beta)} \tr[\Lambda^{\xi_l(\beta)}] = \chi(\Psi,\beta),
\end{aligned}
\end{equation}
where $\beta=\sigma\pi^{-1}$, $\Lambda=G^\dag G$ and $\lambda_i$ denote the cycle length. It is clear that the value depends on $\Lambda=\{p_i\}$ thus the state $\Psi_{AB}$. Denote the Schmidt rank of $\Psi_{AB}$ (the rank of $\Lambda$) as $d_0(\Psi)$, which is related to $S_0(\Psi)$ by $d_0 = 2^{S_0}$; denote the $t$-rank $d_t:= 2^{S_t}$, where $S_t$ is the \renyi-$t$ entropy. From the definition we obtain Eq.~\eqref{eq:renyibound},
\begin{equation}
d_0(\Psi)^{\#cycles(\beta)-t} \leq \chi(\Psi,\beta) \leq d_t(\Psi)^{\#cycles(\beta)-t}.
\end{equation}
Note that the inequalities are saturated when the spectrum of $\Lambda$ is flat, that is, $\Psi_{AB}$ is the Bell state in dimension $d_{\Psi}$. In fact, one can make the bound tighter by considering the $d_{\min}=2^{S_k(\Psi)}$ with $k=\max\{\xi_l(\beta)\}$; $d_{\max}=2^{S_k(\Psi)}$ with $k=\min \{\xi_l(\beta)\}$ but exclude $1$,
\begin{equation} 
\begin{aligned}
  d_{\max}^{(\#cycles(\beta)-t)} \leq \chi(\Psi,\beta) \leq d_{\min}^{(\#cycles(\beta)-t)}.
\end{aligned}
\end{equation}

\end{proof}

When $\Psi_{AB} = \psi_A\otimes \psi_B$, $d_0 = d_t =1$, then $\chi(\Psi,\beta) = 1$. On the other hand, when $\Psi_{AB} = \Psi_+$ is the Bell state, $d_0 = d_t = d$, then $\chi(\Psi,\beta) = d^{\#cycles(\beta)-t}$. For a generic mixed state, the value of $\chi(\rho,\beta)$ is
\begin{equation}
\begin{aligned}
    &\tr[(W_{\pi}^A \otimes W_{\sigma}^B ) \rho^{\otimes t}_{AB}] \\
    &= \sum_{\vec{j}} p_{\vec{j}} \tr[(W_{\pi}^A \otimes W_{\sigma}^B) \Psi_{j[1]} \otimes \Psi_{j[2]}\otimes ...\otimes \Psi_{j[t]} ] \\
    &= \sum_{\vec{j}} p_{\vec{j}} \tr[(W_{\pi}^A \otimes W_{\sigma}^B) ( D_{j[1]} \otimes D_{j[2]}\otimes ...\otimes D_{j[t]} ) \tilde{\Psi}_+^{\otimes t} ( D_{j[1]}' \otimes D_{j[2]}'\otimes ...\otimes D_{j[t]}' ) ] \\
    &= \sum_{\vec{j}} p_{\vec{j}} \tr\left[W_{\sigma\pi^{-1}} ( D_{j[1]}' D_{j[\sigma(1)]} ) \otimes ( D_{j[2]}' D_{j[\sigma(2)]} )\otimes ...\otimes ( D_{j[t]}' D_{j[\sigma(t)]} ) \right] \\
    &:= \chi(\rho_{AB},\beta),
\end{aligned}
\end{equation}
where $\beta=\alpha\sigma^{-1}$. Therefore, $\chi(\rho,\beta)$ only depends on the cycle structure of $\beta$. It is like an ``averaged'' version of the \renyi-entropy.

Now we study the asymptotic property of $\{\Delta_t\}$ when $d\gg 1$.
\begin{prop} \label{prop:Q3456orderAB}
When $d\gg 1$, for the $Q_t$ defined in Eq.~\eqref{eq:Q3Q4Q5Q6}, the asymptotic relation with respect to $d$ is
\begin{equation} \label{eq:DeltaUpperbound}
\begin{aligned}
\Delta_3 = \tr[(\Phi^3_A(Q_3^A)\otimes \Phi^3_B(Q_3^B))\rho^{\otimes 3}_{AB}] &= \mc{O}(d^4), \\
\Delta_4 = \tr[(\Phi^4_A(Q_4^A)\otimes \Phi^4_B(Q_4^B))\rho^{\otimes 4}_{AB}] &= \mc{O}(d^2), \\
\Delta_5 = \tr[(\Phi^5_A(Q_5^A)\otimes \Phi^5_B(Q_5^B))\rho^{\otimes 5}_{AB}] &= \mc{O}(1), \\
\Delta_6 = \tr[(\Phi^6_A(Q_6^A)\otimes \Phi^6_B(Q_6^B))\rho^{\otimes 6}_{AB}] &= \mc{O}(1).
\end{aligned}
\end{equation}
\end{prop}

\begin{proof}
Using Weingarten integral, we have
\begin{equation} \label{eq:DeltatExpand}
\begin{aligned}
    \Delta_t &= \tr[(\Phi^t_A(Q_t^A)\otimes \Phi^t_B(Q_t^B))\rho^{\otimes t}_{AB}] \\
    &= \sum_{\sigma,\beta\in S_t} Wg(\sigma,d)Wg(\beta,d) \sum_{\pi,\alpha\in S_t} \tr[W_\pi Q_t]\tr[W_\alpha Q_t] \tr[ (W_{\sigma\pi^{-1}}^A \otimes W_{\beta\alpha^{-1}}^B ) \rho^{\otimes t}_{AB}].
\end{aligned}
\end{equation}
The highest $d$-order of $\Delta_t$ is bounded by the multiplication of the highest $d$-order of $Wg(\alpha,d)$, $\tr[W_\pi Q_t]$, and $\tr[W_\sigma \rho^{\otimes t}]$. We have already known that $Wg(\alpha,d) = \mc{O}(d^{-t})$. From Proposition~\ref{prop:Q3456order} we know that
\begin{equation}
\begin{aligned}
\tr[W_\pi Q_3] &= \mc{O}(d^5), \tr[W_\pi Q_4] = \mc{O}(d^5), \\
\tr[W_\pi Q_5] &= \mc{O}(d^5), \tr[W_\pi Q_6] = \mc{O}(d^6).
\end{aligned}
\end{equation}
Moreover, from Proposition~\ref{prop:WABbounded1} we have $\tr[(W_{\pi}^A \otimes W_{\alpha}^B ) \rho^{\otimes t}_{AB}] = \mc{O}(1)$. Combine all this results, we obtain Eq.~\eqref{eq:DeltaUpperbound}.
\end{proof}

\begin{prop} \label{prop:Q3456orderABBell}
When $d\gg 1$, for the $Q_t$ defined in Eq.~\eqref{eq:Q3Q4Q5Q6}, for the Bell state $\Psi_+$, the asymptotic relation with respect to $d$ is
\begin{equation}
\begin{aligned}
\Delta_3 = \tr[(\Phi^3_A(Q_3^A)\otimes \Phi^3_B(Q_3^B))\Psi^{\otimes 3}_+] &= \Theta(d^4), \\
\Delta_4 = \tr[(\Phi^4_A(Q_4^A)\otimes \Phi^4_B(Q_4^B))\Psi^{\otimes 4}_+] &= \Theta(d^2), \\
\Delta_5 = \tr[(\Phi^5_A(Q_5^A)\otimes \Phi^5_B(Q_5^B))\Psi^{\otimes 5}_+] &= \Theta(1), \\
\Delta_6 = \tr[(\Phi^6_A(Q_6^A)\otimes \Phi^6_B(Q_6^B))\Psi^{\otimes 6}_+] &= \Theta(1).
\end{aligned}
\end{equation}
\end{prop}

\begin{proof}
Since we already have Proposition~\ref{prop:Q3456orderAB}, we only need to prove that
\begin{equation} \label{eq:DeltaLowerbound}
\begin{aligned}
\Delta_3 &= \Omega(d^4), \quad \Delta_4 = \Omega(d^2), \\
\Delta_5 &= \Omega(1), \quad \Delta_6 = \Omega(1).
\end{aligned}
\end{equation}

Without loss of generality, we first choose the term with $\sigma=\beta=I=()$ in Eq.~\eqref{eq:DeltatExpand}. Later we show that this term is indeed the term with leading order of $d$.
\begin{equation} 
\begin{aligned}
    \Delta_t &\sim Wg(I,d)Wg(I,d) \sum_{\pi,\alpha\in S_t} \tr[W_\pi Q_t]\tr[W_\alpha Q_t] \tr[ (W_{\pi}^A \otimes W_{\alpha}^B ) \Psi^{\otimes t}_+].
\end{aligned}
\end{equation}

From Proposition~\ref{prop:PureABSchmidt} we know that,
\begin{equation}
\chi(\Psi_+,\beta) = d^{\#cycles(\beta)-t}.
\end{equation}
When $\beta=I$, the value of $\chi$ takes the highest order with respect to $d$. Then
\begin{equation} 
\begin{aligned}
    \Delta_t &\sim Wg(I,d)Wg(I,d) \sum_{\pi\in S_t} \tr[W_\pi Q_t]\tr[W_\pi Q_t] \\
    &= \frac{1}{d^{-2t}} \sum_{\pi\in S_t} \left(\sum_{\vec{a}\in\mbb{Z}_d^t} Q_t(\vec{a}) \mathbbm{1}[\pi\subseteq \omega(\vec{a})] \right) \left(\sum_{\vec{b}\in\mbb{Z}_d^t} Q_t(\vec{b}) \mathbbm{1}[\pi\subseteq \omega(\vec{b})] \right) \\
    &= \frac{1}{d^{-2t}} \sum_{\pi\in S_t} \left(\sum_{\vec{a}\in\mbb{Z}_d^t} Q_t(\vec{a}) \mathbbm{1}[\pi\subseteq \omega(\vec{a})] \right)^2.
\end{aligned}
\end{equation}
Note that for any given $\pi$, the term $\left(\sum_{\vec{a}\in\mbb{Z}_d^t} Q_t(\vec{a}) \mathbbm{1}[\pi\subseteq \omega(\vec{a})] \right)^2$ is always positive. Therefore, to estimate the lower bound, we can choose some of the terms in it. If $\pi=(12..t)$, the term is
\begin{equation}
\frac{1}{d^{-2t}} \left(\sum_{a=0}^{d-1} (1+d^2)^2 \right)^2  = \frac{1}{d^{-2t}}[d(1+d^2)^2]^2 \sim \mc{O}(d^{10-2t}).
\end{equation}
For $t=6$, if $\pi=(123)(456)$, the term is
\begin{equation}
\frac{1}{d^{-12}} \left( \sum_{a=0}^{d-1} (1+d^2)^2 + \sum_{\vec{a}\in Z_d^2, a_1\neq a_2} (1+d^2)^2 \right)^2 = \frac{1}{d^{-12}} [d(1+d^2)^2) + d(d-1)(1+d^2)^2]^2 \sim \mc{O}(1).
\end{equation}
Combine with Proposition~\ref{prop:Q3456orderAB}, we finish the proof.
\end{proof}

\begin{prop} \label{Prop:3varianceExactBoundAB}
For observable $O_+ \in \mc{L}((\mc{H}^A)^{\otimes 3})$ with the form $O_+ = \sum_{\vec{a}\in \mbb{Z}_d^3} [1 + (-d)^{wt(\vec{a})-1}] \ketbra{\vec{a}}{\vec{a}}$, when the random unitaries are chosen within unitary $3$-design, we have
\begin{equation}
\begin{aligned}
\tr[ \Phi^3_A(O_+^2)\otimes &\Phi^3_B(O_+^2) \rho^{\otimes 3}_{AB} ] \\
=\frac{1}{(d+2)^2} &\Big\{ \left[\tr(\rho^2_A) + \tr(\rho^2_B)\right] [3d(d-1)^2(d+1)(d^2+3d+4)] \\
&\quad + \left[\tr(\rho_{A}^3)+\tr(\rho_{B}^3)\right][2(d-1) (6 + (d-1) d^2) (d^2+3d+4)] \\
&\quad + \tr(\rho_{AB}^2)[3 d^2 (d^2-1)^2] + \tr(\rho_{AB}^3)[2 (6 + (d-1) d^2)^2] \\
&\quad + \tr(\rho_{AB}\rho_{A}\otimes\rho_{B})[6 d^2 (d^2-1)^2] \\
&\quad + [\tr(\rho_{AB}^2\rho_{A}) +\tr(\rho_{AB}^2\rho_{B})] [6d (d-1)(d+1)((d-1)d^2 + 6)] \\
&\quad + 2\tr[(\rho_{AB}^{T_A})^3][(d-1)d^2 + 6)^2] + (d(d+1)^2 - 4)^2 \Big\}
\end{aligned}
\end{equation}
for all states $\rho\in\mc{D}(\mc{H}^{AB})$.
\end{prop}

\begin{proof}
By applying the Weingarten integral, we have
\begin{equation} 
\begin{aligned}
    \tr[\Phi^3_A(O_+^2)\otimes \Phi^3_B(O_+^2) \rho^{\otimes 3}_{AB}] &= \sum_{\pi,\pi',\sigma,\sigma'\in S_t} C_{\pi,\sigma}C_{\pi',\sigma'} \tr[W_\pi^A Q_A]\tr[W_\pi^B Q_B]\tr[W_\sigma^A\otimes W_{\sigma'}^B  \rho_{AB}^{\otimes t}]
\end{aligned}
\end{equation}
where $C_{\pi,\sigma},C_{\pi',\sigma'}$ is the Weingarten matrix of $S_3$ group. With a direct calculation, we finish the proof.
\end{proof}

\section{Detailed numerical results} \label{Sec:suppNumer}

In this section, we show the detailed numerical results of the statistical error. In the main text, the statistical error of the negativity-moment $\tr(\rho_{AB}^{T_B3})$ has been presented, which is evaluated using the unbiased estimator $\hat{M}_{neg}$.  As constructed in Section~\ref{Sec:suppStat}, the estimator $\hat{M}_{neg}$ is composed of two independent estimators $\hat{M}_{neg}=\hat{M}^{AB}_{++}-\hat{M}^{AB}_{+}$, with expectation values being $\tr(\rho_{AB}^3)+\tr(\rho_{AB}^{T_B3})$ and $\tr(\rho_{AB}^{T_B3})$, respectively. Here, we show the statistical errors for both of them with finite $N_U$ and $N_M$.

The prepared state is set as the mixture of the Bell state $\Psi_{+}$ and the white noise, $\rho_{AB}=(1-p)\Psi_{+}+p\mathbb{I}/D$, which mimics a common experimental preparation. For given $N_U$ and $N_M$, and other related parameters, such as $p$ and the dimension $D$, we run the estimation scheme for $N_{av}=100$ times, and get the average error.
We also consider the effect of the properties of the state on the error, such as the mixedness and the entanglement. The simulation is based on the Matlab package \cite{TOTH2008MATLAB}.

\subsection{Statistical error of global 3-order purity term}

\begin{figure}[tbh!]
\centering
\includegraphics[width=0.7\textwidth]{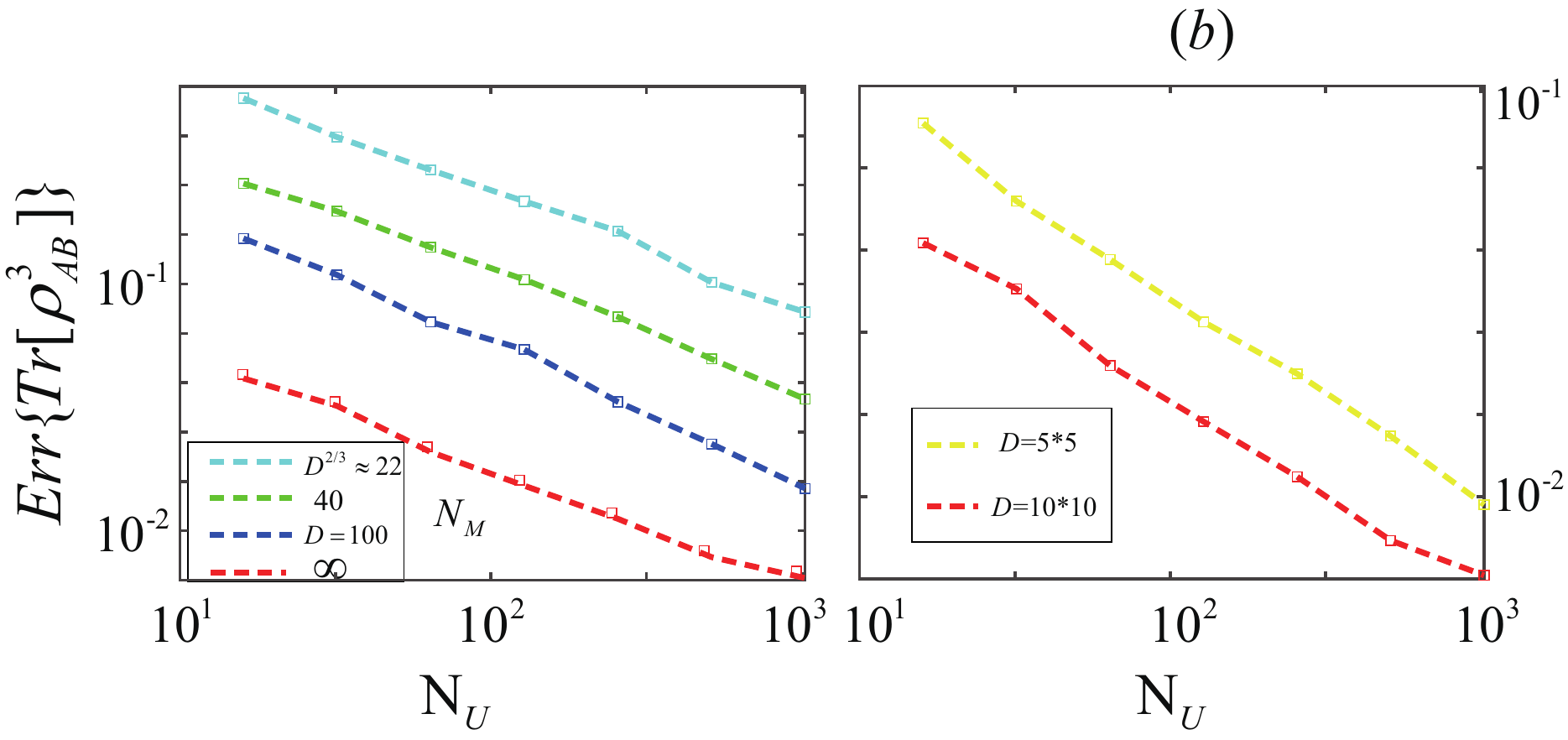}
\caption{Scaling of statistical errors of the estimator $\hat{M}^{AB}_{+}$. (a) Average statistical error of the estimated 3-order purity $\tr(\rho_{AB}^3) $ as a function of $N_U$ for various
$N_M$ with D = 10*10. (b) for D=5*5 and 10*10, with $N_M=\infty$. The unitaries
are sampled from the Haar measure numerically, and the prepared state is Bell state mixed with white noise $p=0.3$, i.e., $\rho_{AB}=(1-p)\Psi_{+}+p\mathbb{I}/D$.
}\label{Fig:3PurALL}
\end{figure}

From Fig.~\ref{Fig:3PurALL}, one can see that for different values of $N_M$, the error always decreases with slope $-0.5$ versus $N_U$ in the Log-Log plot; and the error decreases as the increase of the dimension $D$, which are both described by our analytical result in Proposition~\ref{prop:variance3orderpurity} in Section~\ref{Sec:suppStat}.

Since we adopt global unitary $U_{AB}$ twirling in the evaluation of the $\tr(\rho_{AB}^{3})$, the entanglement of $\rho_{AB}$ does not affect the statistical error, but the purity does. For instance, the pure product state $\ket{\phi_A}\ket{\phi_B}$ share the same error with the Bell state. In Fig.~\ref{Fig:3PurMix}, we plot the error for different white noise level described by the parameter $p$, for given $N_U$ and $N_M$. One can see that the larger the mixedness is, the smaller the error is.
\begin{figure}[tbh!]
\centering
\includegraphics[width=0.4\textwidth]{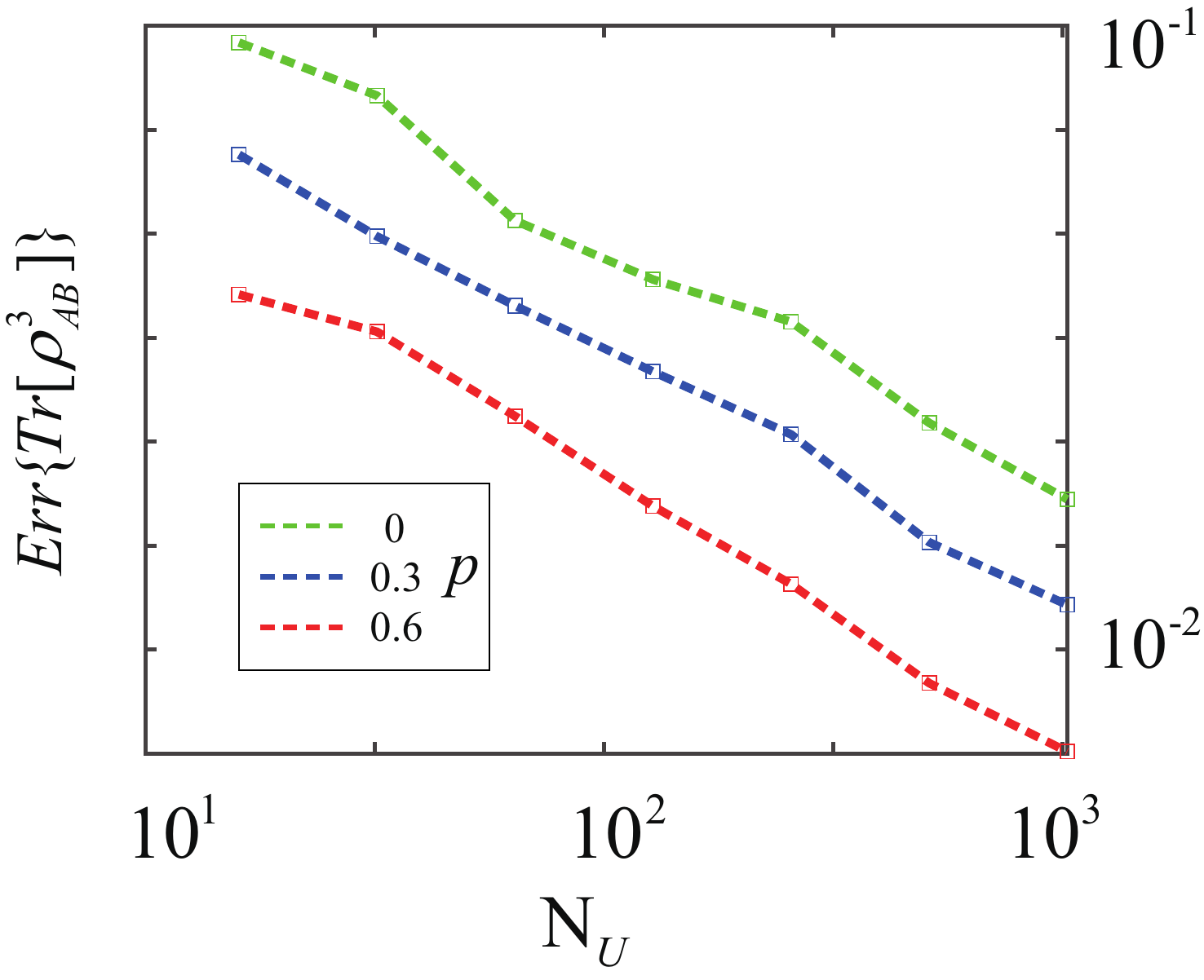}
\caption{The effect of the mixedness on the statistical errors of the estimator $\hat{M}^{AB}_{+}$. (a) Average statistical error of the estimated 3-order purity $\tr(\rho_{AB})^3 $ as a function of $N_U$ for the noisy parameter
$p=0,0.3,0.6$ with D = 10*10. The unitaries
are sampled from the Haar measure numerically, and the prepared state is Bell state mixed with white noise $p=0.3$, i.e., $\rho_{AB}=(1-p)\Psi_{+}+p\mathbb{I}/D$.
}\label{Fig:3PurMix}
\end{figure}

\subsection{Statistical error of negativity + purity term}

Similar to the 3-order purity case, from Fig.~\ref{Fig:NegALL} one can see that for different values of $N_M$, the error always decreases with slope $-0.5$ versus $N_U$ in the Log-Log plot; and the error decreases as the increase of the dimension $D$, which are both described by our analytical result in Proposition~\ref{Prop:varianceOppAB} in Section~\ref{Sec:suppStat}.

Here we adopt the bi-local unitary $U_{A}\otimes U_{B}$ twirling in the evaluation of the $\tr(\rho_{AB}^3)+\tr(\rho_{AB}^{T_B3})$, thus not only the mixedness but also the  entanglement of $\rho_{AB}$ affect the statistical error. In Fig.~\ref{Fig:NegMixEnt}, we plot the error for the pure product state $\ket{\phi_A}\ket{\phi_B}$, the Bell state $\ket{\Psi_{+}}$, and the Bell state mixed with white noise $\rho_{AB}=(1-p)\Psi_{+}+p\mathbb{I}/D$ and $p=0.3$.  We can see that as the increase of the entanglement, i.e., from $\ket{\phi_A}\ket{\phi_B}$ to $\ket{\Psi_{+}}$, the error decreases; when adding the noise and making the state more mixed, the error also decreases.

\begin{figure}[tbh!]
\centering
\includegraphics[width=0.7\textwidth]{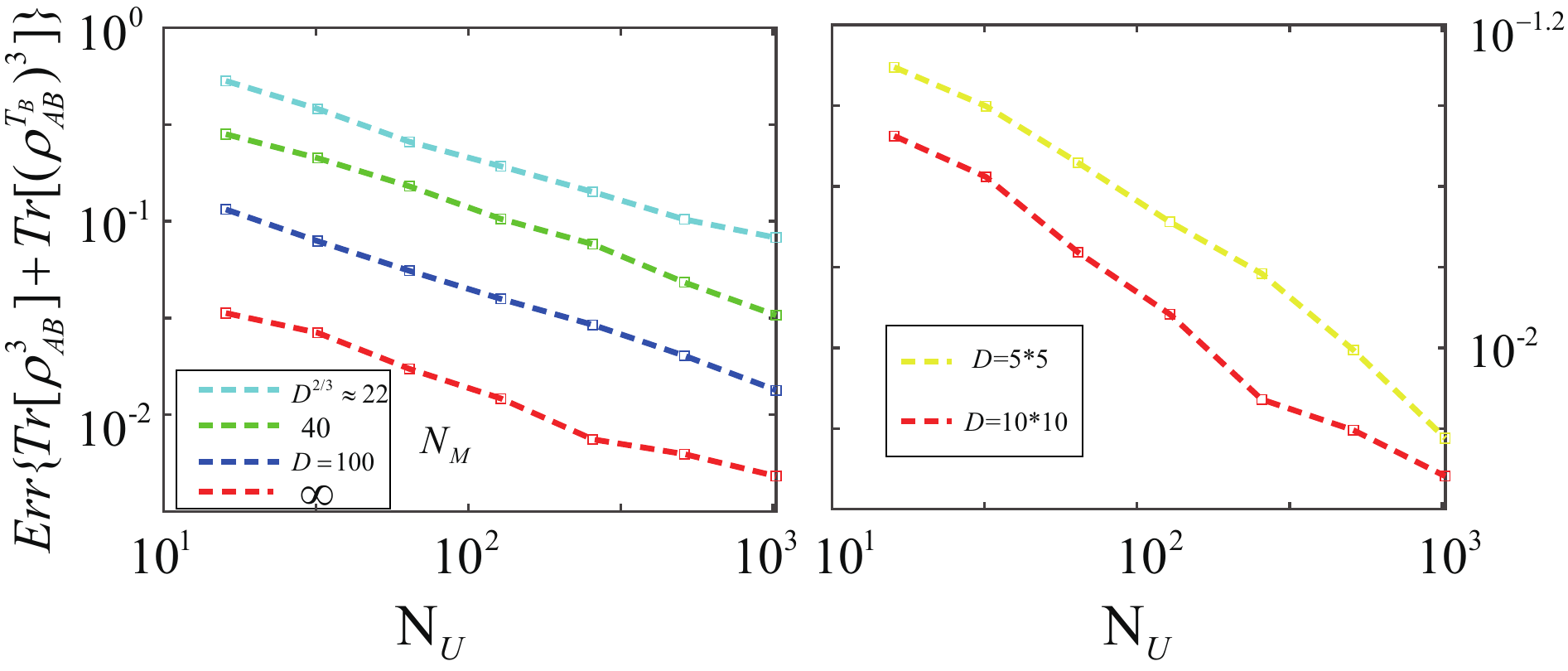}
\caption{Scaling of statistical errors of the estimator $\hat{M}^{AB}_{++}$. (a) Average statistical error of the estimated quantity $\tr(\rho_{AB}^3)+\tr(\rho_{AB}^{T_B3})$ as a function of $N_U$ for various
$N_M$ with D = 10*10. (b) for D=5*5 and 10*10, with $N_M=\infty$. The unitaries
are sampled from the Haar measure numerically, and the prepared state is Bell state mixed with white noise $p=0.3$, i.e., $\rho_{AB}=(1-p)\Psi_{+}+p\mathbb{I}/D$.
}\label{Fig:NegALL}
\end{figure}
\begin{figure}[tbh!]
\centering
\includegraphics[width=0.4\textwidth]{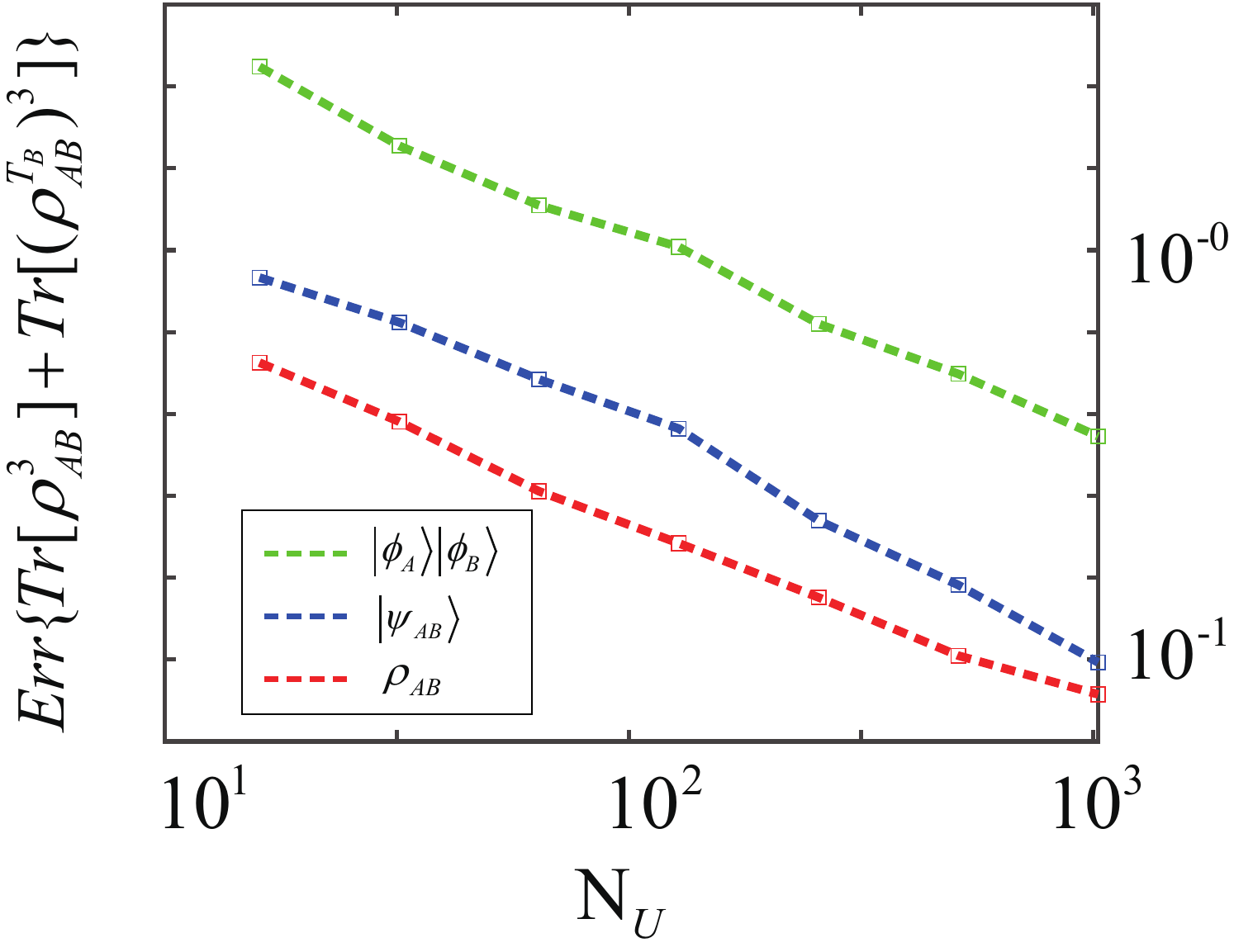}
\caption{The effect of the mixedness and entanglement on the statistical errors of the estimator $\hat{M}^{AB}_{++}$ in $D = 10*10$ system. Average statistical error of the estimated quantity $\tr(\rho_{AB}^3)+\tr(\rho_{AB}^{T_B3})$ as a function of $N_U$ for various
states, the pure product state $\ket{\phi_A}\ket{\phi_B}$, the Bell state $\ket{\Psi_{+}}$ and $\rho_{AB}=(1-p)\Psi_{+}+p\mathbb{I}/D$ with $p=0.3$. The unitaries
are sampled from the Haar measure numerically, and the prepared state is Bell state mixed with white noise.
}\label{Fig:NegMixEnt}
\end{figure}

\end{appendix}

\end{document}